\documentclass[11pt,twoside]{amsart}

\usepackage[english]{babel}
\usepackage{csquotes}
\usepackage{amsfonts}
\usepackage{hyperref}
\usepackage{etoolbox}
\usepackage{amsmath}
\usepackage{amssymb}
\usepackage{amsthm}
\usepackage{stmaryrd}
\usepackage{dsfont}
\usepackage{pifont}
\usepackage{caption}
\usepackage{subcaption}
\usepackage{mathrsfs}
\usepackage{array,multirow,graphicx}
\usepackage{enumerate}
\usepackage[all]{xy}
\usepackage{enumitem}
\usepackage{geometry}
\usepackage{amsfonts}
\usepackage{fancyhdr}
\usepackage{calc}
\usepackage{cancel}
\usepackage{accents}
\usepackage{abraces}
\usepackage[new]{old-arrows}
\usepackage{tikz,tikz-3dplot}
\usepackage{lscape}
\usepackage{algorithm}
\usepackage{algpseudocode}
\usepackage{tabularx}
\usepackage{slashbox}
\usepackage{float}
\usetikzlibrary{decorations.markings}
\usetikzlibrary{shapes,arrows,plotmarks,angles,quotes}
\usepackage{orcidlink}
\usepackage{stackrel}
\usepackage{changepage}
\usepackage[backend=biber,style=alphabetic,sorting=nty]{biblatex}

\usepackage{color}
\definecolor{linkColor}{rgb}{0.0,0.0,0.554}
\definecolor{citeColor}{rgb}{0.0,0.0,0.554}
\definecolor{fileColor}{rgb}{0.0,0.0,0.554}
\definecolor{urlColor}{rgb}{0.0,0.0,0.554}
\definecolor{promptColor}{rgb}{0.0,0.0,0.589}
\definecolor{brkpromptColor}{rgb}{0.589,0.0,0.0}
\definecolor{gapinputColor}{rgb}{0.589,0.0,0.0}
\definecolor{gapoutputColor}{rgb}{0.0,0.0,0.0}
\usepackage{fancyvrb}
\usepackage{adjustbox}
\usepackage{helvet}

\definecolor{cof}{RGB}{219,144,71}
\definecolor{pur}{RGB}{186,146,162}
\definecolor{greeo}{RGB}{91,173,69}
\definecolor{greet}{RGB}{52,111,72}

\tdplotsetmaincoords{70}{165}

\newcommand{\changefont}{%
    \fontsize{8}{8}\selectfont
}
\mathchardef\mhyphen="2D 

\makeatletter
\def\BState{\State\hskip-\ALG@thistlm}
\makeatother

\title[Particle dynamics in spherically symmetric electro-vacuum instantons]{Particle dynamics in spherically symmetric electro-vacuum instantons}

\author{Arthur Garnier \orcidlink{0000-0003-4069-3203}}
\address{\newline
Universit\'e de Picardie,
\newline LAMFA (UMR 7352 du CNRS),
\newline 33 rue St Leu,
\newline F-80039 Amiens Cedex 1,
\newline France}
\email{arthur.garnier@math.cnrs.fr}

\theoremstyle{plain}

\newtheorem{prop-def}[prop]{Proposition-Definition}

\newtheorem*{prop*}{Proposition}
\newtheorem*{prop-def*}{Proposition-Definition}
\newtheorem*{propri*}{Property}
\newtheorem*{lem*}{Lemma}
\newtheorem*{theo*}{Theorem}
\newtheorem*{cor*}{Corollary}
\newtheorem*{rem*}{Remark}
\newtheorem*{definition*}{Definition}
\newtheorem*{exemple*}{Example}
\newtheorem*{notation*}{Notation}

\newcommand{\lra}{\longrightarrow}

\newcommand{\ra}{\rightarrow}
\newcommand{\sdp}{\times\kern-.2em\vrule height1.1ex depth-.05ex}
\newcommand{\epi}{\lra \kern-.8em\ra}
\newcommand{\C}{{\mathbb C}}

\newcommand{\R}{{\mathbb R}}

\newcommand{\longto}{\longrightarrow}

\newcommand{\Sph}{\mathbb{S}}

\newcommand{\Pro}{\mathbb{P}}

\DeclareMathOperator\cotan{cotan}
\DeclareMathOperator\bigo{\mathbf{O}}

 \normalbaselines
\makeatletter
\newlength\@SizeOfCirc%
\newcommand{\CircleArrowRight}[1]{%
    \setlength{\@SizeOfCirc}{\maxof{\widthof{#1}}{\heightof{#1}}}%
    \tikz [x=1.0ex,y=1.0ex,line width=.12ex]%
        \draw [->,anchor=center]%
            node (0,0) {#1}%
            (0,0.8\@SizeOfCirc) arc (85:-240:0.8\@SizeOfCirc);%
}%
\newcommand{\CircleArrowLeft}[1]{%
    \setlength{\@SizeOfCirc}{\maxof{\widthof{#1}}{\heightof{#1}}}%
    \tikz [x=1.0ex,y=1.0ex,line width=.12ex]%
        \draw [<-,anchor=center]%
            node (0,0) {#1}%
            (0,0.8\@SizeOfCirc) arc (85:-240:0.8\@SizeOfCirc);%
}%
\makeatother
\tikzset{
    set arrow inside/.code={\pgfqkeys{/tikz/arrow inside}{#1}},
    set arrow inside={end/.initial=>, opt/.initial=},
    /pgf/decoration/Mark/.style={
        mark/.expanded=at position #1 with
        {
            \noexpand\arrow[\pgfkeysvalueof{/tikz/arrow inside/opt}]{\pgfkeysvalueof{/tikz/arrow inside/end}}
        }
    },
    arrow inside/.style 2 args={
        set arrow inside={#1},
        postaction={
            decorate,decoration={
                markings,Mark/.list={#2}
            }
        }
    },
}


\date{\today}

\begin{document}

\begin{abstract}
In this paper, we study the geodesic motion in spherically symmetric electro-vacuum Euclidean solutions of the Einstein equation. There are two kinds of such solutions: the Euclidean Reissner--Nordstr\"{o}m (ERN) metrics, and the Bertotti--Robinson-like (BR) metrics, the latter having constant Kretschmann scalar.

First, we derive the motion equations for the ERN spacetime and we generalize the results of Battista--Esposito, showing that all orbits in as ERN spacetime are unbounded if and only if it has an event horizon. We also obtain the Weierstrass form of the polar radial motion, providing an efficient tool for numerical computations.

We then study the angular deflection of orbits in the Euclidean Schwarzschild spacetime which, in contrast to the Lorentzian background, can be either positive or negative. We observe the presence of a null and a maximal deflection rings for particles with velocity at infinity $v>1$ and we give approximate values for their size when $v\gtrsim1$.

For BR spacetimes, we obtain analytic solutions for the radial motion in proper length, involving (hyperbolic) trigonometric functions and we deduce that orbits either exponentially go to the singularity or are periodic.

Finally, we apply the previous results and use algorithms related to Weierstrass' elliptic functions to produce a Python code to plot orbits of the spacetimes ERN and BR, and draw ``shadows'' of the first ones, as it was already done before for classical black holes.
\end{abstract}

\maketitle


\section*{Introduction and motivation}

\indent Instantons (or pseudoparticles) were originally defined in \cite{belavin-et-al75} as solutions of the (classical) Yang--Mills field equations, which are non-singular on some section of a complexified spacetime. By analogy, a \textit{gravitational instanton} was defined in \cite{hawking77} to be a solution of the classical Einstein field equation, which is positive-definite (i.e. Riemannian) on some section of a complexified spacetime. Such metrics were first introduced in quantum gravity by Hartle and Hawking \cite{hartle-hawking} in order to make some path integral converge, hence defining the so-called \textit{Hartle--Hawking propagator}. Quoting \cite{gibbons-hawking79}, after the development of instantons in Yang--Mills theory and because (super)gravity is a gauge theory, it seems reasonable to expect gravitational instantons to play a similar role in gravity as instantons do in quantum field theory. For general discussions on gravitational instantons, see \cite{hawking1979,gibbons_centenary,eguchi-gilkey-hanson80,elster84,esposito}.

Since their introduction, gravitational instantons and their interactions with gauge instantons have been a subject of deep interest \cite{tekin02,elnaschie04,mosna-tavares,oh-park-yang11}. More recently, the existence and uniqueness of \textit{toric instantons} have been established in \cite{kunduri-lucietti22}. Moreover, \textit{purely Euclidean instantons} (i.e. the corresponding complex spacetime does not admit any Lorentzian section) have been introduced in \cite{chen-teo11} and thoroughly described in \cite{aksteiner-andersson21}. It is worth mentioning that the gravitational instantons with \textit{positive cosmological constant} were fully described and classified in \cite{page78}.

Besides quantum gravity, in the early \textit{geometric models of matter} introduced in \cite{atiyah-manton-schroers12}, the Fubini--Study metric on the projective plane $\C\Pro^2$ has been proposed as a (compact) model for the spacetime surrounding a \textit{neutron}. Later in \cite{atiyah-franchetti-schroers14}, the authors rather propose the Euclidean Schwarzschild geometry as a model for the neutron. As explained in \cite{jante15}, this has been generalized to other spin-$\tfrac12$-particles such as the proton and the electron, for which the Taub--bolt and Taub--NUT instantons were respectively given as candidates. Moreover, interesting uniqueness results on Euclidean Schwarzschild and Taub--NUT instantons were obtained in \cite{mars99}. These proposals further motivate the investigation of gravitational instantons and, in particular, the study of the geodesic dynamics in such spaces.

Geodesic motion in gravitational instantons has started more than thirty years ago, with the pioneer work \cite{anderson90}, focusing on closed geodesics in compact instantons, with applications in the determination of their injectivity radius. The general geodesic dynamics in (generalized) Taub--NUT instantons has been detailed in \cite{visinescu93}. More recently, the case of Kerr--Newman instantons is the topic of \cite{lindberg-rayan18}, while instantons of Eguchi--Hanson type were studied in \cite{yang-zhang23}. Finally, we mention that the integrability of the conformal geodesic flow on spherically symmetric instantons motivates the work \cite{dunajski-tod22}.

In the present paper, an \textit{instanton} will designate a Riemannian solution of the Maxwell--Einstein equations on a 4-dimensional manifold\footnote{In contrast to \cite{hawking77}, we do not assume that the curvature vanishes at large distances.}. We investigate the geodesic motion in such spaces, which we assume to be \textit{spherically symmetric}. In particular, we will apply the theory of Weierstrass elliptic functions to the radial motion, as it was already done in the Lorentzian framework; see \cite{gibbons-vyska,cieslik-mach22}. We also study the gravitational lensing of trajectories, as done for photons and massive particles in Reissner--Nordstr\"{o}m (resp. Kerr--Newman) spacetimes in \cite{pang-jia} (resp. in \cite{he-lin16}). Our methods can also be compared to the more recent work \cite{viththani24}, where tidal forces are also investigated. As explained below, one of the main aims of the present paper is to highlight some important dynamical differences between the Euclidean and the Lorentzian backgrounds.

The simplest example of a gravitational instanton is the Euclidean Schwarzschild metric \cite{hartle-hawking}, given in Schwarzchild coordinates $(\tau,r,\theta,\phi)$ on $\mathcal{M}:=\R\times]2M;+\infty[\times\Sph^2\simeq\R^2\times\Sph^2$ by
\begin{equation}\label{ES}\tag{ES}
{\rm d}s^2=\left(1-\frac{2M}{r}\right){\rm d}\tau^2+\left(1-\frac{2M}{r}\right)^{-1}{\rm d}r^2+r^2({\rm d}\theta^2+\sin^2\theta{\rm d}\phi^2),
\end{equation}
($M\ge0$ being the mass of the central body) and the link with the Lorentzian Schwarzschild metric is given by setting the Euclidean time $\tau=it$, with $t$ being the (Lorentzian) coordinate time. In other words, the Euclidean Schwarzschild metric is obtained from the Lorentzian one by applying a \textit{Wick rotation}. As for any metric, it is natural to study the geodesic motion associated to it. Regarding the metric \eqref{ES} above, this question was addressed thoroughly in \cite{battista-esposito22}.

As suggested in \cite[\S II]{gibbons-hawking77}, one can also look for the Reissner--Nordstr\"{o}m analogue of the Euclidean Schwarzschild solution; a metric that was used in \cite[\S 4]{mellor-moss89} and \cite[\S II.B]{monteiro_santos}, for instance. Assuming the central body has an electric charge $Q\in\R$, the Euclidean Reissner--Nordstr\"{o}m (ERN) metric is given by
\begin{equation}\label{ERN}\tag{ERN}
{\rm d}s^2=\left(1-\frac{2M}{r}+\frac{Q^2}{r^2}\right){\rm d}\tau^2+\left(1-\frac{2M}{r}+\frac{Q^2}{r^2}\right)^{-1}{\rm d}r^2+r^2({\rm d}\theta^2+\sin^2\theta{\rm d}\phi^2).
\end{equation}

In this paper, we generalize the approach of \cite{battista-esposito22} to this metric and study the motion of a test particle in the ERN spacetime. Specifically, we prove that the method of \cite[\S 3.1]{gibbons-vyska} still applies to the ERN metric, hence obtaining analytic solutions for (non-purely radial) geodesics in terms of Weierstrass' elliptic functions. As we will see, one of the remarkable results obtained in \cite{battista-esposito22} extends to the ERN spacetime with horizon (i.e. such that $Q^2\le M^2$), namely the fact that the energy of an exterior geodesic is confined in the open interval $]-1,1[$. In particular, no elliptic-like geodesics exist in the sub-extremal case. This fails in the super-charged case $Q^2>M^2$, where arbitrary high energy is allowed and attained by a circular geodesic. As mentioned in \cite{battista-esposito22}, these facts show that these Riemannian solutions present substantial differences in their dynamics, when compared to their usual relativistic avatars.

Moreover, we will see how the polar motion equation simplifies in the Schwarzschild case $Q=0$ and we retrieve the results from \cite{battista-esposito22} using only the elementary geometric properties of the real elliptic curve describing the phase portrait in (affinely transformed) Binet variable.

Another remarkable dynamical distinction between Lorentzian and Euclidean Schwarzschild spacetimes relies in the gravitational deflection of orbits. As is well-known, given an orbit coming from and to infinity, the deflection angle $\delta\phi$ (in the motion plane) between its two asymptotic directions is always positive in Schwarzschild geometry, for photons as well as for massive particles. This means that test-particles can only be attracted by the central body. However, in Euclidean Schwarzschild geometry, the deflection angle can be positive or negative, depending on the velocity at infinity $v$ and the perihelion $r_{\rm min}$ of the orbit; in this geometry, particles can be attracted or repelled by the central mass. More precisely, we observe that, at fixed $v<1$, we have $\delta\phi<0$ for all values of $r_{\rm min}$, while for $v>1$, the deflection $\delta\phi$ vanishes (resp. is maximal positive) at some perihelion $r_{\rm min}=\rho_0$ (resp. $r_{\rm min}=\rho_{\rm max}$). In \eqref{Rm0} and \eqref{Rm1}, we give approximate values for $\rho_0$ and $\rho_{\rm max}$ when $v\gtrsim1$. The existence of $\rho_{\rm max}$ gives rise to a visible maximal deflection ring in the shadow of such a spacetime. As already said, the fact the $\delta\phi<0$ at small perihelia indicates that particles are repelled by the central mass and thus the horizon becomes invisible to the observer, see Figures \ref{repulsive} and \ref{shadowsERN}.

The only other possible type of spherically symmetric instanton is given by a \textit{Bertotti--Robinson-like} metric, whose line element has the form
\begin{equation}\label{BRr}\tag{BR}
{\rm d}s^2=Q^2\left(\frac{1-2mr+q^2r^2}{r^2}{\rm d}\tau^2+\frac{{\rm d}r^2}{r^2(1-2mr+q^2r^2)}+{\rm d}\theta^2+\sin^2\theta{\rm d}\phi^2\right),
\end{equation}
where $Q\ne0$ and $m,q\in\R$ are some constants. This is the Euclidean analogue of the general Bertotti--Robinson electro-vacuum Lorentzian solution. As in the Lorentzian case, this metric essentially differs from \eqref{ERN} is this sense that its Kretschmann invariant is constant. To the knowledge of the author, the dynamics of this solution, Euclidean or Lorentzian, doesn't appear in the literature. This is treated in \S \ref{dynamics_BR}, where we provide a full analytic solution of the geodesic equation, in terms of (hyperbolic) trigonometric functions; see \eqref{analytic_BR} and \eqref{BR_affine}.

Finally, the efficient algorithms available to approximate Weierstrass' elliptic functions \cite{carlson,coquereaux} are used to produce a Python code\footnote{available at \url{https://github.com/arthur-garnier/euclidean_orbits_and_shadows.git}}, designed to draw orbits in the two types of instantons discussed here, as well as to obtain the ``shadow'' of an ERN space by ray-tracing, as it was already done for black holes in \cite{agol-dexter,gyoto,GRay,PYYY,osiris}, for instance. Since there are no null geodesics in Euclidean geometry, photons are replaced by particles with a velocity at infinity that should be provided by the user. Conformally to what was mentioned above concerning the deflection angle, we observe the presence of a maximal deflection ring when $v>1$ and we notice that the horizon is not visible. In particular, the optical difference between the cases $Q^2\lesssim M^2$ and $Q^2\gtrsim M^2$ is not as obvious as in the Lorentzian background. We still observe that the size of the maximal deflection ring diminishes as the charge increases, see Figure \ref{shadows_nice}.

The layout of the paper is as follows: first, we state that the metrics \eqref{ERN} and \eqref{BRr} are the only spherically symmetric solutions of the Einstein--Maxwell field equation with complex vector potential $A_\mu=-iQr^{-1}{\rm d}\tau$. The detailed proof of this result can be found in the Appendix \ref{proof_unicity}. Then, we derive the motion equations and the motion constants for the metric \eqref{ERN}, and we prove that the energy $E$ of a geodesic satisfies $E^2<1$ when $Q^2<M^2$, as mentioned above. We then obtain the Weierstrass equation from of the polar radial motion equation and we investigate the particular case where $Q=0$. Next, we study the gravitational deflection in Euclidean Schwarzschild spacetime and provide the aforementioned approximations for the deflection angles, as well as for the null and maximal deflection rings. Concerning the Bertotti--Robinson family \eqref{BRr}, we derive the motion equations and obtain analytic solutions with (hyperbolic) trigonometric functions. Finally, we quickly explain how the Python code is constructed and we finish with some figures illustrating our results and programs.

\section{The two types of spherically symmetric electro-vacuum instantons}\label{unicity_section}

In this section, we state the unicity result for spherically symmetric electro-vacuum instantons. We systematically use Stoney units where $G=c=4\pi\epsilon_0=1$. Let $(M,Q)\in\R_+\times\R$ and consider the numbers $r_+,r_-$ defined by 
\[r_{\pm}:=\left\{\begin{array}{ll}M\pm\sqrt{M^2-Q^2} & \text{if }Q^2<M^2, \\ 0 & \text{otherwise.}\end{array}\right.\]
Let also $\mathcal{M}:=\R\times ]r_+,+\infty[\times\Sph^2\simeq\R^2\times\Sph^2$, with coordinates $x^\mu=(\tau,r,\theta,\phi)$, the pair $(\theta,\phi)$ describing spherical coordinates on $\Sph^2$. Since we work on the Euclidean section, we restrict our study to the open subset $r>r_+$, just as in \cite{esposito}. This is a reasonable restriction, as we are interested by the exterior region of the spacetime.

We denote by ${\rm d}\Omega^2:={\rm d}\theta^2+\sin^2\theta{\rm d}\phi^2$ the usual round metric on $\Sph^2$ and we have the following result, the detailed proof of which can be found in Appendix \ref{proof_unicity}.

\begin{adjustwidth}{1cm}{1cm}
\emph{Let ${\rm d}s^2=g_{\mu\nu}{\rm d}x^\mu{\rm d}x^\nu$ be a spherically symmetric solution of the Einstein--Maxwell equation with complex vector potential
\[A_\mu=-iQr^{-1}{\rm d}\tau,\]
defined for $r\gg0$.}

\emph{If the Kretschmann invariant $K=R^{\alpha\beta\mu\nu}R_{\alpha\beta\mu\nu}$ associated to $g_{\mu\nu}$ is independent of $r$, then there are constants $m,q\in\R$ such that the metric takes the Bertotti--Robinson form
\[{\rm d}s^2=Q^2\left[\frac{1-2mr+q^2r^2}{r^2}{\rm d}\tau^2+\frac{{\rm d}r^2}{r^2(1-2mr+q^2r^2)}+{\rm d}\Omega^2\right]\]
in which case the Kretschmann invariant is $K=8Q^{-4}$.}

\emph{Otherwise, there are coordinate transformations of the form $\tilde{\tau}=C\tau$ and $R=r/(Ar+B)$ (with $A\in\R$ and $B,C\in\R^*$) as well as constants $\tilde{M},\tilde{Q}\in\R$ such that $g_{\mu\nu}$ takes the Reissner--Nordstr\"{o}m form
\[{\rm d}s^2=\left(1-\frac{2\tilde{M}}{R}+\frac{\tilde{Q}^2}{R^2}\right){\rm d}{\tilde{\tau}}^2+\left(1-\frac{2\tilde{M}}{R}+\frac{\tilde{Q}^2}{R^2}\right)^{-1}{\rm d}R^2+R^2{\rm d}\Omega^2,\]
whose Kretschmann invariant is $K=8R^{-8}(6\tilde{M}^2R^2-12\tilde{M}\tilde{Q}^2R+7\tilde{Q}^4)$.}
\end{adjustwidth}

We make the following observations:
\begin{itemize}
\item The Ricci scalar of any of the above solutions vanishes.
\item The proof shows in particular that for a vector potential $A_\tau=\mathcal{Q}r^{-1}$, a spherically symmetric solution of the field equation is Euclidean (resp. Lorentzian) if and only if $\mathcal{Q}$ is purely imaginary (resp. is real).
\item Recalling the notation from the proof, we observe that in the Reissner--Nordstr\"{o}m case, the new potential is
\[A'_\mu=-i\tilde{Q}R^{-1}{\rm d}\tilde{\tau}=-iQ(r^{-1}+\alpha\beta^{-1}){\rm d}\tau=A_\mu-i\nabla_\mu f,\]
where $f:=Q\alpha\beta^{-1}\tau$. Therefore, the coordinate transformation $(\tau,r)\mapsto(\tilde{\tau},R)$ induces a gauge transformation $A_\mu\mapsto A_\mu-i\nabla_\mu f$.
\item The second metric of the statement with $m=q=0$ gives the Euclidean version of the original Bertotti--Robinson line element derived in \cite{robinson59,bertotti59}
\[{\rm d}s^2=\frac{Q^2}{r^2}\left[{\rm d}\tau^2+{\rm d}r^2+r^2{\rm d}\Omega^2\right].\]
Observe moreover that in Binet variable $u=1/r$, the general Bertotti--Robinson metric has an even simpler form
\[{\rm d}s^2=Q^2\left[(u^2-2mu+q^2){\rm d}\tau^2+\frac{{\rm d}u^2}{u^2-2mu+q^2}+{\rm d}\Omega^2\right].\]
\end{itemize}

If now ${\rm d}s^2=g_{\mu\nu}{\rm d}x^\mu{\rm d}x^\nu=g_{\mu\mu}({\rm d}x^\mu)^2$ is an \textit{asymptotically flat} spherically symmetric electro-vacuum instanton, then the Kretschmann scalar should vanish as $r\to+\infty$, thus only the Reissner--Nordstr\"{o}m form from the previous theorem is allowed, with $\tilde{\tau}=C\tau$ and $R=r/(Ar+B)$. But the asymptotic conditions
\[\lim_{r\to+\infty}(Ar+B)^{-2}=\lim_{r\to +\infty}\frac{g_{\theta\theta}}{r^2}=1=\lim_{r\to+\infty}g_{\tau\tau}=\lim_{r\to+\infty}C^2(1-2\tilde{M}/R+\tilde{Q}^2/R^2)\]
impose $R=r$ and $\tilde{\tau}=\tau$. Finally, the electromagnetic tensor $F_{\mu\nu}$ has $iF_{\tau r}=\tilde{Q}R^{-2}=Qr^{-2}$, so that $\tilde{Q}=Q$ and we have obtained the following Euclidean analogue of the Birkhoff--Hoffmann theorem:
\begin{adjustwidth}{1cm}{1cm}
\emph{The Euclidean Reissner--Nordstr\"{o}m metric is the only spherically symmetric, asymptotically (Euclidean) flat metric satisfying the electro-vacuum Einstein--Maxwell equation associated to the complex vector potential
\[A_\mu:=-iQr^{-1}{\rm d}\tau.\]
More precisely, if ${\rm d}s^2=g_{\mu\nu}{\rm d}x^\mu{\rm d}x^\nu$ is a such a metric, defined for $r\gg0$, then there exists a constant $\tilde{M}\in\R$ such that
\[{\rm d}s^2=\left(1-\frac{2\tilde{M}}{r}+\frac{Q^2}{r^2}\right){\rm d}\tau^2+\left(1-\frac{2\tilde{M}}{r}+\frac{Q^2}{r^2}\right)^{-1}{\rm d}r^2+r^2{\rm d}\theta^2+r^2\sin^2\theta{\rm d}\phi^2.\]
In particular, the metric \eqref{ERN} is the only spherically symmetric solution of the Einstein--Maxwell equation associated to $A_\mu$ defined on $\mathcal{M}$, which reduces to the Euclidean Schwarzschild metric \eqref{ES} when $Q\to0$.}
\end{adjustwidth}

\section{Geodesic motion in Euclidean Reissner--Nordstr\"{o}m instantons}

\subsection{Motion equations and energy of orbits}\label{energyy}
Recall the notation from the beginning of the previous section and consider a non-constant geodesic $\gamma=(\tau,r,\theta,\phi)$ in $\mathcal{M}=\R\times ]r_+,+\infty[\times\Sph^2$ for the metric \eqref{ERN}, with affine parameter $\lambda$. We will analyse the geodesic equation in the same fashion as in \cite{battista-esposito22}.

By spherical symmetry, we may assume that $\theta\equiv\pi/2$ and letting
\[\Delta(r):=1-\frac{2M}{r}+\frac{Q^2}{r^2},\]
the relativistic Lagrangian $\mathcal{L}=\tfrac12 g_{\mu\nu}\dot{\gamma}^\mu\dot{\gamma}^\nu$ reads
\[2\mathcal{L}=\Delta(r)\dot{\tau}^2+\Delta(r)^{-1}\dot{r}^2+r^2\dot{\phi}^2.\]
Thus, the temporal and angular Euler--Lagrange equations provide constants $C,J\in\R$ such that
\begin{equation}\label{tandphi}
\dot{\tau}=\frac{C}{\Delta(r)},~\dot{\phi}=\frac{J}{r^2}
\end{equation}
and thus the scalar
\[\mathcal{H}:=2\mathcal{L}=\Delta(r)^{-1}(C^2+\dot{r}^2)+\frac{J^2}{r^2}>0\]
is conserved along $\gamma$ and the proper length $s$ satisfies ${\rm d}s^2=\mathcal{H}{\rm d}\lambda^2$ so that we get 
\begin{equation}\label{mot_ERN}
\left(\frac{{\rm d}\phi}{{\rm d}s}\right)^2=\frac{\dot{\phi}^2}{\mathcal{H}}=\frac{L^2}{r^4},~\left(\frac{{\rm d}r}{{\rm d}s}\right)^2=\frac{\dot{r}^2}{\mathcal{H}}=\Delta(r)\left(1-\frac{L^2}{r^2}\right)-E^2.
\end{equation}
where $E$ (resp. $L$) is the \textit{energy per unit mass} (resp. \textit{angular momentum per unit mass}) of $\gamma$, defined by analogy with the Lorentzian framework as the proper temporal (resp. azimuthal) conjugate momentum 
\[E:=p_{\tau(s)}=\tfrac{1}{\sqrt{|\mathcal{H}|}}p_{\tau(\lambda)}=\frac{g_{\tau\mu}\dot{\gamma}^\mu}{\sqrt{\mathcal{H}}}=\frac{\Delta(r)\dot{\tau}}{\sqrt{\mathcal{H}}}=\frac{C}{\sqrt{\mathcal{H}}},~\left(\text{resp. } L:=\tfrac{1}{\sqrt{|\mathcal{H}|}}p_\phi=\frac{J}{\sqrt{\mathcal{H}}}\right)\]

For the rest of this section, we assume that $\gamma$ is non-purely radial and we denote the differentiation with respect to the Euclidean time $\tau$ with a dot.

We want to apply the Weierstrass analysis of this equation and since it has degree 4, one first needs to choose a real root of the quartic right-hand side. Such a real root is guaranteed to exist as soon as $E^2<1$, a property that we shall prove to always hold, provided that the metric \eqref{ERN} presents an event horizon (that is, when $Q^2<M^2$). To do this, we need expressions for the energy and angular momentum, as functions of the initial conditions $\gamma(0)=:(\tau_0,r_0,\pi/2,\phi_0)$ and $\dot{\gamma}(0)=:(1,\dot{r}_0,0,\dot{\phi}_0)$. By symmetry, we may assume that $\tau_0=\phi_0=0$ and if we let $\alpha:=\left.\frac{{\rm d}\tau}{{\rm d}\lambda}\right|_{\lambda=0}$, then the constant $C$ reads $C=\alpha\Delta(r_0)$ and we also have $J=\alpha r_0^2\dot{\phi}_0$. Then,
\[\mathcal{H}=\Delta(r_0)^{-1}\left(C^2+\left(\frac{{\rm d}r}{{\rm d}\lambda}\right)^2\right)+\frac{J^2}{r_0^2}=\alpha^2\left(\Delta(r_0)+\Delta(r_0)^{-1}\dot{r}_0^2+r_0^2\dot{\phi}_0^2\right),\]
so that we arrive at the following expressions for the energy and angular momentum:
\begin{equation}\label{mot_cons}
E=\frac{\Delta(r_0)}{\sqrt{\Delta(r_0)+\Delta(r_0)^{-1}\dot{r}_0^2+r_0^2\dot{\phi}_0^2}},~L=\frac{r_0^2\dot{\phi}_0}{\sqrt{\Delta(r_0)+\Delta(r_0)^{-1}\dot{r}_0^2+r_0^2\dot{\phi}_0^2}}.
\end{equation}

We can now state the main result of this section, generalizing the results from \cite{battista-esposito22} to the charged case. It implies in particular that the metric \eqref{ERN} features an event horizon exactly when there is no bounded orbit. The proof, relying on a tedious analysis of a real polynomial, is given in Appendix \ref{proof_energy}.
\begin{adjustwidth}{1cm}{1cm}\label{brun}
\emph{If $Q^2\le M^2$, then any (exterior) non-constant geodesic $\gamma$ for the metric \eqref{ERN} has $E^2<1$. Otherwise, there are circular orbits with arbitrary energy.}
\end{adjustwidth}

\subsection{Reduction of the polar radial equation to Weierstrass' form}\label{weier_form}
Let $\gamma=(\tau,r,\phi)$ be a non-purely radial equatorial geodesic. Then $J\ne0$ so that the map $s\mapsto\phi(s)$ is a diffeomorphism onto its image and from \eqref{mot_ERN} we find the polar radial equation
\begin{equation}\label{polmot_ERN}
\left(\frac{{\rm d}r}{{\rm d}\phi}\right)^2=\frac{1-E^2}{L^2}r^4-\frac{2M}{L^2}r^3+\left(\frac{Q^2}{L^2}-1\right)r^2+2Mr-Q^2=:F(r)
\end{equation}

In this section, the dot denotes differentiation with respect to the polar variable $\phi$. We use the same trick as in \cite[\S 3.1]{gibbons-vyska} to reduce the degree of the above equation and then re-write it in Weierstrass form. Let $\overline{r}\in\C$ be a root of the quartic $F$ (which, in view of the result from \S \ref{energyy}, is guaranteed to be real positive when $Q^2\le M^2$) and consider the shifted Binet variable 
\[u:=\frac{1}{r-\overline{r}}.\]
Then, the equation \eqref{polmot_ERN} becomes
\begin{align*}
\dot{u}^2=&\frac{\dot{r}^2}{(r-\overline{r})^4}=u^4F(\overline{r}+1/u) \\
=&\frac{1}{L^2}\left[2(2(1-E^2)\overline{r}^3-3M\overline{r}^2+(Q^2-L^2)\overline{r}+ML^2)u^3\right.\\
&\left.+(6(1-E^2)\overline{r}^2-6M\overline{r}+Q^2-L^2)u^2+2(2(1-E^2)\overline{r}-M)u+1-E^2\right]
\end{align*}
It is now straightforward to re-write this in Weierstrass form. Indeed, we can choose a root $\overline{r}\in\C$ of the quartic
\[\frac{1-E^2}{L^2}r^4-\frac{2M}{L^2}r^3+\left(\frac{Q^2}{L^2}-1\right)r^2+2Mr-Q^2=0,\]
which we can choose to be real positive if $Q^2\le M^2$. If we let
\begin{equation}\label{weierP1}\tag{5a}
\left\{\begin{array}{l}
\delta=\frac{1-E^2}{L^2},\\[.5em]
\gamma=2\left(2\delta\overline{r}-\frac{M}{L^2}\right),\\[.5em]
\beta=6\overline{r}\left(\delta\overline{r}-\frac{M}{L^2}\right)+\frac{Q^2}{L^2}-1,\\[.5em]
\alpha=2\left(2\delta\overline{r}^3-\frac{3M}{L^2}\overline{r}^2+\left(\frac{Q^2}{L^2}-1\right)\overline{r}+M\right),\end{array}\right.~~\text{and}~~\left\{\begin{array}{l} g_2:=\frac14\left(\frac{\beta^2}{3}-\alpha\gamma\right),\\[.5em]
g_3:=\frac18\left(\frac{\alpha\beta\gamma}{6}-\frac{\alpha^2\delta}{2}-\frac{\beta^3}{27}\right),\\[.5em]
\wp:=\frac{\alpha}{4(r-\overline{r})}+\frac{\beta}{12},\end{array}\right.
\end{equation}
then the function $\wp$ satisfies the Weierstrass equation
\begin{equation}\label{weierP2}\tag{5b}
\dot{\wp}^2=4\wp^3-g_2\wp-g_3.
\end{equation}
In other words, the polar radial motion is given by
\begin{equation}\label{weierP3}\tag{5c}
r(\phi)=\overline{r}+\frac{\alpha}{4\wp(\phi)-\beta/3},
\end{equation}
where $\wp=\wp_{g_2,g_3}$ is the Weierstrass function associated to the pair $(g_2,g_3)\in\C^2$.

More practically, given an initial radius $r_0:=r(\phi_0)$, we have to find some $z_0\in\C$ such that $\wp(z_0)=\frac{\alpha}{4(r_0-\overline{r})}+\frac{\beta}{12}$, a task that can be achieved using Carlson's integrals \cite{carlson}
\begin{equation}\label{carlsonRF}\tag{$R_F$}
R_F(x,y,z):=\frac12\int_0^\infty\frac{{\rm d}\zeta}{\sqrt{(\zeta+x)(\zeta+y)(\zeta+z)}}.
\end{equation}
Then, the equation \eqref{weierP3} can be recast in the following form
\begin{equation}\label{weierP3_bis}\tag{5c'}
r(\phi)=\overline{r}+\frac{\alpha}{4\wp_{g_2,g_3}(z_0+\phi)-\beta/3},\text{ with }z_0:=R_F(\wp_0-z_1,\wp_0-z_2,\wp_0-z_3)\in\C,
\end{equation}
where $z_{1,2,3}\in\C$ are the roots of the Weierstrass cubic $4z^3-g_2z-g_3$ and $\wp_0:=\frac{\alpha}{4(r_0-\overline{r})}+\frac{\beta}{12}$.

The main advantage of this formulation is that the integrals \eqref{carlsonRF} can be approximated efficiently by the Carlson algorithm \cite[\S 2]{carlson} and we can approach $\wp$ using the Coquereaux--Grossmann--Lautrup algorithm \cite[\S 3]{coquereaux}. This is the method we use to approximate ERN orbits and produce a Python code.

\addtocounter{equation}{1}

\subsection{The special case of the Euclidean Schwarzschild geometry}\label{schwarzschild}
When $Q=0$, we have $F(0)=0$ so that we may take $\overline{r}=0$ so that the polar equation in Binet variable simplifies to
\[\dot{u}^2=u^4F(1/u)=2Mu^3-u^2-\frac{2M}{L^2}u+\frac{1-E^2}{L^2}\]
which is \cite[equation (2.16)]{battista-esposito22}. Therefore, the Weierstrass form
\[\dot{\wp}^2=4\wp^3-g_2\wp-g_3\]
is obtained by letting $\wp:=M/(2r)-1/12$, as well as
\begin{equation}\label{ges}\tag{$g_{12}^{\rm ES}$}
g_2:=\frac{1}{12}+\frac{M^2}{L^2}\text{ and }g_3:=\frac{1}{216}-\frac{M^2}{12L^2}(2-3E^2).
\end{equation}
These expression can be compared to the Lorentzian Schwarzschild case, where for a test-particle of mass $\mu=2\mathcal{L}$ (twice the Schwarzschild Lagrangian), the constants $g_2$ and $g_3$ read (cf \cite[\S 4, p. 84]{hagihara})
\begin{equation}\label{gls}\tag{$g_{12}^{\rm LS}$}
g_2=\frac{1}{12}+\frac{\mu M^2}{L^2}\text{ and }g_3=\frac{1}{216}-\frac{M^2}{12L^2}(2\mu+3E^2).
\end{equation}

\textit{Therefore, one may view a Euclidean Schwarzschild geodesic as a space-like Lorentzian Schwarzschild geodesic with complex energy $E_{{\rm Euc}}=iE_{{\rm Lor}}$.} Notice that this last equality is expected since the energy is defined as the temporal momentum and because the Euclidean time $\tau$ and the Lorentzian time $t$ are related by the relation $\tau=it$.

The Weierstrass formulation also permits to derive a shorter proof of the fact that $E<1$ for every geodesic $\gamma$ in Euclidean Schwarzschild geometry, with initial radius $r_0>r_+=2M$. As above, we may simplify the notation by rescaling the radius and assuming that $M=1$. If $L=0$, then the second equation \eqref{mot_ERN} reduces to $({{\rm d}r}/{{\rm d}s})^2=1-E^2-{2}/{r}$ so $1-E^2\ge 2/r>0$, as claimed. Now if $L\ne0$, then we may use Weierstrass' form and the discriminant of the equation $4z^3-g_2z-g_3$ reads
\[\Delta:=16(g_2^3-27g_3^2)=L^{-6}\left[(1-E^2)L^4-(27E^4-36E^2+8)L^2+16\right]=\frac{1-E^2}{L^2}+\bigo\left(\frac{1}{L^4}\right),\]
so that if, for the sake of contradiction, we assume $E^2>1$, then $\Delta<0$ for $L\gg0$, something which can be achieved by rescaling the initial azimuthal angular velocity. Therefore, we assume that $\Delta<0$ and look for an absurdity.

Consider the Weierstrass cubic $\mathfrak{q}(x):=4x^3-g_2x-g_3$; then the phase portrait in Weierstrass variable $x=1/(2r)-1/12$ describes a portion of the (real) elliptic curve
\[\mathcal{E}_{\mathfrak{q}}:=\{(x,y)\in\R^2~|~y^2=\mathfrak{q}(x)\}.\]
Since $\lim_{x\to+\infty}\mathfrak{q}(x)=+\infty$, we have $\mathcal{E}_\mathfrak{q}\cap\{x=x_0\}\ne\emptyset$ for $x_0\gg0$ but because $\mathfrak{q}(1/6)=-(E/2L)^2<0$, the non-compact connected component of $\mathcal{E}_\mathfrak{q}$ lies in the open half-plane $\{x>1/6\}$ (which corresponds in radial variable to $\{r<2\}$), so that the considered phase portrait cannot describe a portion of this component. However, to say that $\Delta<0$ amounts to say that $\mathcal{E}_\mathfrak{q}$ is connected, a contradiction.

At this point, we know that $\Delta\ge0$. If $\Delta>0$, then the elliptic curve $\mathcal{E}_\mathfrak{q}$ has an additional compact connected component and since $\mathfrak{q}(-1/12)=(1-E^2)/(4L^2)>0$, this component intersects the subset $\{x=-1/12\}$. Therefore, the corresponding orbit is indeed unbounded. In the case where $\Delta=0$, the curve is singular but connected so that the orbit imposes $0\le\mathfrak{q}(1/6)=-(E/2L)^2\le0$ so $r=2$, a new contradiction. We summarize the discussion as follows:
\begin{adjustwidth}{1cm}{1cm}
\emph{Any Euclidean Schwarzschild orbit has $E^2<1$, and is unbounded. Moreover, the phase portrait of a non-radial equatorial Euclidean Schwarzschild orbit in Weierstrass variable $\wp=M/(2r)-1/12$ (with polar argument) describes a portion of the unique compact connected component of the associated (real) elliptic curve $\mathcal{E}$, included in $\mathcal{E}\cap\{-1<12\wp<2\}$. In particular, the discriminant of $\mathcal{E}$ is positive.}
\end{adjustwidth}

\section{Gravitational deflection of Euclidean Schwarzschild orbits}\label{deflection}

Consider an equatorial orbit $\gamma=(\tau,r,\phi)$, with energy $-1<E<1$ and angular momentum $L\in\R^*$. Since $\gamma$ is unbounded, we may consider its \textit{velocity at infinity}, defined by $v^2:=\lim_{r\to\infty}({\rm d}r/{\rm d}\tau)^2$. Using equations \eqref{tandphi} and \eqref{mot_ERN}, we find the expression
\[v^2=\lim_{r\to\infty}\left(\frac{{\rm d}r}{{\rm d}\tau}\right)^2=\lim_{r\to\infty}\left(\frac{\dot{r}}{\dot{\tau}}\right)^2=\lim_{r\to\infty}\frac{\Delta(r)^2}{C^2}\left(\Delta(r)\left(\mathcal{H}-\frac{J^2}{r^2}\right)-C^2\right)=\frac{1}{E^2}-1\]
The equation \eqref{polmot_ERN} can be recast in Binet variable $u=1/r$ and yields
\begin{equation}\label{polES}
\left(\frac{{\rm d}u}{{\rm d}\phi}\right)^2+u^2=2Mu^3-\frac{2M}{L^2}u+\frac{1}{b^2}=2Mu\left(u^2-\frac{1}{L^2}\right)+\frac{1}{b^2},
\end{equation}
where the constant $b:=\sqrt{\frac{L^2}{1-E^2}}$ is the \textit{impact parameter} of $\gamma$, satisfying $L=\pm bvE$. Since
\[\lim_{r\to\infty}\frac{{\rm d}\phi}{{\rm d}\tau}=\lim_{r\to\infty}\frac{J\Delta(r)}{Cr^2}=0,\]
the orbit admits asymptotic lines, and we are first interested in the \textit{deflection angle at infinity} $\delta\phi$ between these two asymptotic directions, as a function of the orbit's \textit{perihelion} $r=r_{\rm min}$.

\begin{center}
\begin{figure}[h!]
\begin{tikzpicture}[scale=1,rotate=-90]
  \coordinate (z) at (0,0);
  \coordinate (zp) at (0,4);
  \coordinate (zm) at (0,-4);
  
  \draw[color=red,domain=-1.55*pi/4:1.55*pi/4,scale=1,samples=200] plot ({deg(\x)}:{2/sqrt(3^2*cos(\x r)^2-1)});
  
  \coordinate (pu) at ({7*1/3},{7*sqrt(1-1/3^2)});
  \coordinate (pup) at ({7.5*1/3},{7.5*sqrt(1-1/3^2)});
  \coordinate (pd) at ({7.5*1/3},{-7.5*sqrt(1-1/3^2)});
  \coordinate (mpu) at ({-2*1/3},{-2*sqrt(1-1/3^2)});
  \coordinate (mpd) at ({-2*1/3},{2*sqrt(1-1/3^2)});
  \coordinate (a) at ({2/sqrt(3^2-1)},0);
  \coordinate (ap) at ({2/sqrt(3^2-1)},5);
  \coordinate (am) at ({2/sqrt(3^2-1)},-5);
  \coordinate (f1) at ({2*3/sqrt(3^2-1)},0);
  
  \coordinate (Pb) at ({2*3/sqrt(3^2-1)+1.15*(5.5*1/3)},{0+1.15*(5.5*sqrt(1-1/3^2))});
  \coordinate (Pbp) at ({2*3/sqrt(3^2-1)+1.25*(5.5*1/3)},{0+1.25*(5.5*sqrt(1-1/3^2))});
  
  \fill[fill=black] (f1) circle (6pt);
  
  \draw[dotted] (mpu)--(pup) (mpd)--(pd);
  \draw (z)--(a) (ap)--(am);
  \draw[dotted] (f1)--(Pbp);
  \draw[ultra thick] (a)--(f1);
  \draw[<->,dashed,very thick] (Pb)--(pu);
  \draw (f1)--($(a)!1.5!(f1)$);
  \draw (z)--($(z)!-0.5!(f1)$);
  
  \draw ($(Pb)!0.5!(pu)$) node[right]{$b$};
  \draw ($(a)!1.1!(f1)$) node[below right]{$0$};
  \draw ($(f1)!0.5!(a)$) node[right]{$r_{\rm min}$};

  \draw ($(z)!-0.75!(f1)$) coordinate (A) (z) coordinate (B) (pd) coordinate (C) pic ["$\phi$",draw,->,angle radius=0.8cm]{angle};
  \draw (pu) coordinate (A) (z) coordinate (B) (mpd) coordinate (C) pic ["$\delta\phi$",draw,->,angle radius=1.5cm]{angle};
  
  \coordinate (de) at (0.71,2);
  \fill[fill=black] (de) circle (1.5pt);
  \fill[fill=black] (z) circle (1.5pt);
  \draw (pu) coordinate (A)  (de) coordinate (B)  (ap) coordinate (C) pic ["$\delta\phi/2$",draw,->,angle radius=2.5cm]{angle}; 
\end{tikzpicture}
\caption{Schematics of an orbit with angle deflection $\delta\phi=2\phi-\pi$.}\label{schematics}
\end{figure}
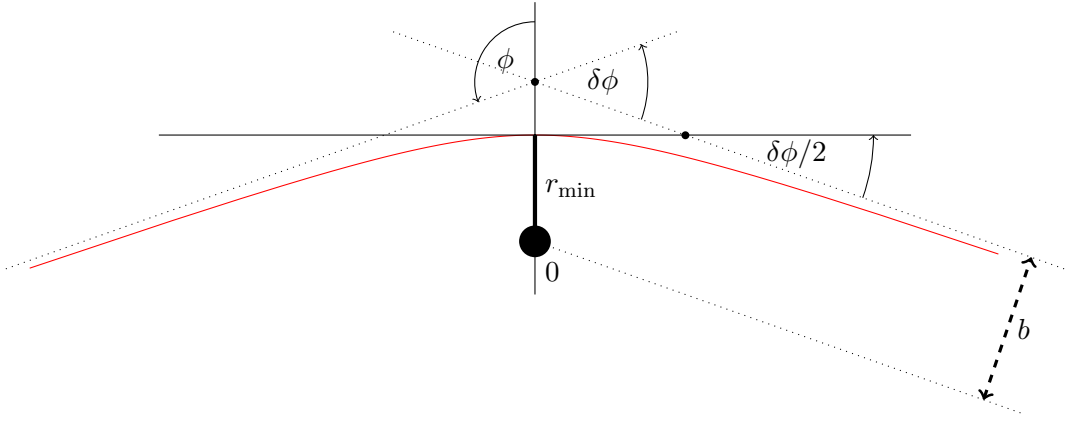
\end{center}

\subsection{Analytic expression of the deflection angle using Carlson's integrals}
As illustrated in the Figure \ref{schematics}, the deflection (at infinity) $\delta\phi$ is given by
\[\delta\phi=2|\phi(r=\infty)-\phi(r=r_{\rm min})|-\pi\]
and using \eqref{polmot_ERN} again leads to the expression
\[\delta\phi=2\int_{r_{\rm min}}^\infty{\rm d}\phi-\pi=2\int_{r_{\rm min}}^\infty\frac{{\rm d}\phi}{{\rm d}r}{\rm d}r-\pi=2\int_{r_{\rm min}}^\infty\frac{1}{\sqrt{\frac{1}{b^2}-\frac{1}{r^2}\left(1-\frac{2M}{r}+\frac{2Mr}{L^2}\right)}}\frac{{\rm d}r}{r^2}-\pi.\]
This expression can be simplified using the Weierstrass variable $\wp=Mu/2-1/12$. Indeed, using the constants $g_2,g_3$ given by equation \eqref{ges}, we have
\[\delta\phi=2\left|\int_0^{1/r_{\rm min}}\frac{{\rm d}u}{\sqrt{\frac{1}{b^2}-u^2+2Mu\left(u^2-\frac{1}{L^2}\right)}}\right|-\pi=2\int_{-\tfrac{1}{12}+\tfrac{M}{2r_{\rm min}}}^ {-\tfrac{1}{12}}\frac{{\rm d}p}{\sqrt{4p^3-g_2p-g_3}}-\pi.\]
This expression can be rewritten in terms of elliptic integrals, as in the pioneer work \cite{darwin_angle}. However, it is both easier and numerically more adequate to express it with the integrals \eqref{carlsonRF}, which we numerically approximate using Carlson's algorithm \cite{carlson}. Observe first that because $r=r_{\rm min}$ is a turning point of $\gamma$, we have
\[0=\left.\frac{{\rm d}\wp^2}{{\rm d}\phi^2}\right|_{r=r_{\rm min}}=4\wp(r=r_{\rm min})^3-g_2\wp(r=r_{\rm min})-g_3.\]
In other words, the point $\wp_{\rm max}:=M/(2r_{\rm min})-1/12$ is a root of the Weierstrass cubic and this leads to the factorization
\[4p^3-g_2p-g_3=(p-\wp_{\rm max})\left(4p^2+\left(\frac{2M}{r_{\rm min}}-\frac13\right)p+\left(\frac{1}{36}-g_2-\frac{M}{3r_{\rm min}}+\frac{M^2}{r_{\rm min}^2}\right)\right),\]
which allows to find the other two roots $\wp_{\pm}\in\C$ of the cubic. We then have
\begin{align*}
\int^{-\tfrac{1}{12}}_{\wp_{\rm max}}\frac{{\rm d}p}{\sqrt{4p^3-g_2p-g_3}}&=\int^0_{\tfrac{M}{2r_{\rm min}}}\frac{{\rm d}\zeta}{\sqrt{4\left(\zeta-\frac{1}{12}\right)^3-g_2\left(\zeta-\frac{1}{12}\right)-g_3}} \\
&=\frac12\int^0_{\tfrac{M}{2r_{\rm min}}}\frac{{\rm d}\zeta}{\sqrt{\left(\zeta-\wp_{\rm max}-\tfrac{1}{12}\right)\left(\zeta-\wp_--\tfrac{1}{12}\right)\left(\zeta-\wp_+-\tfrac{1}{12}\right)}} \\
&=\frac12\left[\int_{\tfrac{M}{2r_{\rm min}}}^\infty-\int_0^\infty\right]\frac{{\rm d}\zeta}{\sqrt{\left(\zeta-\wp_{\rm max}-\tfrac{1}{12}\right)\left(\zeta-\wp_--\tfrac{1}{12}\right)\left(\zeta-\wp_+-\tfrac{1}{12}\right)}} \\
&=R_F\left(e_{\rm max}+\tfrac{M}{2r_{\rm min}},e_-+\tfrac{M}{2r_{\rm min}},e_++\tfrac{M}{2r_{\rm min}}\right)-R_F(e_{\rm max},e_-,e_+),
\end{align*}
where $e_{{\rm max},\pm}:=-\wp_{{\rm max},\pm}-1/12$. Thus, we obtain the following expression for the deflection:
\begin{equation}\label{deltaphi_carlson}
\delta\phi=2\left(R_F\left(e_{\rm max}+\tfrac{M}{2r_{\rm min}},e_-+\tfrac{M}{2r_{\rm min}},e_++\tfrac{M}{2r_{\rm min}}\right)-R_F(e_{\rm max},e_-,e_+)\right)-\pi
\end{equation}
This is an efficient formula for numerical calculations (see below), but it is natural to ask for an approximation of $\delta\phi$ when $r_{\rm min}\to\infty$. This is the goal of the next subsection. Observe moreover that the two terms of \eqref{deltaphi_carlson} involving the integrals \eqref{carlsonRF} can be complex, but their difference is real and non-negative.

\subsection{Approximation of the deflection angle with perturbed solution}
The aim of this section is to obtain an expansion of $\delta\phi$ in powers of the perihelion $r_{\rm min}$, up to order 3. This choice of order will become transparent later, when we study the null and maximal deflection rings.

To do so, we could use for instance the well-known perturbed solution method \cite{straumann,gergely-darazs,briet-hobil,he_et_al}. However, we will avoid complicated calculations with a simple observation. First, using \eqref{mot_ERN}, the turning point condition at $r=r_{\rm min}$ gives the following relation
\[0=\left(1-\frac{2M}{r}\right)\left(1-\frac{L^2}{r^2}\right)-E^2,\]
or, in terms of $b$ and $v$,
\[r_{\rm min}^2\left(1-\frac{2M}{r_{\rm min}}\left(1+\frac{1}{v^2}\right)\right)=b^2\left(1-\frac{2M}{r_{\rm min}}\right),\]
that is,
\begin{equation}\label{b_rm}
b=r_{\rm min}\sqrt{\frac{1-\frac{2M}{r_{\rm min}}\left(1+\frac{1}{v^2}\right)}{1-\frac{2M}{r_{\rm min}}}}=r_{\rm min}\sqrt{1-\frac{2M}{v^2({r_{\rm min}}-2M)}}.
\end{equation}
Recall also that $L^2=b^2v^2(1+v^2)^{-1}$ and that the equation \eqref{polmot_ERN} in terms of $b$ and $L$ reads
\[\left(\frac{{\rm d}r}{{\rm d}\phi}\right)^2=r^4\left(\frac{1}{b^2}-\frac{1}{r^2}\left(1-\frac{2M}{r}+\frac{2Mr}{L^2}\right)\right).\]
On the other hand, in Lorentzian Schwarzschild geometry, the same equation
\[\left(\frac{{\rm d}r}{{\rm d}\phi}\right)^2=r^4\left(\frac{1}{b_{\rm Sch}^2}-\frac{1}{r^2}\left(1-\frac{2M}{r}-\frac{2Mr}{L_{\rm Sch}^2}\right)\right)\]
differs from the above one by just a sign in the term in $L^{-2}$, and we have in this case 
\[b_{\rm Sch}^2=r_{\rm min}\sqrt{1+\frac{2M}{v^2({r_{\rm min}}-2M)}},\]
as well as $L_{\rm Sch}^2=b^2v^2(1-v^2)^{-1}$. This means that the expression of the deflection angle in Euclidean Schwarzschild geometry with squared velocity $v^2$ is the same as the expression of the Lorentzian deflection angle with same perihelion and ``squared velocity'' $-v^2$. This is mathematically well-defined since the data only depend on the squared quantities $(b^2,v^2)$, and not on the pair $(b,v)$ itself, so no complex number is involved. Moreover, this interpretation is physically consistent with the relation $\tau=it$ between the Euclidean time $\tau$ and the Lorentzian time $t$ and the definition of the Euclidean velocity $v^2=\lim_{r\to\infty}({\rm d}r/{\rm d}\tau)^2$, while in the Lorentzian case, $v^2=\lim_{r\to\infty}({\rm d}r/{\rm d}t)^2$.

After \cite{accioly-ragusa} or \cite{li-zhou-li-he} for instance, up to order 3, in Lorentzian Schwarzschild geometry we have
\begin{equation}\label{def_sch}
\delta\phi=\frac{2M}{b}\left(1+\frac{1}{v^2}\right)+\frac{3\pi M^2}{4b^2}\left(1+\frac{4}{v^2}\right)+\frac{2M^3}{3b^3}\left(5+\frac{45}{v^2}+\frac{15}{v^4}-\frac{1}{v^6}\right)+\bigo\left(\frac{M^4}{b^4}\right).
\end{equation}
Therefore, changing the sign of the terms in $v^{-2}$ yields the following estimate, in Euclidean Schwarzschild geometry,
\begin{equation}\label{estimdeltaphi}\tag{$\delta_3^b$}
\delta\phi=\frac{2M}{b}\left(1-\frac{1}{v^2}\right)+\frac{3\pi M^2}{4b^2}\left(1-\frac{4}{v^2}\right)+\frac{2M^3}{3b^3}\left(5-\frac{45}{v^2}+\frac{15}{v^4}+\frac{1}{v^6}\right)+\bigo\left(\frac{M^4}{b^4}\right).
\end{equation}
Using now the relation \eqref{b_rm}, we arrive at the expression
\begin{align*}\label{estimdeltaphir}\tag{$\delta_3^r$}
\delta\phi=&\frac{2M}{r_{\rm min}}\left(1-\frac{1}{v^2}\right)+\frac{M^2}{r_{\rm min}^2}\left(\frac{3\pi}{4}\left(1-\frac{4}{v^2}\right)+\frac{2}{v^2}\left(1-\frac{1}{v^2}\right)\right) \\
&+\frac{M^3}{r_{\rm min}^3}\left(\frac{3\pi}{2v^2}\left(1-\frac{4}{v^2}\right)+\frac{10}{3}-\frac{26}{v^2}+\frac{9}{v^4}-\frac{7}{3v^6}\right)+\bigo\left(\frac{M^4}{r_{\rm min}^4}\right).
\end{align*}

Concerning the metric \eqref{ERN} with non-zero charge $Q\ne0$, the perturbed solution method gives, up to order 2 and after elementary calculations we omit,
\[\delta\phi=\frac{2M}{b}\left(1-\frac{1}{v^2}\right)+\frac{3\pi M^2}{4b^2}\left(1-\frac{4}{v^2}\right)-\frac{\pi Q^2}{4b^2}\left(1-\frac{2}{v^2}\right)+\bigo(b^{-3}),\]
in agreement with \cite{pang-jia}. Expressing the impact parameter $b$ in terms of the perihelion $r_{\rm min}$ as
\[b=r_{\rm min}\sqrt{1-\frac{2Mr_{\rm min}-Q^2}{v^2(r_{\rm min}^2-2Mr_{\rm min}+Q^2)}},\]
yields the expansion
\begin{equation}\label{def_with_charge}
\delta\phi=\frac{2M}{r_{\rm min}}\left(1-\frac{1}{v^2}\right)+\frac{M^2}{r_{\rm min}^2}\left(\frac{3\pi}{4}\left(1-\frac{4}{v^2}\right)+\frac{2}{v^2}\left(1-\frac{1}{v^2}\right)\right)-\frac{\pi Q^2}{4r_{\rm min}^2}\left(1-\frac{2}{v^2}\right)+\bigo(r_{\rm min}^{-3}).
\end{equation}
As a sanity check, observe that replacing $v^2$ by $-v^2$ in the previous expression and letting $v\to1$ leads to the light deflection formula of \cite[\S III.B]{briet-hobil}
\[\delta\phi=\frac{4M}{r_{\rm min}}+\frac{M^2}{r_{\rm min}^2}\left(\frac{15\pi}{4}-4\right)-\frac{3\pi Q^2}{4r_{\rm min}^2}+\bigo(r_{\rm min}^{-3}).\]

\subsection{Null and maximal deflection rings}
We start by observing that, at lowest order, the deflection angle for a usual (Lorentzian) massive Schwarzschild orbit with velocity at infinity $v$ is given by
\[\delta\phi_{\rm Lor}\approx\frac{2M}{r_{\rm min}}\left(1+\frac{1}{v^2}\right)>0,\]
while in the Euclidean background, the estimation \eqref{estimdeltaphir} reads
\[\delta\phi_{\rm Euc}\approx \frac{2M}{r_{\rm min}}\left(1-\frac{1}{v^2}\right).\]
This suggests that for the Euclidean Schwarzschild solution, the deflection $\delta\phi$ may vanish for some values of $r_{\rm min}$. This can also be noticed from the motion equation itself. Indeed, recall from \cite[\S 144]{mcmahon-snyder} that a polar curve, parametrized in Binet variable $u$ is \textit{concave} with respect to its pole (hence, has a positive deflection) if and only if $u+\frac{{\rm d}^2u}{{\rm d}\phi^2}>0$. Now, this quantity in Schwarzschild geometry is given by
\[\ddot{u}+u=M(3u^2-\epsilon L^{-2}),\]
where $\epsilon$ is the signature of the metric. Hence, this is always positive in the Lorentzian background $\epsilon=-1$, while it may change of sign in the Euclidean world $\epsilon=1$.

To numerically appreciate this phenomenon, in the Figure \ref{repulsive} we depict pencils of orbits with different velocities. Their deflection at infinity is computed using the formula \eqref{deltaphi_carlson} and Carlson's algorithm. Observe that the deflection stays negative when $v<1$. The more extreme case where $v=5$ is displayed in the Figure \ref{butterfly}, along with the graph representing the deflection angle as a function of the perihelion. Observe the presence of a null deflection at $r_{\rm min}\approx 2.16$ and a maximal deflection at $r_{\rm min}\approx2.5$. The Figure \ref{many_velocities} depicts the deflection as a function of the perihelion, for several velocities at infinity.
\vspace{-10mm}

\begin{center}
\begin{figure}
\begin{subfigure}[c]{0.495\columnwidth}
\centering
\includegraphics[width=1\linewidth]{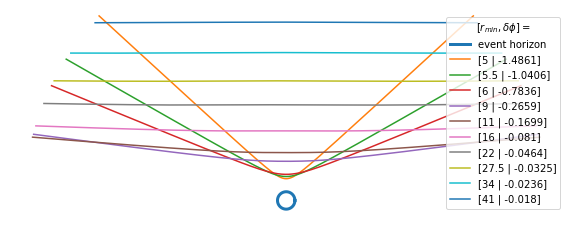}
\caption{$v=0.9$}
\end{subfigure}
\hfill
\begin{subfigure}[c]{0.495\columnwidth}
\centering
\includegraphics[width=1\linewidth]{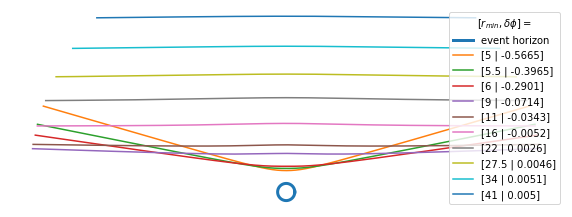}
\caption{$v=1.1$}
\end{subfigure}
\\
\begin{subfigure}[c]{0.495\columnwidth}
\centering
\includegraphics[width=1\linewidth]{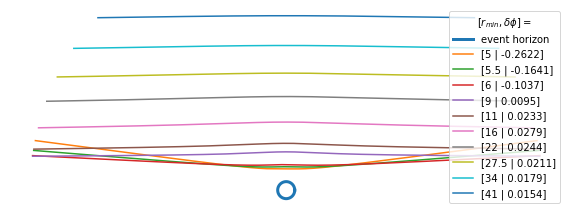}
\caption{$v=1.25$}
\end{subfigure}
\hfill
\begin{subfigure}[c]{0.495\columnwidth}
\centering
\includegraphics[width=1\linewidth]{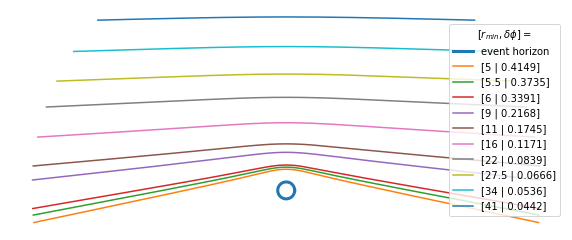}
\caption{$v=3$}
\end{subfigure}
\caption{Orbits with different velocities and their perihelion (in units of $M$) and deflection at infinity.}\label{repulsive}
\end{figure}
\end{center}
\begin{center}
\begin{figure}
\begin{subfigure}[c]{0.495\columnwidth}
\centering
\includegraphics[width=1\linewidth]{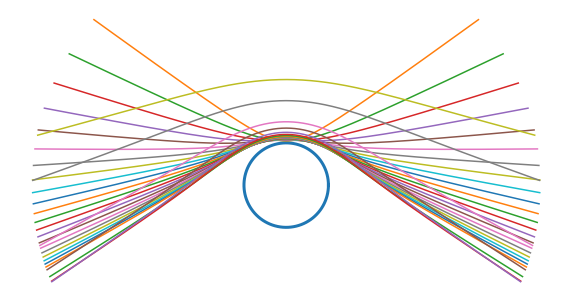}
\caption{Orbits with $2.1<r_{\rm min}<5$.}
\end{subfigure}
\hfill
\begin{subfigure}[c]{0.495\columnwidth}
\centering
\includegraphics[width=1\linewidth]{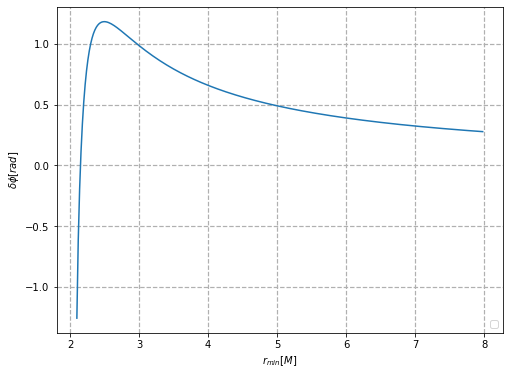}
\caption{The deflection $\delta\phi$ as a function of $r_{\rm min}$.}
\end{subfigure}
\caption{Orbits with $v=5$ and corresponding deflection graph.}\label{butterfly}
\end{figure}
\end{center}
\begin{center}
\begin{figure}
\includegraphics[scale=0.45]{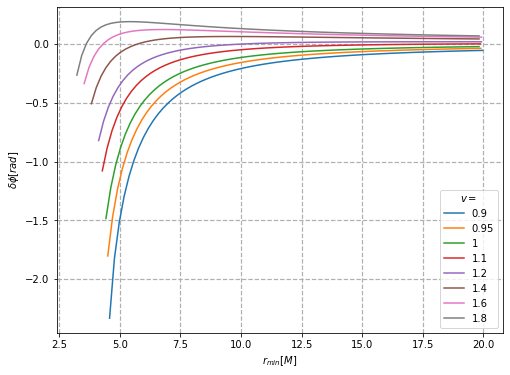}
\caption{Representations of $\delta\phi=\delta\phi(r_{\rm min})$ for several values of $v$.}\label{many_velocities}
\end{figure}
\end{center}

Observe that in the Lorentzian (resp. Euclidean) Schwarzschild spacetime, the impact parameter $b$ of a massive particle (resp. of any particle) satisfies $bv=L/E$. Moreover, a photon in the usual Schwarzschild metric has $b=L/E$ and this motivates the following terminology: an orbit in a Euclidean Schwarzschild spacetime with velocity at infinity $v$ will be called of \textit{sub-photon type} (resp. of \textit{sup-photon type}) if $v\lesssim1$ (resp. $v\gtrsim1$).

We have numerically observed that sub-photon orbits have a negative deflection (i.e. are repelled by the central mass) and that the absolute value of $\delta\phi$ behaves like in the usual Schwarzschild spacetime. This can also be noticed using the lowest order approximations from \eqref{estimdeltaphir} and \eqref{def_sch} recalled at the beginning of the subsection. 

For orbits with fixed velocity $v>1$, one could use numerical methods on the expression \eqref{deltaphi_carlson} to find an approximation of the critical value $\rho_0$ (resp. $\rho_{\rm max}$) such that the deflection angle of the orbit with perihelion $r_{\rm min}=\rho_0$ (resp. $r_{\rm min}=\rho_{\rm max}$) vanishes (resp. is maximal). However, we can also use the estimate \eqref{estimdeltaphir}. Indeed, as it is a third order approximation of the deflection, to solve $\delta\phi=0$ at this order amounts to solve a quadratic equation. Similarly, we can differentiate \eqref{estimdeltaphir} with respect to $r_{\rm min}$ to obtain an approximation of $\rho_{\rm max}$ using the quadratic formula again. This is precisely the reason why we chose this order at the first place. Though the resulting expressions for $\rho_0$ and $\rho_{\rm max}$ are not very enlightening, it is worth noticing that they both diverge when $v\to1$. This suggests to expand these expressions in powers of $v-1$, yielding estimations of $\rho_0$ and $\rho_{\rm max}$ for orbits of sup-photon type. We find
\begin{equation}\label{Rm0}
\frac{\rho_0}{M}=\frac{9\pi}{16}(v-1)^{-1}+1-\frac{21\pi}{32}+\frac{64}{9\pi}+\bigo(v-1)\approx\frac{1.767}{v-1} +1.202+\bigo(v-1)
\end{equation}
and
\begin{equation}\label{Rm1}
\frac{\rho_{\rm max}}{M}=\frac{9\pi}{8}(v-1)^{-1}+1-\frac{21\pi}{16}+\frac{32}{3\pi}+\bigo(v-1)\approx\frac{3.534}{v-1}+0.272+\bigo(v-1).
\end{equation}
Observe that, at lowest order, we have the remarkable relation $\rho_{\rm max}=2\rho_0$. Moreover, at this order, the maximal value of the deflection is given by
\begin{equation}\label{deltamax}\tag{$\delta_{\rm max}$}
\delta\phi_{\rm max}\approx\frac{16(v-1)}{9\pi}\left(1-\frac{1}{v^2}\right).
\end{equation}
However, this estimate badly fails when $v\gg1$. We illustrate our approximations in the Figure \ref{approxes}.

\vspace{-5mm}

\begin{center}
\begin{figure}
\begin{subfigure}[c]{0.495\columnwidth}
\centering
\includegraphics[width=1\linewidth]{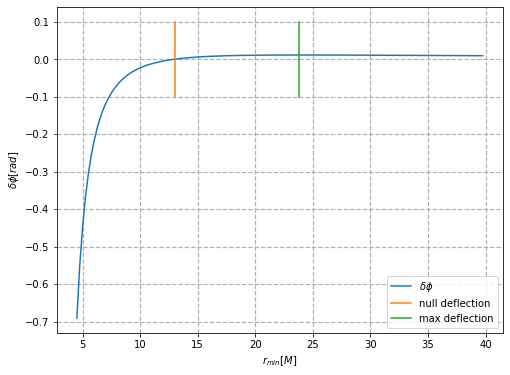}
\caption{$v=1.15$}
\end{subfigure}
\hfill
\begin{subfigure}[c]{0.495\columnwidth}
\centering
\includegraphics[width=1\linewidth]{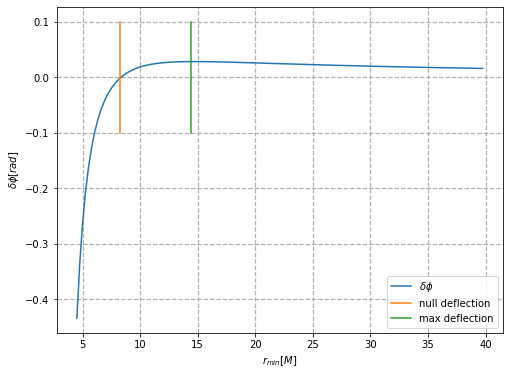}
\caption{$v=1.25$}
\end{subfigure}
\\
\begin{subfigure}[c]{0.495\columnwidth}
\centering
\includegraphics[width=1\linewidth]{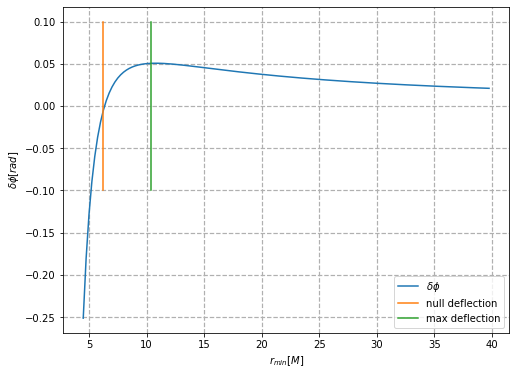}
\caption{$v=1.35$}
\end{subfigure}
\hfill
\begin{subfigure}[c]{0.495\columnwidth}
\centering
\includegraphics[width=1\linewidth]{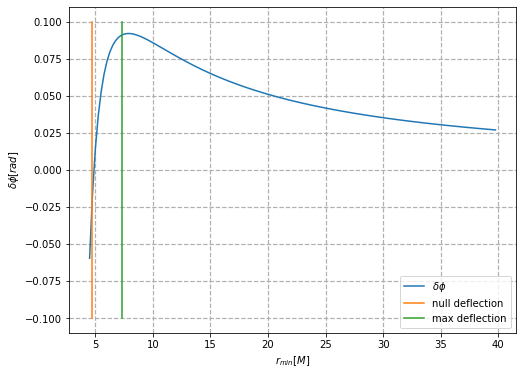}
\caption{$v=1.5$}
\end{subfigure}
\caption{Accuracy of the approximations \eqref{Rm0} and \eqref{Rm1} for several velocities.}\label{approxes}
\end{figure}
\end{center}

\subsection{Deflection of orbits passing through a given point}

Consider now an orbit starting from a given point $(\phi,r)=(0,r_0)$, for some fixed $r_0>2M$, whose velocity vector at this point makes an angle $0<\alpha\le\tfrac{\pi}{2}$ with respect to the radial direction. Thinking of the point $(\phi,r)=(0,r_0)$ as an observer receiving particles coming from infinity, we are interested in the total deflection angle $\delta_\alpha\phi$ between the tangent line to the orbit at $(0,r_0)$ and the asymptotic direction. The situation may be visualized using the Figure \ref{schematicsa}.
\vspace{-5mm}

\begin{center}
\begin{figure}
\begin{tikzpicture}[rotate=-80.2]
  \coordinate (z) at (0,0);
  \coordinate (zp) at (0,4);
  \coordinate (zm) at (0,-4);
  
  \draw[color=red,domain=-1.17*pi/4:1.17*pi/4,scale=1,samples=200] plot ({deg(\x)}:{2/sqrt(1.7^2*cos(\x r)^2-1)});
  
  \coordinate (pu) at ({8*1/1.7},{8*sqrt(1-1/1.7^2)});
  \coordinate (pup) at ({8.5*1/1.7},{8.5*sqrt(1-1/1.7^2)});
  \coordinate (pd) at ({8*1/1.7},{-8*sqrt(1-1/1.7^2)});
  \coordinate (mpu) at ({-3.5*1/1.7},{-3.5*sqrt(1-1/1.7^2)});
  \coordinate (mpd) at ({-3.5*1/1.7},{3.5*sqrt(1-1/1.7^2)});
  \coordinate (a) at ({2/sqrt(1.7^2-1)},0);
  \coordinate (ap) at ({2/sqrt(1.7^2-1)},4.5);
  \coordinate (am) at ({2/sqrt(1.7^2-1)},-4.5);
  \coordinate (f1) at ({2*1.7/sqrt(1.7^2-1)},0);
  
  \coordinate (rph) at ({2*cosh(1.4)/sqrt(1.7^2-1)},{2*sinh(1.4)});
  \coordinate (rphp) at ({2*cosh(1.4)/sqrt(1.7^2-1)-1.05*2*sinh(1.4)/sqrt(1.7^2-1)},{2*sinh(1.4)-1.05*2*cosh(1.4)});
  \coordinate (rphm) at ({2*cosh(1.4)/sqrt(1.7^2-1)+0.5*2*sinh(1.4)/sqrt(1.7^2-1)},{2*sinh(1.4)+0.5*2*cosh(1.4)});
  
  \coordinate (R) at ({2*1.7/sqrt(1.7^2-1)+1.5*(2*cosh(1.4)/sqrt(1.7^2-1)-2*1.7/sqrt(1.7^2-1))},{0+1.5*(2*sinh(1.4)-0)});
  \coordinate (L) at ({2*1.7/sqrt(1.7^2-1)-1.5*(2*cosh(1.4)/sqrt(1.7^2-1)-2*1.7/sqrt(1.7^2-1))},{0-1.5*(2*sinh(1.4)-0)});
  \coordinate (Pb) at ({2*1.7/sqrt(1.7^2-1)+1.2*(5.5*1/1.7)},{0+1.2*(5.5*sqrt(1-1/1.7^2))});
  \coordinate (Pbp) at ({2*1.7/sqrt(1.7^2-1)+1.3*(5.5*1/1.7)},{0+1.3*(5.5*sqrt(1-1/1.7^2))});
  
  \fill[fill=black] (rph) circle (3pt);
  \fill[fill=black] (f1) circle (6pt);
  
  \draw[dotted] (mpu)--(pup) (mpd)--(pd);
  \draw (z)--(a);
  \draw[very thick] (rph)--(rphp) (rph)--(rphm);
  \draw (rph)--(R) (f1)--(L);
  \draw[dotted] (f1)--(Pbp);
  \draw[ultra thick] (a)--(f1) (f1)--(rph);
  \draw[<->,dashed,very thick] (Pb)--(pu);
  \draw (f1)--($(a)!2.5!(f1)$);
  \draw (z)--($(z)!-0.75!(f1)$);
  
  \draw ($(Pb)!0.5!(pu)$) node[right]{$b$};
  \draw ($(f1)!0.5!(rph)$) node[below]{$r_0$};
  \draw ($(a)!1.1!(f1)$) node[below right]{$0$};
  \draw ($(f1)!0.5!(a)$) node[right]{$r_{\rm min}$};
  
  \draw (rphm) coordinate (A)  (rph) coordinate (B)  (R) coordinate (C) pic ["$\alpha$",draw,->,angle radius=1.5cm]{angle};
  \draw (L) coordinate (A)  (f1) coordinate (B)  ($(a)!2.5!(f1)$) coordinate (C) pic ["$\Phi$",draw,->,angle radius=1cm]{angle};
  \draw ($(z)!-0.75!(f1)$) coordinate (A) (z) coordinate (B) (pd) coordinate (C) pic ["$\phi$",draw,->,angle radius=0.8cm]{angle};
  \draw (mpu) coordinate (A) (z) coordinate (B) (pd) coordinate (C) pic ["$\delta\phi=2\phi-\pi$",draw,->,angle radius=3.2cm]{angle};
  
  \coordinate (De) at (0.36,-0.485);
  \coordinate (de) at (1.46,2);
  \fill[fill=black] (De) circle (1.5pt);
  \fill[fill=black] (z) circle (1.5pt);
  \draw (rph) coordinate (A) (De) coordinate (B) (mpd) coordinate (C) pic ["$\delta_\alpha\phi$",draw,->,angle radius=1.5cm,very thick]{angle};  
\end{tikzpicture}
\caption{Schematics of an orbit passing through the point $(\phi,r)=(0,r_0)$, with total deflection $\delta_\alpha\phi=\Phi+\alpha+\tfrac12(\delta\phi-\pi)=\alpha-\pi+\phi+\Phi$.}\label{schematicsa}
\end{figure}
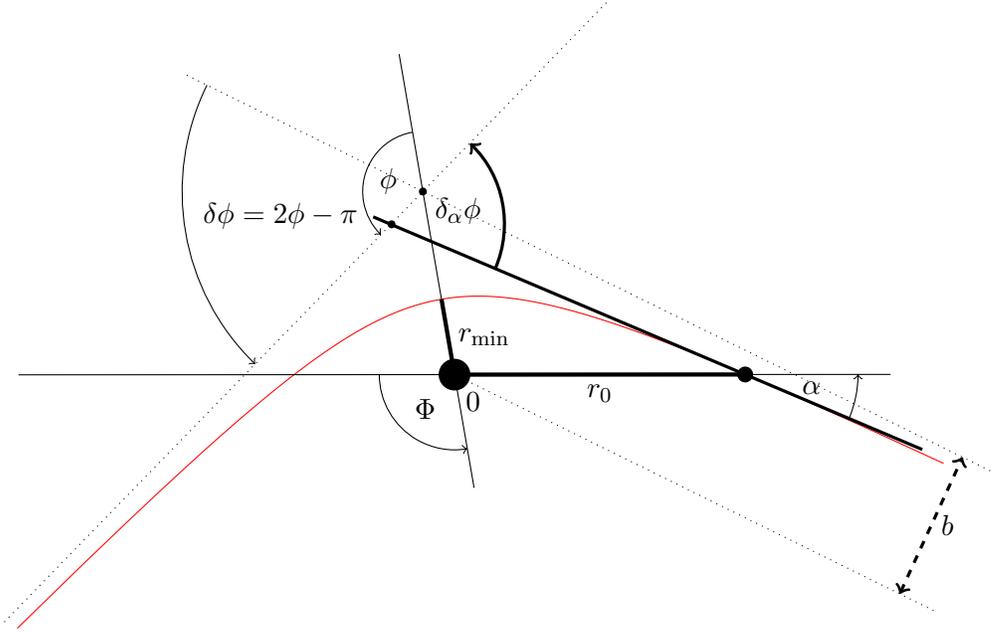
\end{center}

We first proceed as in the first subsection, to obtain a numerically efficient analytic formula for $\delta_\alpha\phi$, using Carlson's integrals. Recall the Weierstrass variable $\wp=M/(2r)-1/12$, satisfying the Weierstrass equation with constants $g_{2,3}=g_{2,3}(v,r_{\rm min})$ given by \eqref{ges}. Recall also that $\wp_{\rm max}=M/(2r_{\rm min })-1/12$ is the maximal value of $\wp$ along the geodesic. We have
\[\phi=\int_{\infty}^{r_{\rm min}}\frac{{\rm d}r}{r^2\sqrt{\frac{1}{b^2}-\frac{1}{r^2}\left(1-\frac{2M}{r}+\frac{2Mr}{L^2}\right)}}=\int_{\wp_{\rm max}}^{-\frac{1}{12}}\frac{{\rm d}p}{\sqrt{4p^3-g_2p-g_3}},\]
and
\[\Phi=\int_{r_0}^{r_{\rm min}}\frac{{\rm d}r}{r^2\sqrt{\frac{1}{b^2}-\frac{1}{r^2}\left(1-\frac{2M}{r}+\frac{2Mr}{L^2}\right)}}=\int_{\wp_{\rm max}}^{\wp_0}\frac{{\rm d}p}{\sqrt{4p^3-g_2p-g_3}},\]
where $\wp_0:=M/(2r_0)-1/12$, so that
\[\delta_\alpha\phi=\alpha-\pi+\phi+\Phi=\alpha-\pi+\left[\int_{\wp_{\max}}^{-\frac{1}{12}}+\int_{\wp_{\rm max}}^{\wp_0}\right]\frac{{\rm d}p}{\sqrt{4p^3-g_2p-g_3}}.\]
We now introduce the two remaining roots $\wp_{\pm}\ne\wp_{\rm max}$ of the cubic $4p^3-g_2p-g_3$, as well as the constants $e_{{\rm max},\pm}:=-\wp_{{\rm max},\pm}-1/12$ as above and arrive at the expression
\begin{align*}\label{Deltaphi_carlson}\tag{$13$}
\delta_\alpha\phi=&R_F\left(e_{\rm max}+\tfrac{M}{2r_{\rm min}},e_-+\tfrac{M}{2r_{\rm min}},e_++\tfrac{M}{2r_{\rm min}}\right)+R_F\left(e_{\rm max}+\tfrac{M}{2r_0},e_-+\tfrac{M}{2r_0},e_++\tfrac{M}{2r_0}\right) \\
&+\alpha-\pi-2R_F(e_{\rm max},e_-,e_+).
\end{align*}
\addtocounter{equation}{1}

It may be observed that, at fixed $\alpha\in]0;\pi/2]$, when $r_0\gg0$, we have $\Phi\approx\tfrac{\pi}{2}-\alpha$, so $r_{\rm min}\approx r_0\cos(\Phi)\approx r_0\sin(\alpha)$ and $\delta_\alpha\phi\approx\tfrac12\delta\phi(r_{\rm min}=r_0\sin\alpha)$. On the other hand, at fixed $r_0$, when $\alpha=\tfrac{\pi}{2}$, we have $r_{\rm min}=r_0$ and $\Phi=0$. We thus obtain the boundary conditions
\[\delta_\alpha\phi\stackrel[r_0\to\infty]{}\approx \frac12\delta\phi(r_{\rm min}=r_0\sin\alpha)\text{ and }\delta_{{\pi}/{2}}\phi=\frac12\delta\phi(r_{\rm min}=r_0).\]
In particular, at lowest order,
\[\delta_\alpha\phi\stackrel[r_0\to\infty]{}\approx\frac{M}{r_0\sin\alpha}\left(1-\frac{1}{v^2}\right).\]
We illustrate the situation in the Figure \ref{brush}.
\vspace{-5mm}

\begin{center}
\begin{figure}
\begin{subfigure}[c]{0.495\columnwidth}
\centering
\includegraphics[width=1\linewidth]{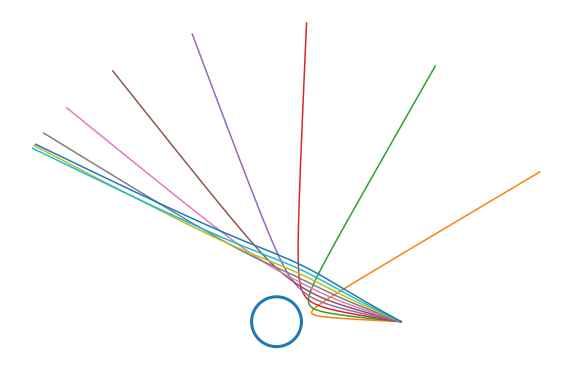}
\caption{Orbits with $0.05<\alpha<0.5$.}\label{brusha}
\end{subfigure}
\hfill
\begin{subfigure}[c]{0.495\columnwidth}
\centering
\includegraphics[width=1\linewidth]{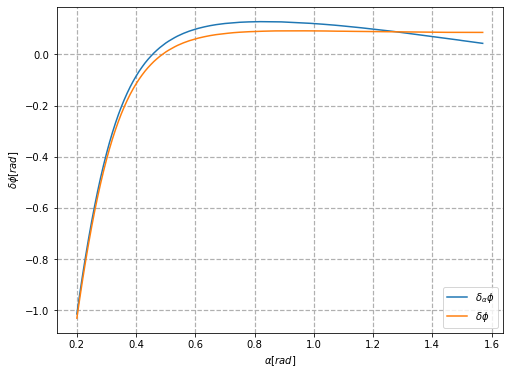}
\caption{Total deflection $\delta_\alpha\phi$ and deflection at infinity $\delta\phi$, as functions of $\alpha$.}\label{brushb}
\end{subfigure}
\caption{Orbits with $r_0=10M$ and $v=1.5$, and the corresponding graph.}\label{brush}
\end{figure}
\end{center}

For fixed values of $r_0$ and $v$, we now aim to find the critical angle $\alpha_0$ (resp. $\alpha_{\rm max}$) at which the total deflection $\delta_\alpha\phi$ vanishes (resp. is maximal). First recall the following relation from \cite{beloborodov02} (which is just the Pythagorian theorem in curved space)
\[\tan^2\alpha=\left.\frac{p_\phi\dot{\gamma}^\phi}{p_r\dot{\gamma}^r}\right|_{r=r_0}=\left.r^2\left(1-\frac{2M}{r}\right)\left(\frac{{\rm d}\phi}{{\rm d}r}\right)^2\right|_{r=r_0}=\frac{1-\frac{2M}{r_0}}{\frac{r_0^2}{b^2}-1+\frac{2M}{r_0}-\frac{2Mr_0}{L^2}}\]
yielding
\begin{equation}\label{tanalpha}
\alpha=\arctan\sqrt{\frac{1-\frac{2M}{r_0}}{\frac{r_0^2}{b^2}-1+\frac{2M}{r_0}-\frac{2Mr_0}{L^2}}}.
\end{equation}
The value of the impact parameter $b$ and the angular momentum $L$ are obtained using the relations \eqref{b_rm} and $L=bv(1+v^2)^{-1/2}$. Therefore, we obtain the value of $\alpha$ as a function of the triple $(v,r_0,r_{\rm min})$. With the notation of the previous subsection, we are thus interested in the values of $\alpha$ at $(v,r_0,\rho_0)$ and $(v,r_0,\rho_{\rm max})$.

For sup-photon type orbits with high perihelion, we may inject the approximations \eqref{Rm0} and \eqref{Rm1} into \eqref{tanalpha}, expand in powers of $r_0$ up to order 2, and then expand in powers of $v-1$ up to order 0 and obtain the following estimates
\begin{equation}\label{alpha0}
\alpha_0=\frac{9\pi M}{16r_0}\left(1+\frac{M}{r_0}\right)(v-1)^{-1}+\frac{M}{r_0}\left(\frac{64}{9\pi}-\frac{21\pi}{32}\right)-\frac{M^2}{r_0^2}\left(\frac{57\pi}{32}-\frac{64}{9\pi}\right)+\bigo(v-1,r_0^{-3})
\end{equation}
and
\begin{equation}\label{alphamax}
\alpha_{\rm max}=\frac{9\pi M}{8r_0}\left(1+\frac{M}{r_0}\right)(v-1)^{-1}-\frac{M}{r_0}\left(\frac{21\pi}{16}-\frac{32}{3\pi}\right)-\frac{M^2}{r_0^2}\left(\frac{57\pi}{16}-\frac{32}{3\pi}\right)+\bigo(v-1,r_0^{-3}).
\end{equation}

We illustrate our approximations in the Figure \ref{approxesa}.
\vspace{-.5cm}

\begin{center}
\begin{figure}
\begin{subfigure}[c]{0.495\columnwidth}
\centering
\includegraphics[width=1\linewidth]{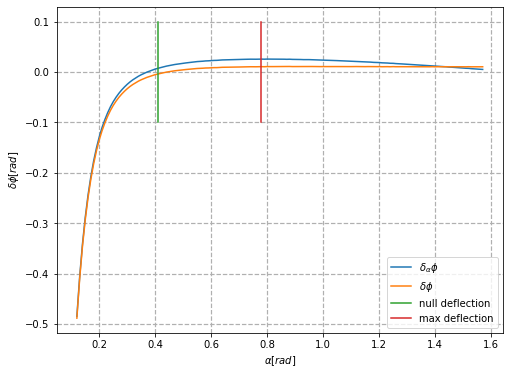}
\caption{$v=1.15$}
\end{subfigure}
\hfill
\begin{subfigure}[c]{0.495\columnwidth}
\centering
\includegraphics[width=1\linewidth]{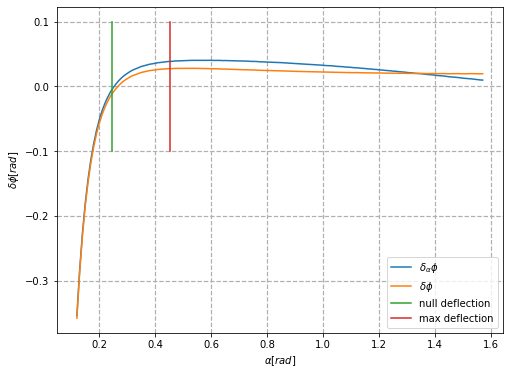}
\caption{$v=1.25$}
\end{subfigure}
\\
\begin{subfigure}[c]{0.495\columnwidth}
\centering
\includegraphics[width=1\linewidth]{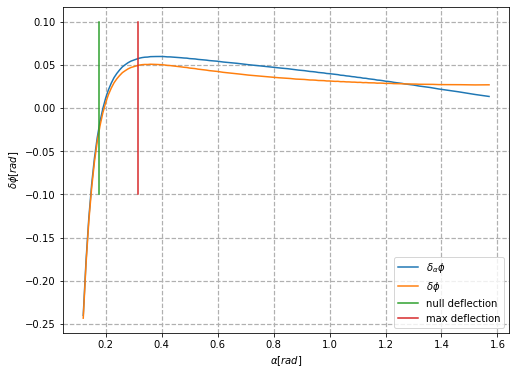}
\caption{$v=1.35$}
\end{subfigure}
\hfill
\begin{subfigure}[c]{0.495\columnwidth}
\centering
\includegraphics[width=1\linewidth]{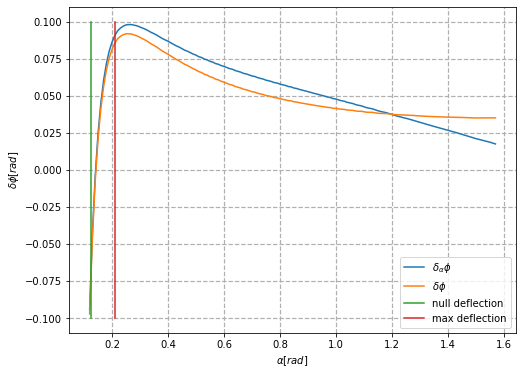}
\caption{$v=1.5$}
\end{subfigure}
\caption{Accuracy of the approximations \eqref{alpha0} and \eqref{alphamax} for several velocities and $r_0=30M$.}\label{approxesa}
\end{figure}
\end{center}

Finally, concerning the observable sizes $R_0$ and $R_{\rm max}$ of the null and maximal deflection rings, we have $R_0=r_0\tan\alpha_0$ and $R_{\rm max}=r_0\tan\alpha_{\rm max}$. For an orbit of sup-photon type with $r_0\gg0$, we may use the relation \eqref{tanalpha} and the estimates \eqref{Rm0} and \eqref{Rm1} to obtain
\begin{equation}\label{R0}
R_0=\frac{9\pi M}{16}\left(1+\frac{M}{r_0}\right)(v-1)^{-1}+{M}\left(\frac{64}{9\pi}-\frac{21\pi}{32}\right)-\frac{M^2}{r_0}\left(\frac{57\pi}{32}-\frac{64}{9\pi}\right)+\bigo(v-1,r_0^{-2})
\end{equation}
and
\begin{equation}\label{Rmax}
R_{\rm max}=\frac{9\pi M}{8}\left(1+\frac{M}{r_0}\right)(v-1)^{-1}-{M}\left(\frac{21\pi}{16}-\frac{32}{3\pi}\right)-\frac{M^2}{r_0}\left(\frac{57\pi}{16}-\frac{32}{3\pi}\right)+\bigo(v-1,r_0^{-2}).
\end{equation}
Observe that this is consistent with the approximations \eqref{alpha0} and \eqref{alphamax} since, at fixed velocity, we have $\lim_{r_0\to\infty}\alpha_0=0$ and thus $R_0=r_0\tan\alpha_0\sim r_0\alpha_0$ and similarly, $R_{\rm max}\sim r_0\alpha_{\rm max}$.

We may proceed similarly for the charged case where $Q\ne0$, but the calculations are a bit tougher. At the lowest order, we have
\begin{equation}\label{charged_rings1}
\rho_{\rm max}\stackrel[v\to1]{}\sim\frac{\pi M}{8}\left(9-\frac{Q^2}{M^2}\right)(v-1)^{-1},
\end{equation}
as well as
\begin{equation}\label{charged_rings2}
\alpha_{\rm max}\stackrel[\substack{r_0\to\infty \\ v\to1}]{}\sim\frac{\pi M}{8r_0}\left(9-\frac{Q^2}{M^2}\right)\left(1+\frac{M}{r_0}\right)(v-1)^{-1}
\end{equation}
and
\begin{equation}\label{charged_rings3}
R_{\rm max}\stackrel[\substack{r_0\to\infty \\ v\to1}]{}\sim\frac{\pi M}{8}\left(9-\frac{Q^2}{M^2}\right)\left(1+\frac{M}{r_0}\right)(v-1)^{-1}.
\end{equation}
Furthermore, just as in the Schwarzschild case, at lowest order, we have $\rho_{\rm max}=2\rho_0$ and similarly for $\alpha_0$ and $R_0$. Note that for these approximations to make sense, the velocity at infinity $v=\sqrt{E^{-2}-1}$ should be real and in view of \S \ref{energyy}, this is ensured only when $Q^2<M^2$.

\section{Geodesic motion in Bertotti--Robinson spacetimes}\label{dynamics_BR}
In this section we investigate the dynamics in Euclidean and Lorentzian Bertotti--Robinson spacetimes and provide in particular analytic solutions for the geodesic equation. For $m,q\in\R$, consider the line element in Binet radial variable
\begin{equation}\label{BR}\tag{BRb}
{\rm d}s^2=Q^2\left[\epsilon \Delta(u){\rm d}\tau^2+\Delta(u)^{-1}{\rm d}u^2+{\rm d}\Omega^2\right],
\end{equation}
where $\Delta(u):=u^2-2mu+q^2$. This is an electro-vacuum solution of Einstein's equation for $A_\mu=\sqrt{-\epsilon}Qu{\rm d}\tau$.

Let $\gamma=(\tau,u,\theta,\phi)$ be a geodesic with respect to the metric \eqref{BR}, parametrized by an affine parameter $\lambda$. By spherical symmetry, we may assume that $\theta\equiv\pi/2$, in which case the Lagrangian reads
\[\mathcal{L}=\frac12 g_{\mu\mu}(\dot{\gamma}^\mu)^2=\frac{Q^2}{2}\left[\epsilon \Delta(u)\dot{\tau}^2+\Delta(u)^{-1}\dot{u}^2+\dot{\phi}^2\right].\]
Thus, the temporal and angular Euler--Lagrange equations immediately yield constants $J,C\in\R$ such that $\dot{\tau}=C\Delta(u)^{-1}$ and $\dot{\phi}=J$ and so the quantity
\[\mathcal{H}:=2\mathcal{L}=Q^2\left[\Delta(u)^{-1}(\epsilon C^2+\dot{u}^2)+J^2\right]\]
is conserved along $\gamma$. Thus, the proper length $s$ satisfies ${\rm d}s^2=\mathcal{H}{\rm d}\lambda^2$ and we get
\begin{equation}\label{mot_BR}
\left(\frac{{\rm d}u}{{\rm d}s}\right)^2=\frac{\dot{u}^2}{\mathcal{H}}=\Delta(u)\left(\frac{1}{Q^2}-L^2\right)-\epsilon E^2,
\end{equation}
where we denote $L:=J/\sqrt{\mathcal{H}}$ and $E:=C/\sqrt{\mathcal{H}}$. From now on, differentiation with respect to $s$ will be denoted by a dot. Since the right-hand side of \eqref{mot_BR} is quadratic in $u$, we may solve it explicitly.

Fix now an exterior equatorial geodesic $\gamma=(\tau,u,\phi)$ in the spacetime \eqref{BR}, with proper angular momentum $L=\dot{\phi}(0)$ and let $u_0:=u(0)$ and $\dot{u}_0:=\dot{u}(0)$. The Binet component $u$ of $\gamma$, as a function of the proper length $s$, is given by
\begin{equation}\label{analytic_BR}
u(s)=m+(u(0)-m)\cos(s\sqrt{\ell})+\dot{u}(0)\frac{\sin(s\sqrt{\ell})}{\sqrt{\ell}},\text{ where }\ell:=L^2-1/Q^2.
\end{equation}
Moreover, the only circular orbit has $u\equiv m$.

Before establishing \eqref{analytic_BR}, we observe the following consequences:
\begin{itemize}
\item The calculations of the previous proof allow to re-write the motion equation \eqref{mot_BR} as
\[\left(\frac{{\rm d}u}{{\rm d}s}\right)^2=(1/Q^2-\dot{\phi}_0^2)(u-u_0)(u+u_0-2m)+\dot{u}_0^2.\]
In particular, the proper spatial dynamics does not depend on the parameter $q$, nor on the signature of the metric \eqref{BR}. Moreover, the translated variable $v:=u-m$ satisfies the equation
\[\dot{v}^2+\ell v^2=\ell(u_0-m)^2+\dot{v}_0^2.\]
Therefore, in the non-degenerate case $\ell\ne0$, the phase portrait in Binet variable is either an ellipse if $\ell>0$, and a hyperbola otherwise.
\item A non-circular geodesic $\gamma=(\tau,u,\theta,\phi)$ with angular momentum $L=\dot{\phi}(0)$ is bounded if and only if $Q^2L^2>{1}$, in which case it has periodic radial component with proper period 
\[\omega_\gamma=\frac{2\pi Q}{\sqrt{Q^2L^2-1}}=\frac{2\pi}{L}+\bigo\left(\frac{1}{Q^2L^3}\right).\]
When $Q^2L^2=1$, the geodesic is affine in proper length and when $Q^2L^2<1$, we have
\[u(s)\stackbin[s\to\infty]{}=\bigo\left(e^{s\sqrt{1-Q^2L^2}}\right),\]
so the Binet variable $u$ blows-up exponentially in proper length.
\item There are exterior circular orbit only when $q^2>m^2$, in which case the only such orbit has $r=1/m$. In particular, there are no circular orbits in the original Bertotti--Robinson space \cite{bertotti59,robinson59}.
\item Any (exterior) Bertotti--Robinson spacetime is geodesically complete in Binet variable.
\end{itemize}

Back to the derivation of \eqref{analytic_BR}, observe first that, up to dilating the affine parameter, we may assume that $\left.\left(\frac{{\rm d}\tau}{{\rm d}\lambda}\right)\right|_{\lambda=0}=1$. Then, we have
\[\mathcal{H}=Q^2\left[\frac{\epsilon\left(\Delta(u_0)\left.\left(\frac{{\rm d}\tau}{{\rm d}\lambda}\right)\right|_{\lambda=0}\right)^2+\left.\left(\tfrac{{\rm d}u}{{\rm d}\lambda}\right)\right|_{\lambda=0}^2}{\Delta(u_0)}+J^2\right]=Q^2\left[\epsilon \Delta(u_0)+\mathcal{H}\frac{\dot{u}_0^2+\Delta(u_0)L^2}{\Delta(u_0)}\right],\]
implying
\[\mathcal{H}=\frac{\epsilon Q^2\Delta(u_0)^2}{\Delta(u_0)(1-Q^2L^2)-Q^2\dot{u}_0^2}~\Longrightarrow~E^2=\frac{C^2}{{\mathcal{H}}}=\frac{\Delta(u_0)(1-Q^2L^2)-Q^2\dot{u}_0^2}{\epsilon Q^2}=-\frac{\ell \Delta(u_0)+\dot{u}_0^2}{\epsilon}\]
and thus the equation \eqref{mot_BR} becomes
\begin{equation}\label{new_mot_BR}
\dot{u}^2=\ell(\Delta(u_0)-\Delta(u))+\dot{u}_0^2=\ell(u_0-u)(u_0+u-2m)+\dot{u}_0^2.
\end{equation}
If $\dot{u}_0=0$, then the technical lemma from Appendix \ref{goubelix} applied to $y=u-m$, $\beta=u_0-m$ and $\alpha=-\ell$ leads to the stated expression for $u$. Otherwise, on a neighbourhood of $0$, we have $\dot{u}\ne0$ and differentiating \eqref{mot_BR} with respect to $s$ yields 
\[\ddot{u}=-\frac{\ell \Delta'(u)}{2}=\ell(m-u),\]
a linear ODE of order 2 whose solution reads
\[u(s)=m+(u_0-m)\cos(s\sqrt{\ell})+\dot{u}_0\frac{\sin(s\sqrt{\ell})}{\sqrt{\ell}},\]
and this function indeed satisfies \eqref{new_mot_BR} and is globally defined.

Now, if $\gamma$ is circular, then $\dot{u}_0=0$ and applying the lemma again, we find that $u(s)=m+(u_0-m)\cos(s\sqrt{\ell})$ is constant, so that $\ell=0$ or $u_0=m$. This can also be seen by analysing the potential $V=\sqrt{\ell\Delta}$. But suppose now that $u\ne m$, then $\ell=0$ and since $\dot{u}(0)=0$, we get
\[\ell=0\Longleftrightarrow L^2=\frac{1}{Q^2}\Longleftrightarrow \frac{Q^2\dot{\phi}_0^2}{\mathcal{H}}=1\Longleftrightarrow\Delta(u_0)=0,\]
contradicting the fact that $\gamma$ is exterior.

From \eqref{analytic_BR} we can deduce the analytic expression for the geodesic motion in terms of affine parameter: if the affine parameter $\lambda$ is chosen so that $\dot{\tau}_0=1$, then the motion constants and expressions of $r$ and $\phi$ as functions of $\lambda$ are given as follows:
\begin{subequations}\label{BR_affine}
\begin{align}
&\mathcal{H}=Q^2\left(\frac{\epsilon(1-2mr_0+q^2r_0^2)}{r_0^2}+\frac{\dot{r}_0^2}{r_0^2(1-2mr_0+q^2r_0^2)}+\dot{\phi}_0^2\right)\text{ and }\ell:=\frac{\dot{\phi}_0^2}{\mathcal{H}}-\frac{1}{Q^2}, \label{BR_affine1} \\
&r(\lambda)=\left(m+\left(\frac{1}{r_0}-m\right)\cos(\lambda\sqrt{\ell})-\frac{\dot{r}_0}{r_0^2}\frac{\sin(\lambda\sqrt{\ell})}{\sqrt{\ell\mathcal{H}}}\right)^{-1},\text{ as well as }\phi(\lambda)=\phi_0+\lambda\dot{\phi}_0. \label{BR_affine2}
\end{align}
\end{subequations}
Moreover, the geodesic $\gamma$ has energy $E=\frac{1-2mr_0+q^2r_0^2}{r_0^2\sqrt{\mathcal{H}}}$ and angular momentum $L=\frac{\dot{\phi}_0}{\sqrt{\mathcal{H}}}$.

\section{Implementation of orbits and ray-tracing}

We have developed a package\footnote{\url{https://github.com/arthur-garnier/euclidean_orbits_and_shadows.git}} under Python, for drawing orbits in the spacetimes \eqref{ERN} and \eqref{BRr}, as well as for ray-tracing the first ones, using a standard \textit{backward ray-tracing} method (see for instance \cite{osiris}). As mentioned in \S \ref{weier_form}, the Carlson and Coquereaux--Grossmann--Lautrup algorithms \cite{carlson,coquereaux} are used to compute geodesics in the Reissner--Nordstr\"{o}m instanton.

Let us briefly recall the backward ray-tracing procedure we use for shadowing a spacetime. First, we consider an artificial celestial hemisphere on which we project our original image, seeing it as a portion of its tangent plane parallel to the screen (and on the other side of the black hole). As a projection, we simply choose the standard and widely used \emph{equirectangular projection}, which has the advantage of taking the celestial hemisphere to a square, which we may rescale to fit our image.

Next, for each pixel of the screen, we consider the geodesic for the spacetime (null in the Lorentzian and normalized in the Euclidean case) starting at this point and with velocity directed by the line from the point observer. We then solve the geodesic equations (backwards) and we see if the ray ends in (came from) the black hole or touches the sphere somewhere. If so, the RGB value of the pixel on the screen is given by the value of the landing pixel on the sphere and we carry this process on until every pixel has been worked out. 

To simplify calculations, we make heavy use of the spherical symmetry of the spacetimes considered here: given an initial datum, use a linear rotation to bring the initial velocity (and hence the full orbit) in the plane $\{\theta=\pi/2\}$. Then, we give values to the various constants involved in the expression of the radial geodesic and, instead of computing the full orbit, we simply solve the equation $r=r_{\Sph}$ where $r_{\Sph}$ is the radius of the celestial sphere, using the Weierstrass function. This can be done rather easily, precisely and quickly: we compute some values until we cross the sphere and the first such point is used as an initial value for the Newton method. We finally rotate the result back and find our landing pixel. Thus, no full orbit calculation is required. For more details and illustrations, see \cite{garnier_CQG23}. We should mention however that since, in contrast to the Lorentzian framework, photons (i.e. null geodesics) do not properly exist in the Euclidean world, we have to trace orbits with some prescribed velocity at infinity $v$: an additional input to the program. For each pixel (i.e each particle), we provide its initial position and the direction of its velocity vector, whose norm is adjust to match the velocity $v$.

We finish our discussion by providing some figures illustrating the programs and our results.

First, concerning the orbits drawings, it should be mentioned that, while the functions of the package are designed to draw orbits in 3D, we chose to plot planar orbits here, to make the figures more readable. In each case, the mass is set to unity and we vary the charge in the different plots, distinguishing between the sub-extremal and sup-extremal cases for the charge. In each figure, the legend gives the values of the energy and angular momentum of each displayed orbit.

The Figure \ref{ERN_orbs} depicts some orbits in the spacetime \eqref{ERN}. In particular, we illustrate the results from \S \ref{energyy}: when there is a horizon, the squared energy is smaller than unity and there are circular orbits with arbitrary energy when we have a naked singularity. Notice also the presence of bounded orbits in the sup-extremal case.

In the Figures \ref{BR_orbs1} and \ref{BR_orbs2}, we show some Bertotti--Robinson orbits, with $\epsilon=\pm1$, $Q=1/2$ and $m=1$. The Figure \ref{BR_orig} displays orbits in the original Bertotti--Robinson spaces where $m=q=0$.

Finally, some particular unstable orbits may be observed in the sup-extremal Euclidean cases. In the Figure \ref{flowers} are displayed some ``flower orbits'' in a horizon-less ERN spacetime, while in the Figure \ref{stars}, we plot some ``star orbits'' in a horizon-less BR space. These can be seen as analogues of ``leaf orbits'' in the usual Lorentzian solutions, see \cite[Figs. 14, 15]{perez-giz-levin}.

Regarding the shadows, just as in \cite{garnier_CQG23}, we use the color grid shown in Figure \ref{zero} as our base picture for the shadows. In the Figure \ref{shadowsERN}, we depict the shadows of an ERN spacetime with various charges. The render time of each figure is approximately 350 seconds, on an 8-core 3.00 GHz CPU with 16 Go of RAM\footnote{We have tested the efficiency of the shadowing program by running it on (randomly generated) images with $10\times10$ to $500\times500$ pixels and we made an exponential regression on the render time. We found 
\[{\rm time}[s]\approx e^{-10.6}\times{\rm pix}^{2.5},\]
with a regression coefficient $r>0.99$.}.

To illustrate our results on deflection, we also wrote a code similar to the shadowing program, but which rather displays the total angular deflection $\delta_\alpha\phi$ (here, $\alpha$ is the angle between the initial velocity vector and the radial direction passing through the common converging point of the ``light rays''. The results are depicted in the Figures \ref{def_euc} and \ref{def_lor}. Observe that, in accordance with the Section \ref{deflection}, the deflection is always positive in the usual Schwarzschild spacetime, with a visible event horizon, while in the Euclidean world, the horizon disappears and the presence of a null and maximal deflection rings is manifest.

Lastly, in the Figure \ref{shadows_nice}, we give the shadows of ERN spacetimes on a celestial background. The original picture is from the NASA and the render time for each shadow is about 2800 seconds, the resolution of the picture being of $1080$ pixels. While the null deflection ring is impossible to pinpoint on such a picture, the maximal deflection ring is still quite noticeable. Observe moreover that, again in contrast to the usual Schwarzschild metric, the horizon-full case $Q^2<M^2$ and the naked singularity case $Q^2>M^2$ produce similar pictures, the most noticeable difference being the size $R_{\rm max}$ of the maximal deflection ring, which decreases as the charge increases. This is in accordance with the lowest order approximation of $R_{\rm max}$ provided in \eqref{charged_rings3}.

\begin{center}
\begin{figure}
\begin{subfigure}[c]{0.32\columnwidth}
\centering
\includegraphics[width=1\linewidth]{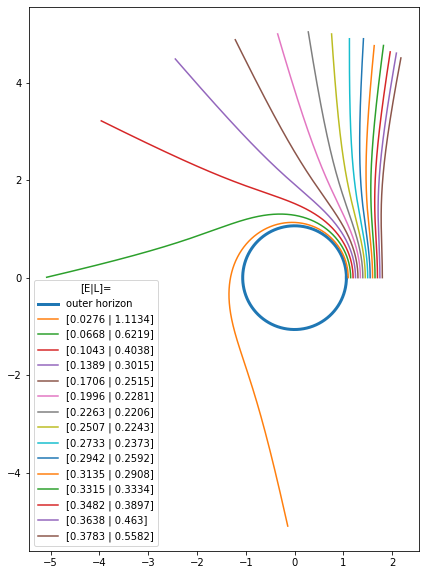}
\caption{Orbits with $Q=0.998$}
\end{subfigure}
\hfill
\begin{subfigure}[c]{0.32\columnwidth}
\centering
\includegraphics[width=0.85\linewidth]{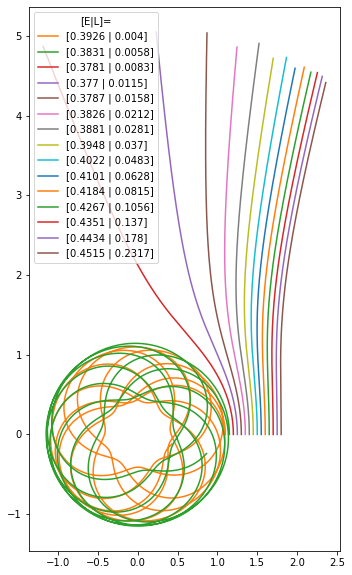}
\caption{Orbits with $Q=1.098$}
\end{subfigure}
\hfill
\begin{subfigure}[c]{0.32\columnwidth}
\centering
\includegraphics[width=1\linewidth]{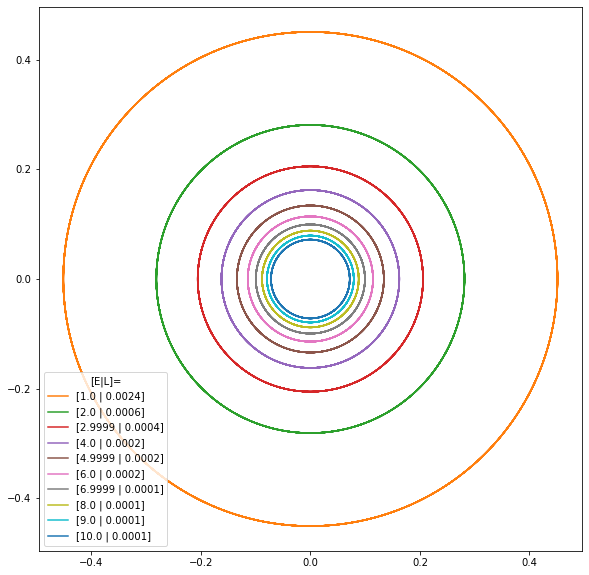}
\caption{Circular orbits with $Q=1.098$}
\end{subfigure}
\caption{Some Euclidean Reissner--Nordstr\"{o}m orbits with $M=1$.}\label{ERN_orbs}
\end{figure}
\end{center}

\begin{center}
\begin{figure}
\begin{subfigure}[c]{0.32\columnwidth}
\centering
\includegraphics[width=0.95\linewidth]{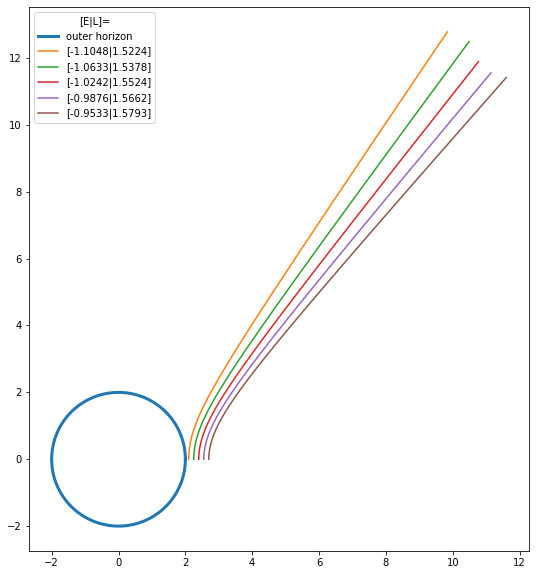}
\caption{$q=0$}
\end{subfigure}
\hfill
\begin{subfigure}[c]{0.32\columnwidth}
\centering
\includegraphics[width=1\linewidth]{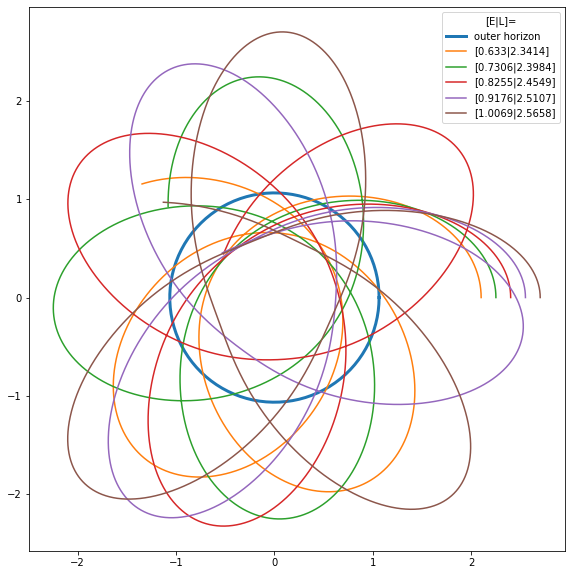}
\caption{$q=0.998$}
\end{subfigure}
\hfill
\begin{subfigure}[c]{0.32\columnwidth}
\centering
\includegraphics[width=1\linewidth]{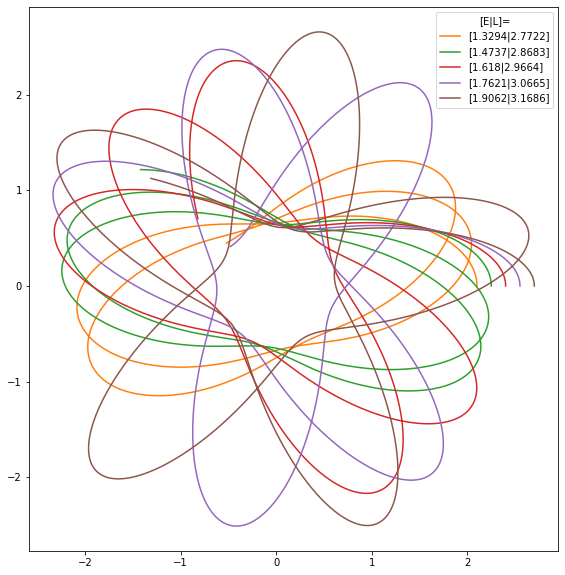}
\caption{$q=1.098$}
\end{subfigure}
\caption{Orbits in Lorentzian Bertotti-Robinson spacetimes with $Q=1/2$ and $m=1$.}\label{BR_orbs1}
\end{figure}
\end{center}

\begin{center}
\begin{figure}
\begin{subfigure}[c]{0.31\columnwidth}
\centering
\includegraphics[width=0.95\linewidth]{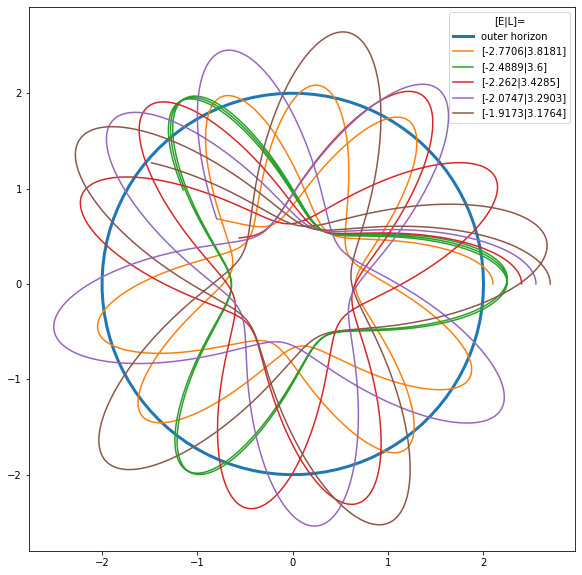}
\caption{$q=0$}
\end{subfigure}
\hfill
\begin{subfigure}[c]{0.31\columnwidth}
\centering
\includegraphics[width=0.96\linewidth]{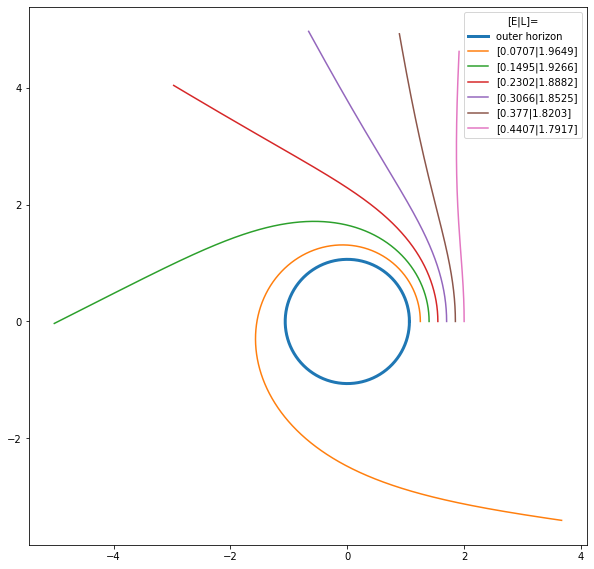}
\caption{$q=0.998$}
\end{subfigure}
\hfill
\begin{subfigure}[c]{0.31\columnwidth}
\centering
\includegraphics[width=1.1\linewidth]{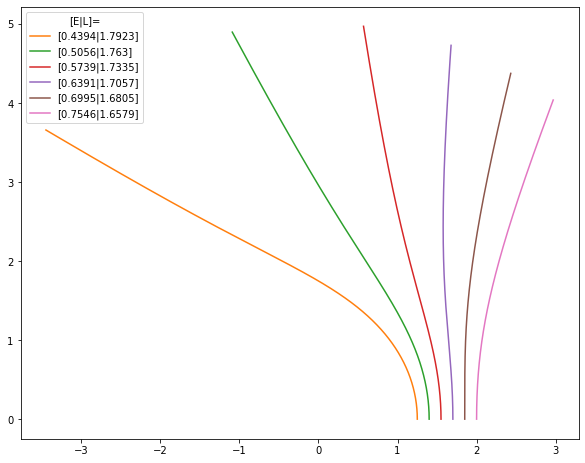}
\caption{$q=1.098$}
\end{subfigure}
\caption{Orbits in Euclidean Bertotti-Robinson spacetimes with $Q=1/2$ and $m=1$.}\label{BR_orbs2}
\end{figure}
\end{center}

\begin{center}
\begin{figure}
\begin{subfigure}[c]{0.49\columnwidth}
\centering
\includegraphics[width=1\linewidth]{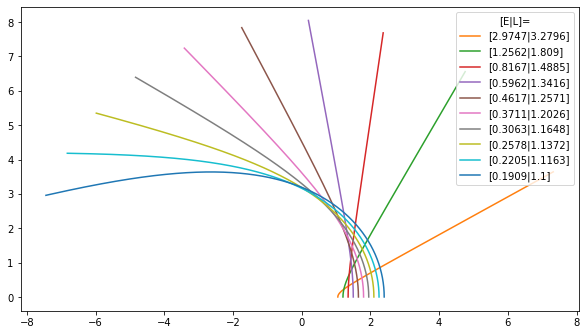}
\caption{$\epsilon=-1$}
\end{subfigure}
\hfill
\begin{subfigure}[c]{0.49\columnwidth}
\centering
\includegraphics[width=0.8\linewidth]{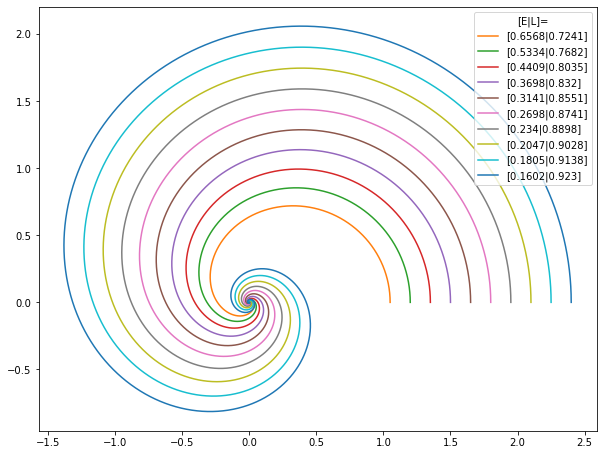}
\caption{$\epsilon=1$}
\end{subfigure}
\caption{Orbits in Bertotti--Robinson spaces with $Q=1$ and $m=q=0$.}\label{BR_orig}
\end{figure}
\end{center}

\begin{center}
\begin{figure}
\begin{subfigure}[c]{0.15\columnwidth}
\centering
\includegraphics[width=1\linewidth]{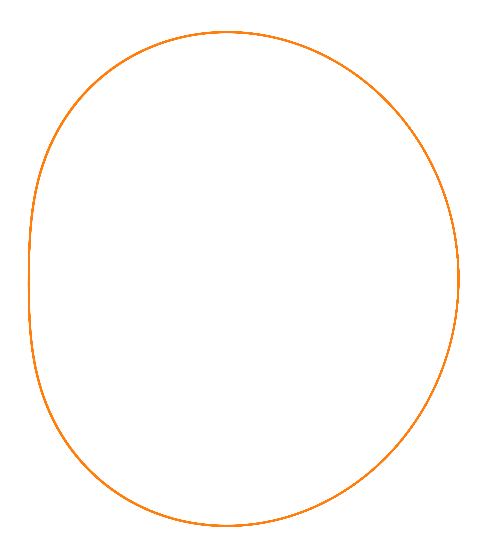}
\end{subfigure}
\hfill
\begin{subfigure}[c]{0.15\columnwidth}
\centering
\includegraphics[width=1\linewidth]{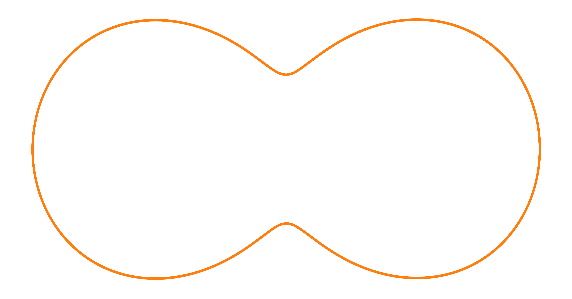}
\end{subfigure}
\hfill
\begin{subfigure}[c]{0.15\columnwidth}
\centering
\includegraphics[width=1\linewidth]{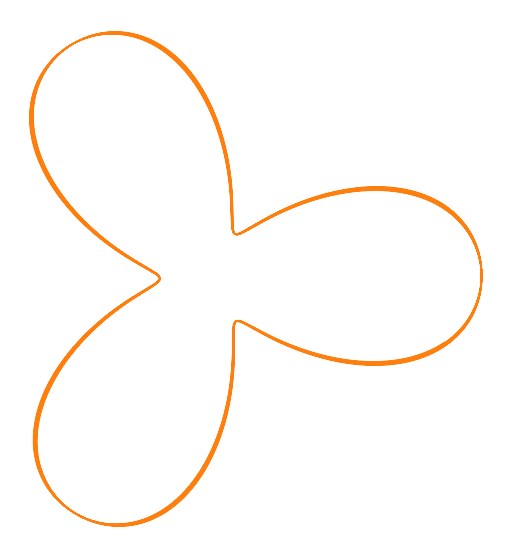}
\end{subfigure}
\hfill
\begin{subfigure}[c]{0.15\columnwidth}
\centering
\includegraphics[width=1\linewidth]{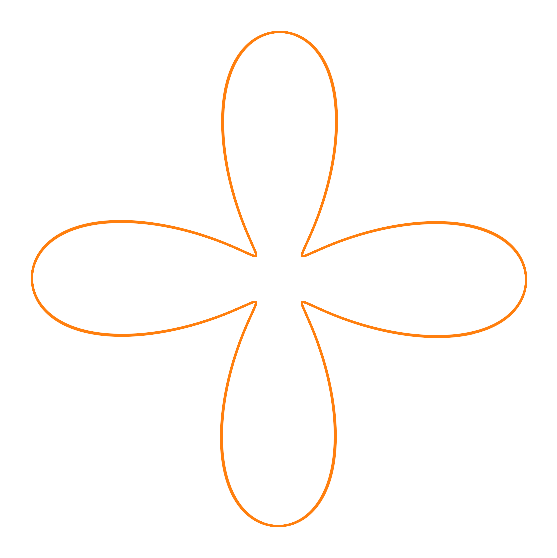}
\end{subfigure}
\hfill
\begin{subfigure}[c]{0.15\columnwidth}
\centering
\includegraphics[width=1\linewidth]{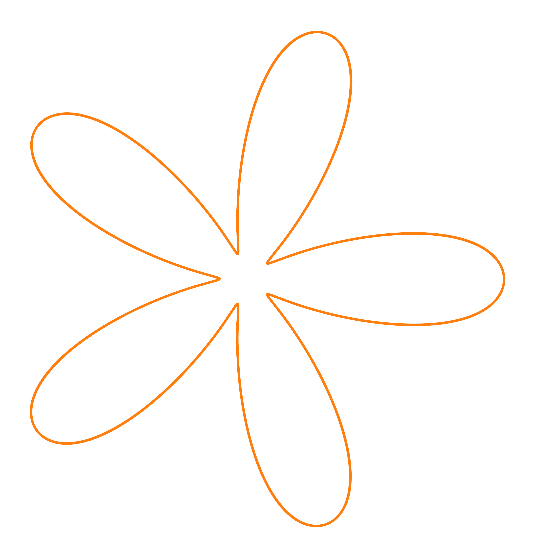}
\end{subfigure}
\hfill
\begin{subfigure}[c]{0.15\columnwidth}
\centering
\includegraphics[width=1\linewidth]{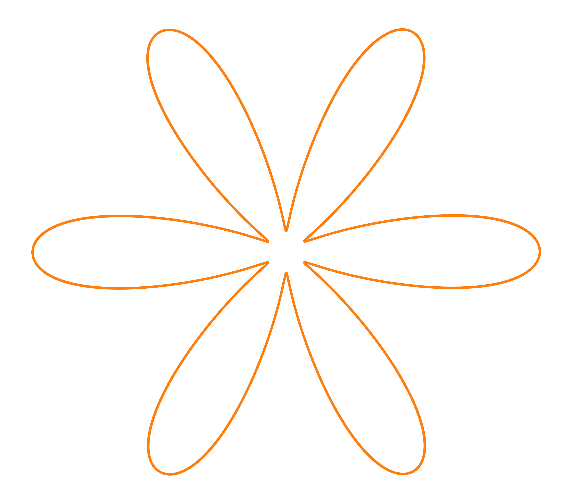}
\end{subfigure}
\caption{Flower orbits in ERN space with $M=1$ and $Q=1.098$.}\label{flowers}
\end{figure}
\end{center}

\begin{center}
\begin{figure}
\begin{subfigure}[c]{0.15\columnwidth}
\centering
\includegraphics[width=0.8\linewidth]{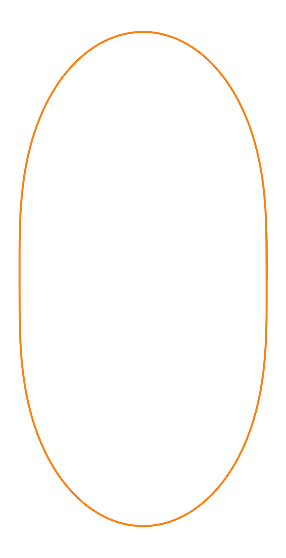}
\end{subfigure}
\hfill
\begin{subfigure}[c]{0.15\columnwidth}
\centering
\includegraphics[width=1\linewidth]{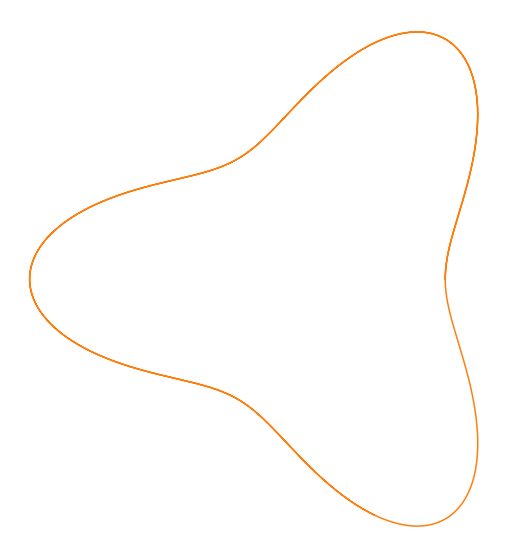}
\end{subfigure}
\hfill
\begin{subfigure}[c]{0.15\columnwidth}
\centering
\includegraphics[width=1\linewidth]{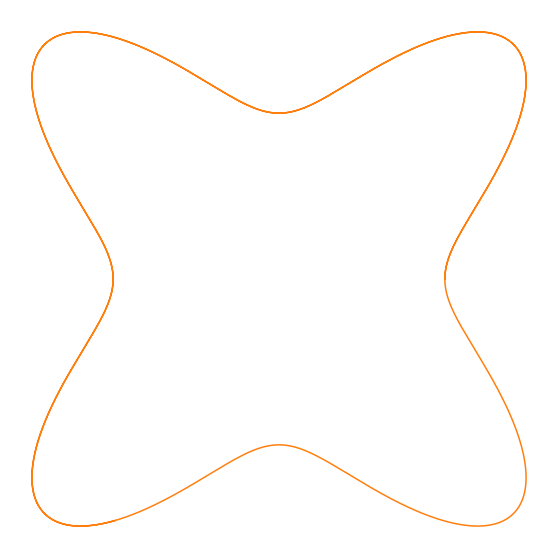}
\end{subfigure}
\hfill
\begin{subfigure}[c]{0.15\columnwidth}
\centering
\includegraphics[width=1\linewidth]{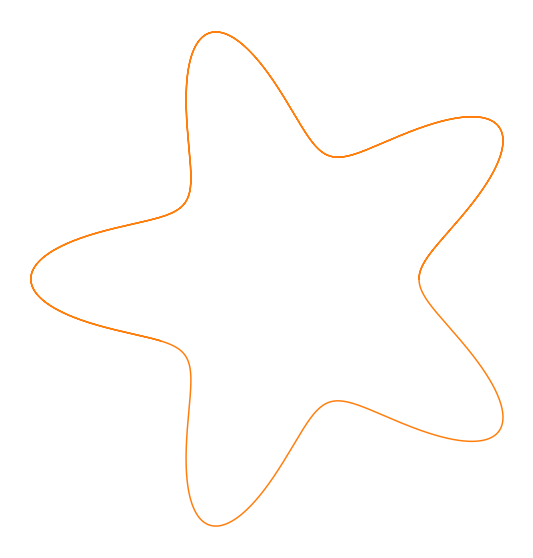}
\end{subfigure}
\hfill
\begin{subfigure}[c]{0.15\columnwidth}
\centering
\includegraphics[width=1\linewidth]{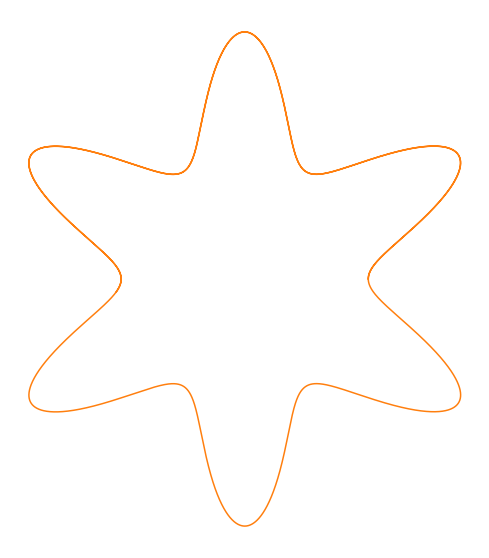}
\end{subfigure}
\hfill
\begin{subfigure}[c]{0.15\columnwidth}
\centering
\includegraphics[width=1\linewidth]{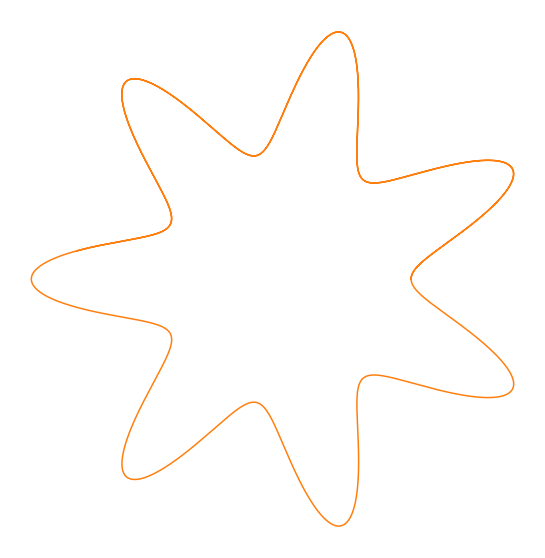}
\end{subfigure}
\caption{Star orbits in BR space with $\epsilon=-1$, $Q=1/2$, $m=1$ and $q=1.098$.}\label{stars}
\end{figure}
\end{center}
\vspace{-3cm}

\begin{center}
\begin{figure}
\begin{subfigure}[c]{0.4\textwidth}
\centering
\includegraphics[scale=0.21]{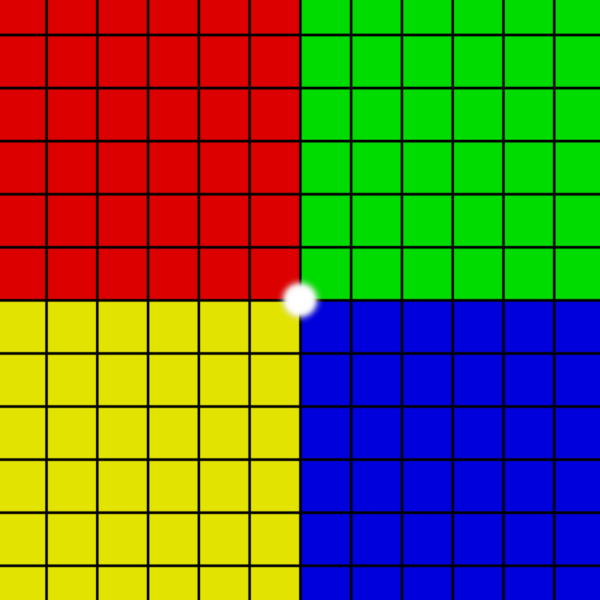}
\caption{The original $600\times600$ pixels color grid.}
\end{subfigure}
\hspace{5mm}
\begin{subfigure}[c]{0.4\textwidth}
\centering
\includegraphics[scale=0.18]{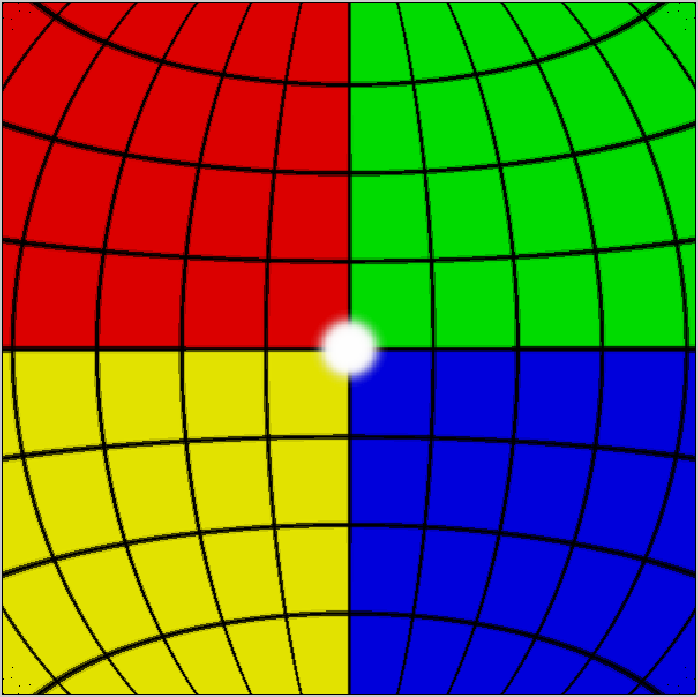}
\caption{Picture obtained with the code ($M=Q=0$, $v=1$).}
\end{subfigure}
\caption{The color grid used in the program.}\label{zero}
\end{figure}
\end{center}

\begin{center}
\begin{figure}
\begin{subfigure}[c]{0.32\columnwidth}
\centering
\includegraphics[width=1\linewidth]{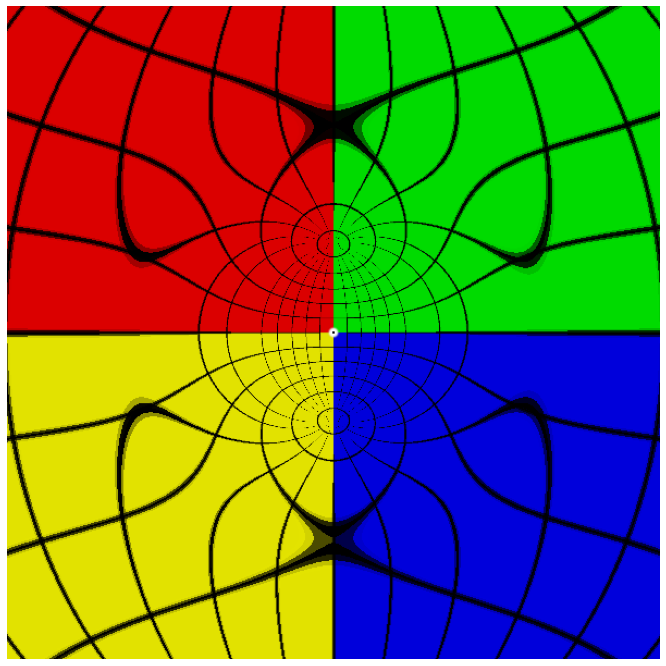}
\caption{$Q=0$}
\end{subfigure}
\hfill
\begin{subfigure}[c]{0.32\columnwidth}
\centering
\includegraphics[width=1\linewidth]{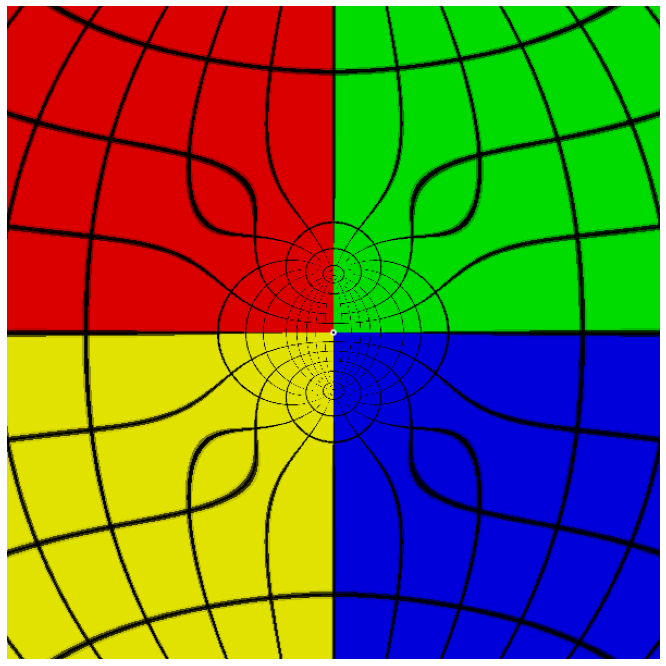}
\caption{$Q=0.998$}
\end{subfigure}
\hfill
\begin{subfigure}[c]{0.32\columnwidth}
\centering
\includegraphics[width=1\linewidth]{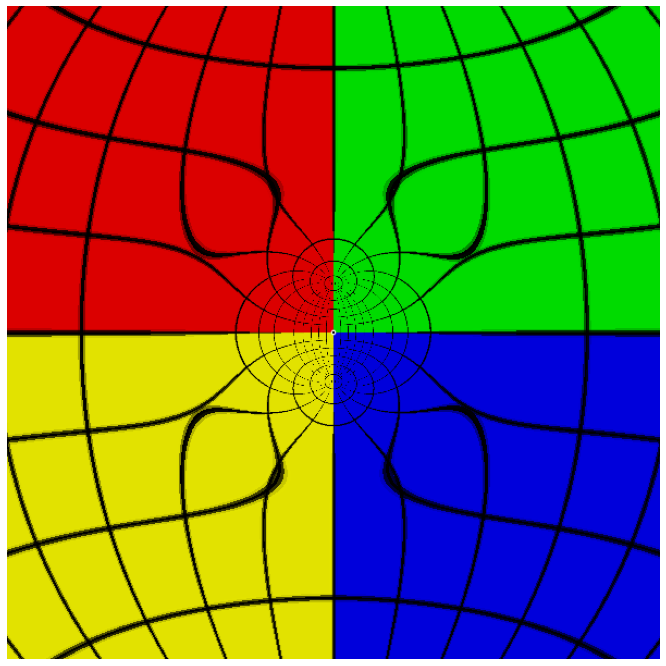}
\caption{$Q=1.098$}
\end{subfigure}
\caption{Shadows of ERN spacetimes with $M=1$ and $v=1.5$.}\label{shadowsERN}
\end{figure}
\end{center}

\begin{center}
\begin{figure}[h!]
\begin{subfigure}[c]{0.32\columnwidth}
\centering
\includegraphics[width=1\linewidth]{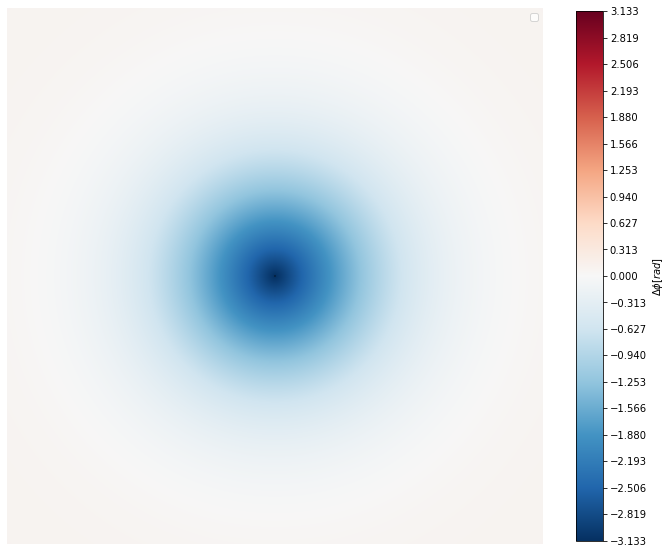}
\caption{$v=1.5$}
\end{subfigure}
\hfill
\begin{subfigure}[c]{0.32\columnwidth}
\centering
\includegraphics[width=1\linewidth]{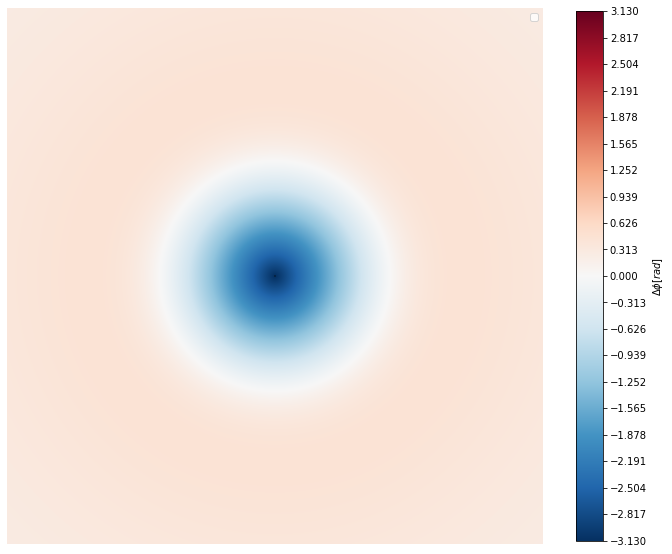}
\caption{$v=2.5$}
\end{subfigure}
\hfill
\begin{subfigure}[c]{0.32\columnwidth}
\centering
\includegraphics[width=1\linewidth]{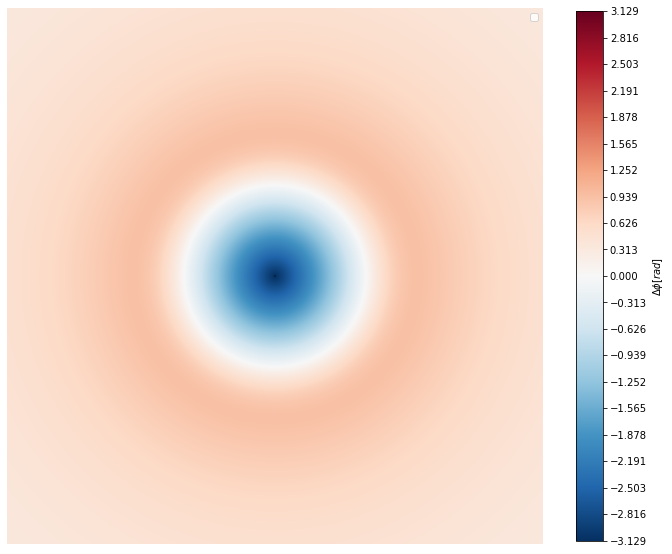}
\caption{$v=4$}
\end{subfigure}
\\
\begin{subfigure}[c]{0.32\columnwidth}
\centering
\includegraphics[width=1\linewidth]{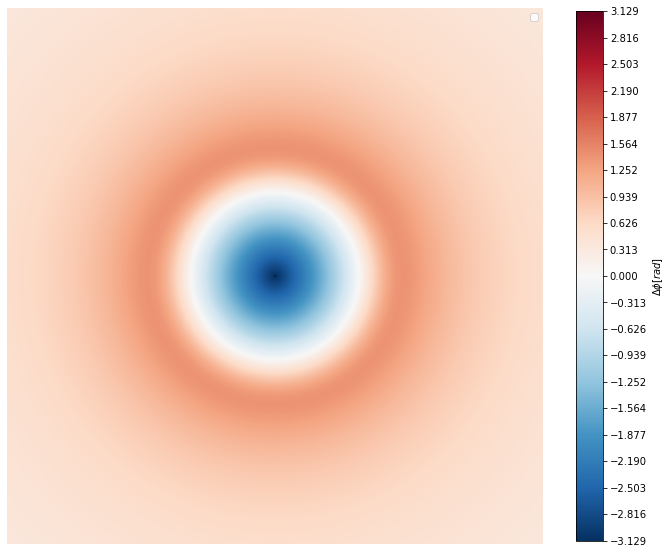}
\caption{$v=6$}
\end{subfigure}
\hfill
\begin{subfigure}[c]{0.32\columnwidth}
\centering
\includegraphics[width=1\linewidth]{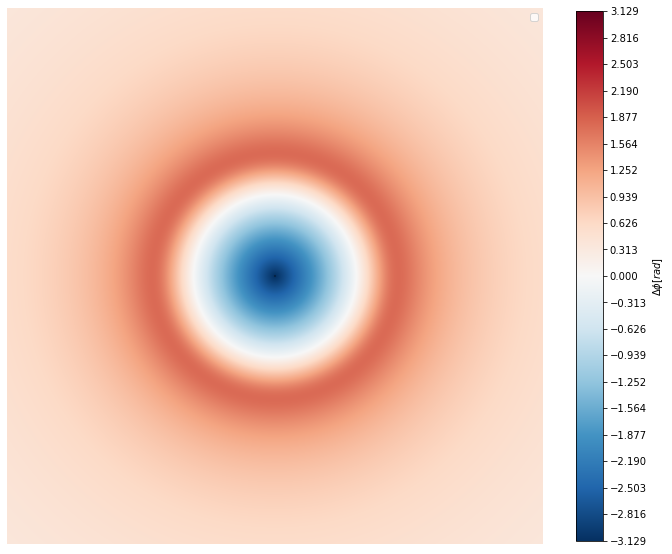}
\caption{$v=8$}
\end{subfigure}
\hfill
\begin{subfigure}[c]{0.32\columnwidth}
\centering
\includegraphics[width=1\linewidth]{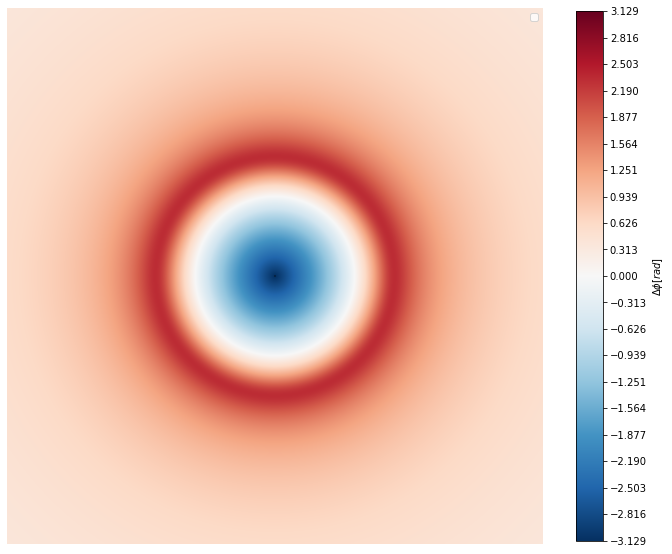}
\caption{$v=12$}
\end{subfigure}
\caption{Deflection maps in Euclidean Schwarzschild spacetime ($M=1$) with several velocities at infinity.}\label{def_euc}
\end{figure}
\end{center}

\begin{center}
\begin{figure}
\begin{subfigure}[c]{0.32\columnwidth}
\centering
\includegraphics[width=1\linewidth]{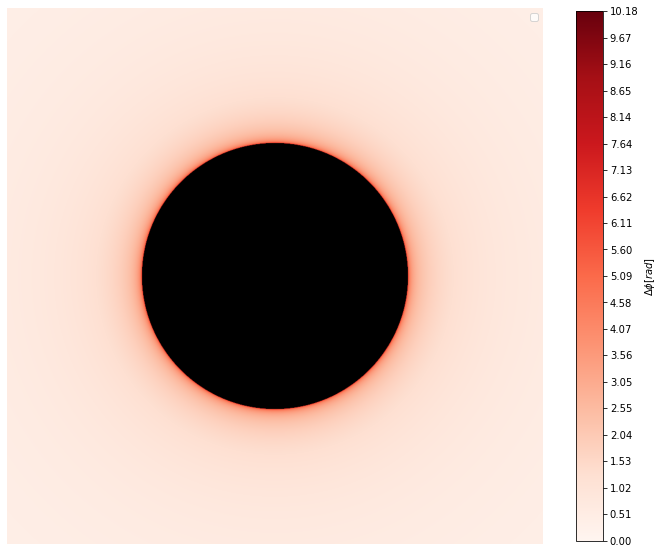}
\caption{$v=0.72$}
\end{subfigure}
\hfill
\begin{subfigure}[c]{0.32\columnwidth}
\centering
\includegraphics[width=1\linewidth]{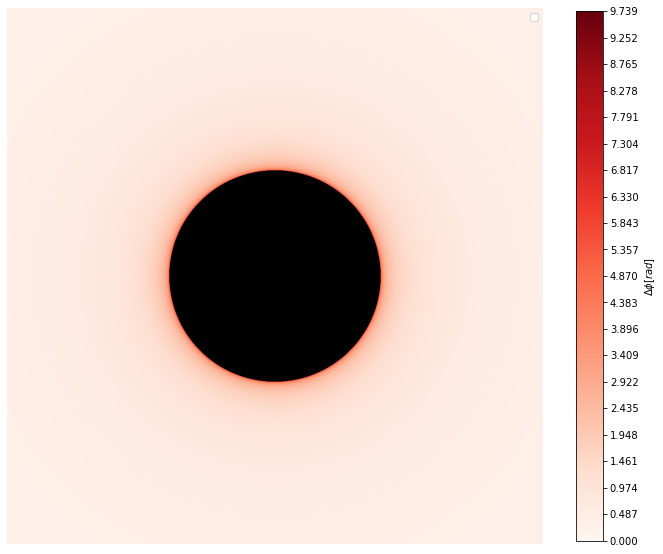}
\caption{$v=0.998$}
\end{subfigure}
\hfill
\begin{subfigure}[c]{0.32\columnwidth}
\centering
\includegraphics[width=1\linewidth]{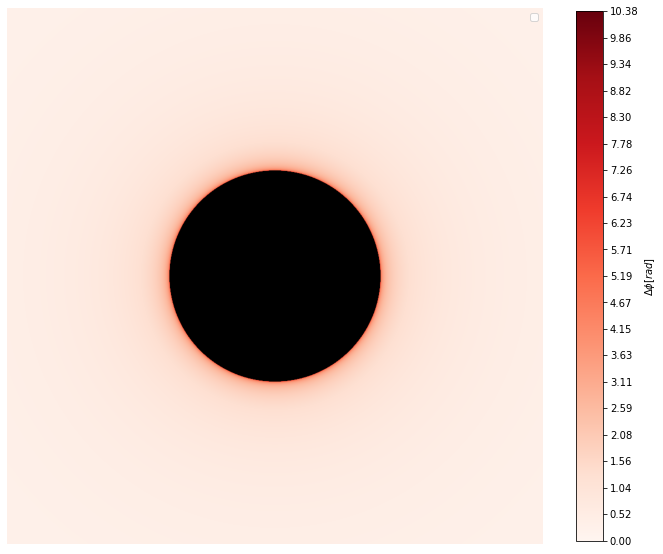}
\caption{$v=1$ (photon)}
\end{subfigure}
\caption{Deflection maps in Lorentzian Schwarzschild spacetime ($M=1$) with several velocities at infinity.}\label{def_lor}
\end{figure}
\end{center}

\begin{center}
\begin{figure}[h!]
\begin{subfigure}[c]{0.42\columnwidth}
\centering
\includegraphics[width=0.98\linewidth]{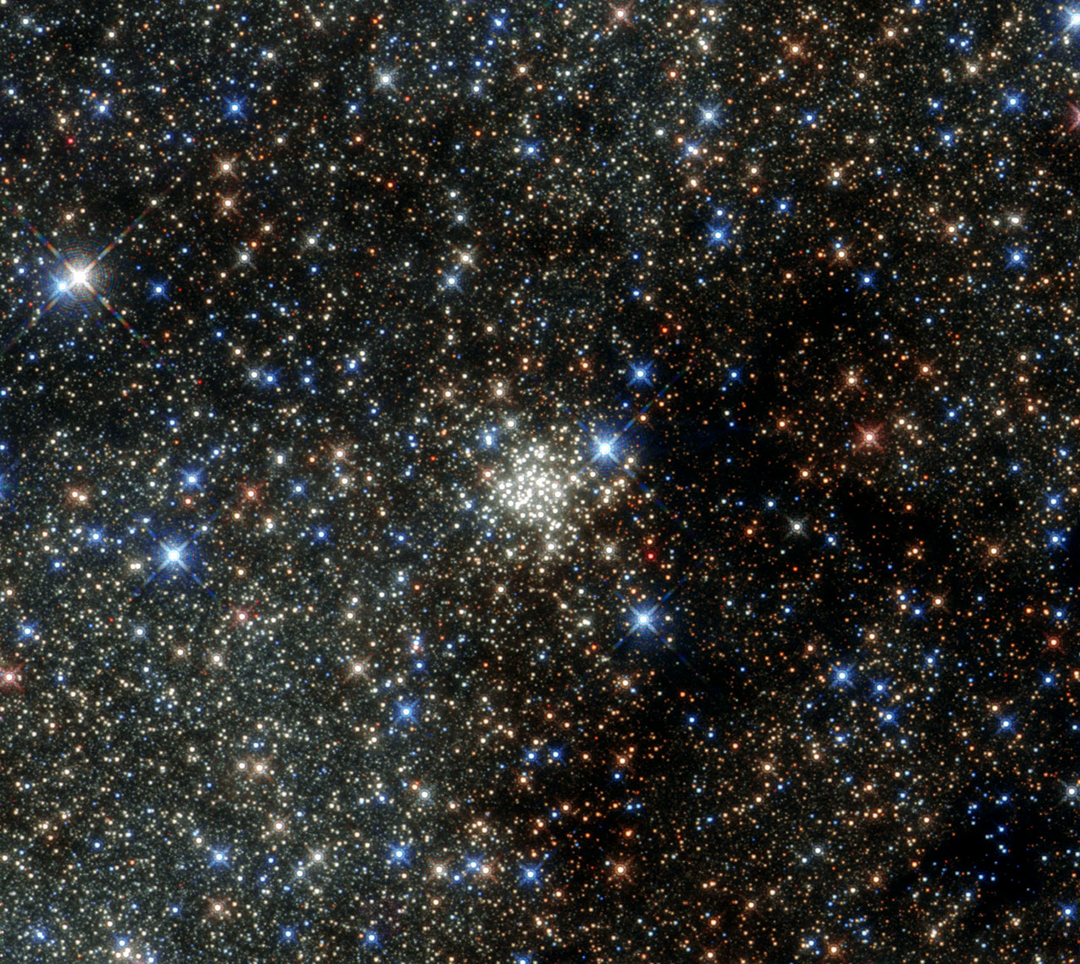}
\caption{The base picture (Arches cluster).}
\end{subfigure}
\hfill
\begin{subfigure}[c]{0.42\columnwidth}
\centering
\includegraphics[width=1\linewidth]{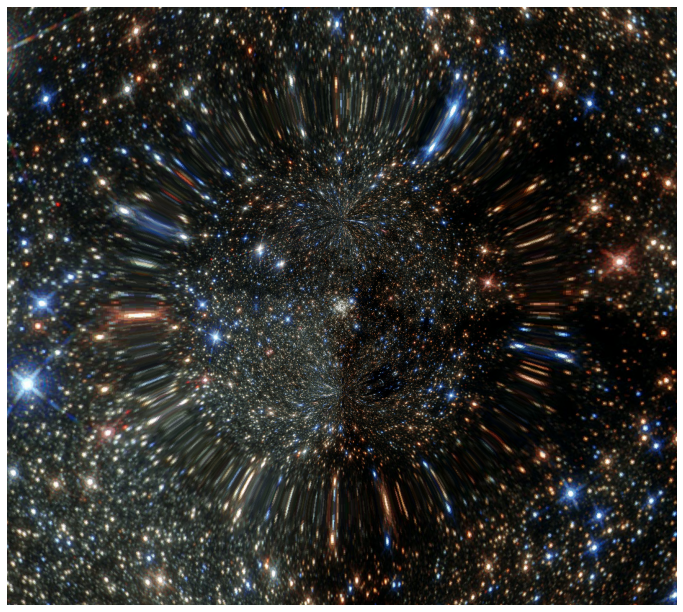}
\caption{$Q=0$}
\end{subfigure}
\\
\begin{subfigure}[c]{0.42\columnwidth}
\centering
\includegraphics[width=1\linewidth]{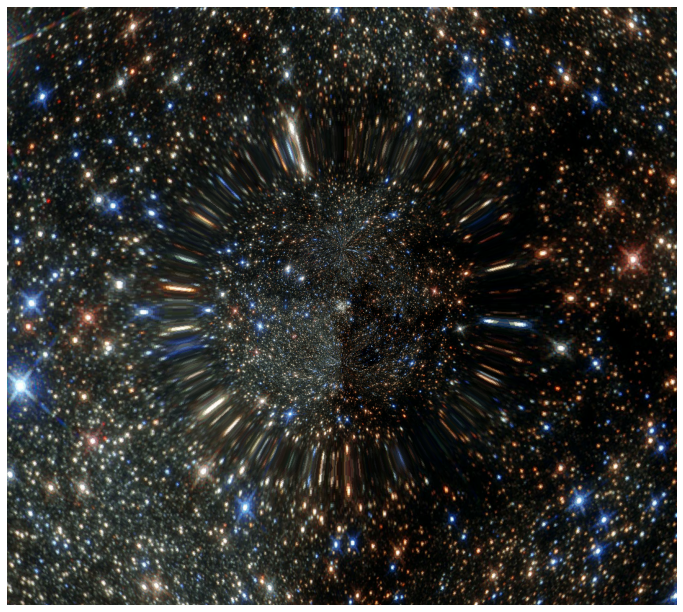}
\caption{$Q=0.998$}
\end{subfigure}
\hfill
\begin{subfigure}[c]{0.42\columnwidth}
\centering
\includegraphics[width=1\linewidth]{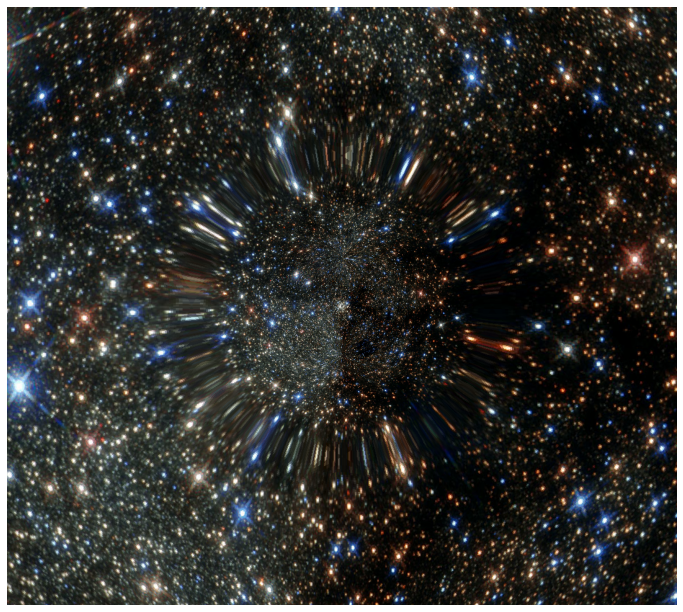}
\caption{$Q=1.098$}
\end{subfigure}
\caption{Shadows of ERN spacetimes with $M=1$ and $v=1.5$ on a celestial background (original image: \url{https://images.nasa.gov/details/GSFC_20171208_Archive_e000717}).}\label{shadows_nice}
\end{figure}
\end{center}

\newpage
\section{Conclusions}

\subsection{Summary}
In this work, we study the geodesic motion in the two different kinds of spherically symmetric electro-vacuum Euclidean solutions of Einstein's equation with complex vector potential, namely the Reissner--Nordstr\"{o}m and Bertotti--Robinson-like instantons. More precisely, after proving that these are indeed the only two possible such solutions (and that these two are incompatible), we derive the motion equations and give the main properties of the test-particle orbits, using their motion constants such as the energy and angular momentum.

We start with the Euclidean Reissner--Nordstr\"{o}m solution, for which we prove that if the spacetime has a horizon, then all exterior (non-constant) geodesics have bounded energy and, otherwise, there are circular orbits with arbitrary energy. In particular, the spacetime features a horizon if and only if there are no bounded orbits. This generalizes the results from \cite{battista-esposito22} on the Euclidean Schwarzschild solution and shows that the differences in the dynamics of the Lorentzian and Euclidean solutions is deeper than the apparent mere sign change in the line element. Then, we show that the polar radial motion equation may be put in Weierstrass form, something that we can take advantage of to numerically solve the motion equation. Moreover, this approach shows that, in the special case of the Euclidean Schwarzschild metric, the (polar) phase portrait in Binet variable\footnote{specifically, in an affine transform of the Binet variable} describes a real elliptic curve. This allows us to derive a new proof of the results from \cite{battista-esposito22}, using the elementary geometry of the curve. Furthermore, we observe that this elliptic curve is always disconnected, while in the usual Lorentzian Schwarzschild solution, it can be either connected or not.

We then study the gravitational bending of geodesics in Euclidean Schwarzschild spacetimes. More precisely, given a geodesic with closest approach radius $r=r_{\rm min}$ and velocity at infinity $v^2=\lim_{r\to\infty}({\rm d}r/{\rm d}\tau)^2$, we first provide an analytic formula for the deflection angle $\delta\phi$ that occurs between the two asymptotic directions of the geodesic, in terms of Carlson's elliptic integrals, which are used to numerically compute the deflection. Then, we give approximations for $\delta\phi$ when $r_{\rm min}\to\infty$ using the previous results of \cite{accioly-ragusa}. We observe that when $v>1$, there are particular values $\rho_0$ and $\rho_{\rm max}$ for which $\delta\phi=0$ for $r_{\rm min}=\rho_0$ and $\delta\phi$ is maximal for $r_{\rm min}=\rho_{\rm max}$. Estimates for $\rho_0$ and $\rho_{\rm max}$ are provided when $v\gtrsim1$. We then do the same for the deflection $\delta_\alpha\phi$ of geodesics passing through a point at fixed radius $r_0$ and with varying angle $0<\alpha<\pi/2$, modelling orbits coming from infinity to the eye of an observer. We give approximations, when $r_0\gg2M$ and $v\gtrsim1$, for the critical angles $\alpha_0$ and $\alpha_{\rm max}$ corresponding to the null and maximal deflection rings, as well as for the observable size of these rings. We use numerical evaluations of the analytic formulas to check the accuracy of our estimates.

Observe that the inequality $\delta\phi<0$ for close perihelia means that the central mass \textit{repels} such test-particles, while it attracts them at bigger perihelia. In particular, at fixed initial radius and velocity at infinity $v>1$, there is a critical angle $0<\alpha_0<\pi/2$ such that $\delta_{\alpha_0}\phi=0$ (see Figure \ref{brushb} for instance), meaning that the corresponding curve in the $(r,\theta,\phi)$-space is undistinguishable from that of a flat geodesic (i.e. a straight line).

Next, we do the same for the Bertotti--Robinson solution for which, after a technical lemma on ordinary differential equations, we give a general analytic solution, in terms of (hyperbolic) trigonometric functions. In particular, orbits are either periodic or unbounded and there is a unique circular orbit, which is exterior exactly when the metric has no horizon.

Finally, we provide some details on a Python code\footnote{available at \url{https://github.com/arthur-garnier/euclidean_orbits_and_shadows.git}}, designed to plot orbits in the aforementioned spacetimes, as well as to draw shadows of the Euclidean Reissner--Nordstr\"{o}m family, by the usual backward ray-tracing method. The numerical computation takes advantage of the Weierstrass form of the polar equation, coupled with the well-known Carlson \cite{carlson} and Coquereaux--Grossmann--Lautrup \cite{coquereaux} algorithms, to produce an efficient and rather fast code. This method is the same as the one used by the author in \cite{garnier_CQG23} to ray-trace Reissner--Nordstr\"{o}m--(anti)de Sitter black holes. We illustrate the code and our results by providing some figures.

\subsection{Discussion and perspectives}
In the Figures \ref{flowers} and \ref{stars}, we depict some particularly shaped unstable periodic orbits, that can be seen as analogues of the so-called ``leaf orbits'' that occur in the Lorentzian framework \cite{perez-giz-levin}. However, we have no theoretical interpretation for these orbits and this could be a matter of interest for future works.

In \S \ref{schwarzschild}, we have observed that a Euclidean Schwarzschild orbit can mathematically be interpreted as a Lorentzian space-like Schwarzchild orbit with purely imaginary energy. To decide whether this is merely a mathematical curiosity or is a manifestation of a deeper physical phenomenon is beyond the knowledge of the author and could be an interesting subject to investigate as well.

In \S \ref{deflection}, we give approximate values for the null and maximal deflection rings of a sup-photon type particle (i.e. with velocity at infinity $v\gtrsim1$). Moreover, we have numerically observed that such particular rings occur exactly when $v>1$. It would be suitable to provide a rigorous proof of this fact.

Besides the theoretical results, one aim of this work is to provide an open, transparent, user-friendly and customizable Python code to plot orbits and draw shadows of spherically symmetric (asymptotically flat) electro-vacuum instantons. This has a negative side though, which is that the code is under-optimal and still slow in comparison to other widely-used ray-tracing codes, such as \texttt{GYOTO}, \texttt{GRay} or the more recent \texttt{OSIRIS} \cite{gyoto,GRay,osiris}. Therefore, it would be interesting to improve it by using a GPU parallelization, for instance.

As in \cite{garnier_CQG23}, the shadowing program is designed to work with a common plane background image and in order to avoid distortions, we took the compromise of projecting the image on a celestial hemisphere. However, to be able to catch any possible particle that is ray-traced, we had to take the whole celestial sphere into account and we arbitrarily chose to project a mirrored version of the original image onto it. Thus, the code could be completed by producing a panoramic version of it.

\subsection*{Acknowledgments}
I am much grateful to Emmanuele Battista for bringing the subject to my attention and for many fascinating discussions, as well as for his careful reading of the present work and his useful comments.

I also warmly thank the referee for their useful and accurate comments and suggestions. Finally, a big thanks goes out to Olivier Goubet for communicating to me the nice connectedness argument proving the lemma \ref{goubelix}.

\subsection*{Code availability statement} This manuscript has associated code in a data repository. The code generated and analysed during the current study is available in the GitHub repository \url{https://github.com/arthur-garnier/euclidean_orbits_and_shadows}.

\begin{appendix}
\addcontentsline{toc}{part}{Appendices}
\section{Proof of the unicity result from \S \ref{unicity_section}}\label{proof_unicity}
Let ${\rm d}s^2=g_{\mu\nu}{\rm d}x^\mu{\rm d}x^\nu$ be a spherically symmetric solution of the field equation. As proved in \cite{hoffmann32}, a Lorentzian spherically symmetric solution of the electro-vacuum Einstein equations is necessarily static and the proof can be adapted verbatim to the Euclidean case. Therefore, the metric is diagonal: ${\rm d}s^2=g_{\mu\mu}({\rm d}x^\mu)^2$ and moreover, the functions $g_{\tau\tau}$ and $g_{rr}$ only depend on $r$. Furthermore, since the restriction to a hypersurface with constant coordinates $(\tau,r)$ must be an $SO(3)$-invariant metric on $\Sph^2$, there is a positive smooth function $\rho : ]r_+,+\infty[\longto\R^*_+$ such that $g_{\theta\theta}{\rm d}\theta^2+g_{\phi\phi}{\rm d}\phi^2=\rho(r)r^2{\rm d}\Omega^2$. Thus the metric may be reduced to the form
\[{\rm d}s^2=\rho(r)\left[u(r){\rm d}\tau^2+v(r){\rm d}r^2+r^2{\rm d}\Omega^2\right].\]

The Christoffel symbols ${\Gamma^\alpha}_{\mu\nu}=\tfrac12g^{\alpha\beta}(g_{\beta\mu,\nu}+g_{\beta\nu,\mu}-g_{\mu\nu,\beta})$ of this metric read
\[\left\{\begin{array}{ll}
{\Gamma^\tau}_{\mu\nu}=\frac{(u\rho)'}{2u\rho}\left(\begin{smallmatrix}0 & 1 & 0 & 0 \\ 1 & 0 & 0 & 0 \\ 0 & 0 & 0 & 0 \\ 0 & 0 & 0 & 0\end{smallmatrix}\right), & {\Gamma^r}_{\mu\nu}=\frac{1}{2v\rho}\left(\begin{smallmatrix}-(u\rho)' & 0 & 0 & 0 \\ 0 & (v\rho)' & 0 & 0 \\ 0 & 0 & -r(r\rho'+2\rho) & 0 \\ 0 & 0 & 0 & -r\sin^2\theta(r\rho'+2\rho)\end{smallmatrix}\right), \\[2em]
{\Gamma^\theta}_{\mu\nu}=\left(\begin{smallmatrix}0 & 0 & 0 & 0 \\ 0 & 0 & \frac{r\rho'+2\rho}{2r\rho} & 0 \\ 0 & \frac{r\rho'+2\rho}{2r\rho} & 0 & 0 \\ 0 & 0 & 0 & -\sin\theta\cos\theta\end{smallmatrix}\right), & {\Gamma^\phi}_{\mu\nu}=\left(\begin{smallmatrix} 0 & 0 & 0 & 0 \\ 0 & 0 & 0 & \frac{r\rho'+2\rho}{2r\rho} \\ 0 & 0 & 0 & \cotan\theta \\ 0 & \frac{r\rho'+2\rho}{2r\rho} & \cotan\theta & 0\end{smallmatrix}\right),\end{array}\right.\]
where we have dropped the dependence in the variable $r$ for simplicity.

On the other hand, the electro-magnetic field tensor associated to the potential $A_\mu$ is given by
\[F_{\mu\nu}=A_{\mu,\nu}-A_{\nu,\mu}=\frac{-iQ}{r^2}\left(\begin{smallmatrix}0 & 1 & 0 & 0 \\ -1 & 0 & 0 & 0 \\ 0 & 0 & 0 & 0 \\ 0 & 0 & 0 & 0\end{smallmatrix}\right)\]
so that the Maxwell equations 
\[0={F^{\mu\nu}}_{;\mu}={F^{\mu\nu}}_{,\mu}+{\Gamma^\mu}_{\mu\lambda}F^{\lambda\nu}+{\Gamma^\nu}_{\mu\lambda}F^{\mu\lambda}=\frac{1}{\sqrt{\det g}}\left(\sqrt{\det g} F^{\mu\nu}\right)_{,\mu}\]
are satisfied for all $\nu\ne\tau$ and we have ${F^{\mu\tau}}_{;\mu}={F^{r\tau}}_{,r}+{\Gamma^\mu}_{\mu r}F^{r\tau}=\frac{iQ(uv)'}{(ruv\rho)^2}$, so that the Maxwell equations hold if and only if there is a constant $k\in\R^*$ such that $v=k/u$. Therefore, the stress-energy tensor $T_{\mu\nu}$ reads
\begin{align*}
T_{\mu\nu}&=\frac{1}{\mu_0}\left(g^{\alpha\beta}F_{\alpha\mu}F_{\beta\nu}-\frac14g_{\mu\nu}F_{\alpha\beta}F^{\alpha\beta}\right)=\frac{1}{4\pi}\left(\frac{Q^2}{2kr^4\rho^2}g_{\mu\nu}-g^{\alpha\alpha}F_{\mu\alpha}F_{\alpha\nu}\right) \\
&=\frac{Q^2}{8\pi r^2\rho}{\rm diag}\left(\frac{-u}{kr^2},\frac{-1}{ur^2},\frac{1}{k},\frac{\sin^2\theta}{k}\right).
\end{align*}
Next, the Ricci tensor $R_{\mu\nu}={\Gamma^\alpha}_{\mu\nu,\alpha}-{\Gamma^\alpha}_{\mu\alpha,\nu}+{\Gamma^\alpha}_{\beta\alpha}{\Gamma^\beta}_{\mu\nu}-{\Gamma^\alpha}_{\nu\beta}{\Gamma^\beta}_{\mu\alpha}$ and Einstein tensor $G_{\mu\nu}=R_{\mu\nu}-\tfrac12Rg_{\mu\nu}$ are diagonal and have
\[\left\{\begin{array}{l}
2kr\rho R_{\tau\tau}=-u(ru\rho''+r\rho u''+2\rho'(ru'+u)+2\rho u'), \\[1em]
2ru\rho^2R_{rr}=-3ru\rho\rho''-r\rho^2u''+3ru(\rho')^2-2\rho\rho'(ru'+u)-2\rho^2u', \\[1em]
2k\rho R_{\theta\theta}=-r^2u\rho''-\rho'(r^ 2u'+4ru)+2\rho'(k-u-ru'),\end{array}\right.\]
as well as
\[\left\{\begin{array}{l}
kr^2\rho^2 G_{\tau\tau}=u(r^2u\rho\rho''-3r^2u(\rho')^2/4+r\rho\rho'(ru'+4u)/2+\rho^2(k-u-ru')),\\[1em]
G_{rr}=kG_{\tau\tau}/u^2-\rho''/\rho+3/2(\rho'/\rho)^2\\[1em]
k\rho^2 G_{\theta\theta}=r(ru\rho\rho''+r\rho^2u''/2-3ru(\rho')^2/4+\rho\rho'(ru'+u)+\rho^2u'),\end{array}\right.\]
and we have $R_{\phi\phi}=\sin^2\theta R_{\theta\theta}$ and $G_{\phi\phi}=\sin^2\theta G_{\theta\theta}$. Thus, if the Einstein equation $G_{\mu\nu}=8\pi T_{\mu\nu}$ holds, then we have
\[0=\frac{2\rho^2}{u^2}\left(u^2(G_{rr}-8\pi T_{rr})-k(G_{\tau\tau}-8\pi T_{\tau\tau})\right)=3(\rho')^2-2\rho\rho''\]
and the function $\rho$ satisfies $(\rho/\rho')'=-1/2$, so that there are constants $\alpha,\beta\in\R$ such that $\rho(r)=(\alpha r+\beta)^{-2}$. The metric now reads
\[{\rm d}s^2=(\alpha r+\beta)^{-2}(u(r){\rm d}\tau^2+ku(r)^{-1}{\rm d}r^2+r^2{\rm d}\Omega^2).\]
We then distinguish two cases:

If $\beta\ne0$, then the field equations reduce to the single equation $G_{rr}=8\pi T_{rr}$, which may be re-written as
\[\beta r^3(\alpha r+\beta)u'(r)+\beta r^2u(r)(\beta-2\alpha r)+(\alpha r+\beta)^2((\alpha^2Q^2-k)r^2+2Q^2\alpha\beta r+Q^2\beta^2)=0,\]
an equation which is equivalent to
\[\left(\frac{ru(r)}{(\alpha r+\beta)^3}\right)'=\frac{(k-\alpha^2Q^2)r^2-2Q^2\alpha\beta r-Q^2\beta^2}{\beta r^2(\alpha r+\beta)^2}\]
and we integrate the right-hand side by decomposing it into simple rational fractions. For $\alpha\ne0$, we obtain the solution
\[u(r)=\frac{(\alpha r+\beta)^3}{r}\left({\gamma}+\frac{1}{\beta}\left(\frac{Q^2}{r}-\frac{k}{\alpha(\alpha r+\beta)}\right)\right),\]
for some additional constant $\gamma\in\R$. If $\alpha=0$, then we get the equation 
\[r^3u'(r)+r^2(u(r)-k)+Q^2\beta^2=0,\]
whose solution is 
\[u(r)=k+\tilde{\alpha}/r+Q^2\beta^2/r^2\]
for some $\tilde{\alpha}\in\R$. Observe now that since $\beta\ne0$, the map $r\mapsto r(\alpha r+\beta)^{-1}$ is a diffeomorphism onto its image for any $\alpha$ and if we let
\[\tilde{\tau}:=\frac{\tau\sqrt{k}}{\beta},~R:=\frac{r}{\alpha r+\beta},~\tilde{Q}:=\frac{Q}{\sqrt{k}},~\tilde{M}:=\left\{\begin{array}{ll} \frac{1}{2}\left(\frac1\alpha+\frac{\alpha Q^2-\beta^2\gamma}{k}\right) & \text{if $\alpha\ne0$},\\[.75em] -\frac{\tilde{\alpha}}{2k\beta} & \text{otherwise,}\end{array}\right.\]
then we obtain the Reissner--Nordstr\"{o}m form of the statement. 

Assume now that $\beta=0$ (forcing $\alpha\ne0$), in which case the equation $G_{rr}=8\pi T_{rr}$ reduces to $k=\alpha^2Q^2$. Moreover, the remaining non-trivial component of the field equation is $G_{\theta\theta}=8\pi T_{\theta\theta}$, which yields
\[r^2u''(r)-2ru'(r)+2(u(r)-\alpha^2 Q^2)=0,\]
whose solution reads $u(r)=\check{\beta}r^2+\check{\gamma}r+\alpha^2 Q^2$ for constants $\tilde{\beta},\tilde{\gamma}\in\R$. Letting
\[\check{\tau}:=\alpha Q\tau,~q:=\frac{\sqrt{\check{\beta}}}{\alpha Q},~m:=-\frac{\check{\gamma}}{2\alpha^2Q^2}\]
we obtain the form
\[{\rm d}s^2=\frac{1}{\alpha^2}\left[\frac{1-2mr+q^2r^2}{r^2}{\rm d}\tilde{\tau}^2+\frac{{\rm d}r^2}{r^2(1-2mr+q^2r^2)}+{\rm d}\Omega^2\right].\]
But in the new coordinates $(\tilde{\tau},r,\theta,\phi)$, the vector potential becomes $A_\mu=-i\alpha^{-1}r^{-1}{\rm d}\tilde{\tau}$, so that we may assume that $\alpha=1/Q$, thus obtaining the stated form of the metric.

To conclude, it remains to compute the Kretschmann scalar $K=R^{\alpha\beta\mu\nu}R_{\alpha\beta\mu\nu}$ of each metric of the statement. To do this, we use the symmetry and give the non-zero components $R_{\alpha\beta\mu\nu}=g_{\alpha\lambda}({\Gamma^\lambda}_{\beta\nu,\mu}-{\Gamma^\lambda}_{\beta\mu,\nu}+{\Gamma^\lambda}_{\sigma\mu}{\Gamma^\sigma}_{\beta\nu}-{\Gamma^\lambda}_{\sigma\nu}{\Gamma^\sigma}_{\beta\mu})$ of the Riemann tensor for indices $\alpha<\beta$ and $\mu<\nu$. First, for the Bertotti--Robinson metric we have
\[R_{\tau r\tau r}=-\frac{Q^2}{r^4},~R_{\theta\phi\theta\phi}=Q^2\sin^2\theta~\Longrightarrow~K=\frac{8}{Q^4}\]
is indeed independent of $r$. For the Reissner--Nordstr\"{o}m form, we have
\[\left\{\begin{array}{ll}
R_{\tau\theta\tau\theta}=\frac{R_{\tau\phi\tau\phi}}{\sin^2\theta}=\frac{(\tilde{Q}^2-\tilde{M}R)(R^2-2\tilde{M}R+\tilde{Q}^2)}{R^4}, & R_{\tau r\tau r}=\frac{2\tilde{M}R-3\tilde{Q}^2}{R^4},\\[1em]
R_{r\theta r\theta}=\frac{R_{r\phi r_\phi}}{\sin^2\theta}=\frac{\tilde{Q}^2-\tilde{M}R}{R^2-2\tilde{M}R+\tilde{Q}^2}, & R_{\theta\phi\theta\phi}=(2\tilde{M}R-\tilde{Q}^2)\sin^2\theta.\end{array}\right.\]
and thus
\[K=\frac{8(6\tilde{M}^2R^2-12\tilde{M}\tilde{Q}^2R+7\tilde{Q}^4)}{R^8},\]
which depends on $R$ and thus the Kretschmann scalar expressed in the original coordinates depends on $r$ as well, as claimed. \qed

\section{Proof of the energy constrain in the horizon-full case (\S \ref{energyy})}\label{proof_energy}
Assume first that $Q^2\le M^2$. For simplicity, we may rescale the radius by $1/M$ and assume that $M=1$. From \eqref{mot_cons}, we have
\[E^2\ge1~\Longleftrightarrow~\Delta(r_0)^2\ge\Delta(r_0)+\Delta(r_0)^{-1}\dot{r}_0^2+r_0^2\dot{\phi}_0^2~\Longleftrightarrow~P(r_0)\le0,\]
where $P=P_{Q^2,\dot{r}_0^2,\dot{\phi}_0^2}\in\R_8[x]$ is the following polynomial of degree at most 8:
\begin{align*}
P(x)&=-x^6\Delta(x)(\Delta(x)^2-\Delta(x)-\dot{r}_0^2\Delta(x)^{-1}-x^2\dot{\phi}_0^2)\\[.5em]
&=\dot{\phi}_0^2x^8-2\dot{\phi}_0^2x^7+(Q^2\dot{\phi}_0^2+\dot{r}_0^2)x^6+2x^5-(8+Q^2)x^4 \\
&\qquad\qquad\qquad +8(1+Q^2)x^3-2Q^2(6+Q^2)x^2+6Q^2x-Q^6.
\end{align*}
We will inductively prove that $P$ is positive on the interval $I=]r_+,+\infty[$ (with $r_+=1+\sqrt{1-Q^2}$), hence contradicting the above inequality and establishing the result. Observe first that $P^{(8)}/8!=\dot{\phi}_0^2\ge0$ so $P^{(7)}$ is non-decreasing on $I$ and since
\[\tfrac{1}{7!}P^{(7)}(r_+)=2\dot{\phi}_0^2(3+4\sqrt{1-Q^2})\ge0\]
we find that $P^{(7)}\ge0$ so that $P^{(6)}$ is non-decreasing on $I$. We repeat the process, with
\[\tfrac{1}{6!}P^{(6)}(r_+)=3\dot{\phi}_0^2\left(5+9(1-Q^2)+14\sqrt{1-Q^2}\right)+\dot{r}_0^2,\]
which is positive since $\gamma$ is non-constant. Thus, we get $P^{(6)}>0$ on $I$ and $P^{(5)}$ is increasing (not only non-decreasing) on $I$. Next, we have
\begin{align*}
\tfrac{1}{5!}P^{(5)}(r_+)=&2\sqrt{1-Q^2}\left[45\dot{\phi}_0^2+25\dot{\phi}_0^2(1-Q^2)+3\dot{r}_0^2\right] \\
&+20\dot{\phi}_0^2+120\dot{\phi}_0^2(1-Q^2)+2(1+3\dot{r}_0^2)>0, \\[1em]
\tfrac{1}{4!}P^{(4)}(r_+)=&5(1+\sqrt{1-Q^2})^2\left[3(Q^2\dot{\phi}_0^2+\dot{r}_0^2)+14\dot{\phi}_0^2(1+\sqrt{1-Q^2})\sqrt{1-Q^2}\right] \\
&+10\sqrt{1-Q^2}+2-Q^2>0, \\[1em]
\tfrac{1}{3!}P^{(3)}(r_+)=&\sqrt{1-Q^2}\left[12\dot{\phi}_0^2(3Q^4+2)+60\dot{r}_0^2+4+4(1-Q^2)(78\dot{\phi}_0^2+5\dot{r}_0^2+1)\right] \\
&+20\dot{r}_0^2+4(1-Q^2)(4+15\dot{r}_0^2)+6\dot{\phi}_0^2(Q^2+(1-Q^2)(56-25Q^2))\ge0,
\end{align*}
and
\begin{align*}
\tfrac12 P''(r_+)=&\sqrt{1-Q^2}\left[60\dot{r}_0^2+4(1-Q^2)(2+15\dot{r}_0^2)+2\dot{\phi}_0^2(33Q^4-136Q^2+112)\right]+15\dot{r}_0^2Q^4 \\
&+\dot{\phi}_0^2(-13Q^6+174Q^4-384Q^2+224)+4(1-Q^2)(30\dot{r}_0^2+(1-Q^2)+1)\ge0, \\[1em]
P'(r_+)=&\sqrt{1-Q^2}\left[6\dot{r}_0^2(Q^4-12Q^2+16)+2\dot{\phi}_0^2(-Q^6+18Q^4-48Q^2+32)\right]\\
&+6\dot{r}_0^2(5Q^4-20Q^2+16)+4\dot{\phi}_0^2(-3Q^6+19Q^4-32Q^2+16)\ge0,
\end{align*}
and thus, $P'$ is increasing on $I$ and has $P'(r_+)\ge0$, so $P$ is increasing on $I$. We conclude by observing that $P(r_+)=\dot{r}_0^2r_+^6\ge0$.

To prove the second statement, assume that $Q^2>M^2$, choose $e>0$ and consider a circular (equatorial) geodesic $\gamma=(\tau,r,\pi/2,\phi)$ with energy $E=e$. Introducing the potential
\[V(r):=\sqrt{\Delta(r)\left(\frac{L^2}{r^2}-1\right)},\]
the condition that $\gamma$ is circular reads
\[V(r)^2+e^2=0=\frac{\partial V}{\partial r}.\]
The first equation imposes 
\[L=\pm r\sqrt{1-\frac{e^2}{\Delta(r)}},\]
while the second yields
\begin{align*}
0=\frac{\partial V}{\partial r}&=\frac{1}{2V}\left(\left(\frac{L^2}{r^2}-1\right)\frac{\partial\Delta}{\partial r}-\frac{2L^2\Delta}{r^3}\right) \\
&\Leftrightarrow(e^2-1)r^4+M(4-3e^2)r^3+2(Q^2(e^2-1)-2M^2)r^2+4MQ^2r-Q^4=0.
\end{align*}
This quartic has a positive root, since it evaluates to $-Q^4<0$ at $r=0$ and is positive when $r\gg0$. \qed

\section{A technical lemma on an implicit differential equation}\label{goubelix}
\begin{lem*}
For any $(\alpha,\beta)\in\R^*\times\R$, their is a unique smooth maximal non-constant solution of the incomplete initial value problem
\[\left\{\begin{array}{rcl}
\dot{y}^2 & = & \alpha(y^2-\beta^2),\\ y(0) & = & \beta.\end{array}\right.\]
Moreover this solution is globally defined and given, for $s\in\R$, by
\[y(s)=\beta\cosh(s\sqrt{\alpha}).\]
\end{lem*}
\begin{proof}
The following argument is due to Olivier Goubet. First, up to exchanging $y\leftrightarrow -y$, we may assume that $\beta\ge0$. Let $y$ be a maximal non-constant solution of the problem, defined on a (maximal, open) interval $I$. Consider the open subset
\[U:=\{s\in I~|~y(s)>\beta\}\subset\R^*\]
and suppose that $]a;b[\subset U$ is a connected component of $U$, of finite length. By continuity of $y$, we have $\beta\le \lim_{s\to a^+}y(s)=y(a)\le\beta$, so that $y(a)=\beta=y(b)$ and by Rolle's theorem, we may choose $a<t<b$ such that $0=y'(t)=\alpha(y(t)^2-\beta^2)$. This leads to $y(t)\le |y(t)|=\beta$: a contradiction.

This proves that $U$ has no connected component of finite length, leaving four possibilities:
\begin{itemize}
\item If $U=\emptyset$, then $y$ is constant.
\item If $U=\R_+^*$, then $y=\beta$ on $\R_-$ and solving the well-defined ODE $\dot{y}=\sqrt{\alpha(y^2-\beta^2)}$ on $U$ leads to 
\[y(s)=\left\{\begin{array}{cc}\beta\cosh(s\sqrt{\alpha}) & \text{if $s>0$}, \\ 1 & \text{otherwise}.\end{array}\right.\]
\item The case $U=\R_-^*$ is the symmetric of the previous one.
\item if $U=\R^*$, then we solve $\dot{y}=\pm\sqrt{\alpha(y^2-\beta^2)}$ on $\R^*_\pm$ and obtain the stated solution, the only one being non-constant and smooth.
\end{itemize}
\end{proof}
\end{appendix}

\printbibliography

@ARTICLE{accioly-ragusa,
  AUTHOR = {Accioly, A. and Ragusa, S.},
  DATE = {2002},
  DOI = {10.1088/0264-9381/19/21/308},
  JOURNALTITLE = {Class. Quantum Grav.},
  NUMBER = {21},
  PAGES = {5429--5434},
  RELATED = {ae61a3c2b7d1b6e100e4278b8dd20a7a},
  RELATEDSTRING = {Corrected in},
  TITLE = {Gravitational deflection of massive particles in classical and semiclassical general relativity},
  VOLUME = {19},
}

@ARTICLE{aksteiner-andersson21,
  AUTHOR = {Aksteiner, S. and Andersson, L.},
  DATE = {2021-12},
  EPRINT = {2112.11863},
  EPRINTCLASS = {gr-qc},
  EPRINTTYPE = {arXiv},
  TITLE = {Gravitational Instantons and special geometry},
}

@ARTICLE{anderson90,
  AUTHOR = {Anderson, M. T.},
  PUBLISHER = {Lehigh University},
  DATE = {1990},
  DOI = {10.4310/jdg/1214444097},
  JOURNALTITLE = {J. Differential Geometry},
  NUMBER = {1},
  PAGES = {265--275},
  TITLE = {Short geodesics and gravitational instantons},
  VOLUME = {31},
}

@ARTICLE{atiyah-franchetti-schroers14,
  AUTHOR = {Atiyah, M. F. and Franchetti, G. and Schroers, B. J.},
  DATE = {2015},
  DOI = {10.1007/JHEP02(2015)062},
  JOURNALTITLE = {J. High Energ. Phys.},
  NUMBER = {62},
  TITLE = {Time evolution in a geometric model of a particle},
  VOLUME = {2015},
}

@ARTICLE{atiyah-manton-schroers12,
  AUTHOR = {Atiyah, M. F. and Manton, N. S. and Schroers, B. J.},
  DATE = {2012},
  DOI = {10.1098/rspa.2011.0616},
  JOURNALTITLE = {Proc. R. Soc. A},
  PAGES = {1252--1279},
  TITLE = {Geometric models of matter},
  VOLUME = {468},
}

@ARTICLE{battista-esposito22,
  AUTHOR = {Battista, E. and Esposito, G.},
  DATE = {2022},
  DOI = {10.1140/epjc/s10052-022-11070-w},
  ISSUE = {1088},
  JOURNALTITLE = {Eur. Phys. J. C},
  TITLE = {Geodesic motion in Euclidean Schwarzschild geometry},
  VOLUME = {82},
}

@ARTICLE{belavin-et-al75,
  AUTHOR = {Belavin, A. A. and Polyakov, A. M. and Schwartz, A. S. and Tyupkin, Yu. S.},
  DATE = {1975},
  DOI = {10.1016/0370-2693(75)90163-X},
  JOURNALTITLE = {Phys. Lett. B},
  NUMBER = {1},
  PAGES = {85--87},
  TITLE = {Pseudoparticle solutions of the Yang--Mills equations},
  VOLUME = {59},
}

@ARTICLE{beloborodov02,
  AUTHOR = {Beloborodov, A. M.},
  DATE = {2002},
  DOI = {10.1086/339511},
  JOURNALTITLE = {ApJ},
  NUMBER = {2},
  TITLE = {Gravitational bending of light near compact objects},
  VOLUME = {566},
}

@ARTICLE{bertotti59,
  AUTHOR = {Bertotti, B.},
  DATE = {1959},
  DOI = {10.1103/PhysRev.116.1331},
  JOURNALTITLE = {Phys. Rev.},
  PAGES = {1331--1333},
  TITLE = {Uniform electromagnetic field in the theory of general relativity},
  VOLUME = {116},
}

@MISC{briet-hobil,
  AUTHOR = {Briët, J. and Hobill, D.},
  DATE = {2008},
  EPRINT = {0801.3859},
  EPRINTTYPE = {arXiv},
  TITLE = {Determining the dimensionality of spacetime by gravitational lensing},
}

@ARTICLE{carlson,
  AUTHOR = {Carlson, B. C.},
  PUBLISHER = {Springer},
  DATE = {1995},
  DOI = {10.1007/BF02198293},
  JOURNALTITLE = {Numerical Algorithms},
  NUMBER = {1},
  PAGES = {13--26},
  TITLE = {Numerical computation of real or complex elliptic integrals},
  VOLUME = {10},
}

@ARTICLE{GRay,
  AUTHOR = {Chan, C. and Psaltis, D. and Özel, F.},
  PUBLISHER = {The Americal Astronomical Society},
  DATE = {2013},
  DOI = {10.1088/0004-637X/777/1/13},
  JOURNALTITLE = {ApJ},
  NUMBER = {1},
  TITLE = {\texttt{GRay}: A massively parallel {G}{P}{U}-based code for ray tracing in relativistic spacetimes},
  VOLUME = {777},
}

@ARTICLE{chen-teo11,
  AUTHOR = {Chen, Y. and Teo, E.},
  DATE = {2011},
  DOI = {10.1016/j.physletb.2011.07.076},
  JOURNALTITLE = {Phys. Lett. B},
  NUMBER = {3},
  PAGES = {359--362},
  TITLE = {A new AF gravitational instanton},
  VOLUME = {703},
}

@ARTICLE{cieslik-mach22,
  AUTHOR = {Cieślik, A. and Mach, P.},
  DATE = {2022},
  DOI = {10.1088/1361-6382/ac95f2},
  JOURNALTITLE = {Class. Quantum Grav.},
  NUMBER = {22},
  TITLE = {Revisiting timelike and null geodesics in the Schwarzschild spacetime: general expressions in terms of Weierstrass elliptic functions},
  VOLUME = {39},
}

@ARTICLE{coquereaux,
  AUTHOR = {Coquereaux, R. and Grossmann, A. and Lautrup, B. E.},
  PUBLISHER = {Oxford University Press},
  DATE = {1990},
  JOURNALTITLE = {IMA Journal of Numerical Analysis},
  PAGES = {119--128},
  TITLE = {Iterative method for calculation of the {W}eierstrass elliptic function},
  VOLUME = {10},
}

@ARTICLE{darwin_angle,
  AUTHOR = {Darwin, C.},
  DATE = {1959},
  DOI = {10.1098/rspa.1959.0015},
  JOURNALTITLE = {Proc. R. Soc. Lond. A},
  PAGES = {180--194},
  TITLE = {The gravity field of a particle},
  VOLUME = {249},
}

@ARTICLE{agol-dexter,
  AUTHOR = {Dexter, J. and Agol, E.},
  PUBLISHER = {The American Astronomical Society},
  DATE = {2009},
  DOI = {10.1088/0004-637X/696/2/1616},
  JOURNALTITLE = {ApJ},
  NUMBER = {2},
  TITLE = {A fast new public code for computing photon orbits in a {K}err spacetime},
  VOLUME = {696},
}

@ARTICLE{dunajski-tod22,
  AUTHOR = {Dunajski, M. and Tod, P.},
  DATE = {2022},
  DOI = {10.1017/S0305004121000463},
  JOURNALTITLE = {Math. Proc. Camb. Phil. Soc.},
  NUMBER = {1},
  PAGES = {123--154},
  TITLE = {Conformal geodesics on gravitational instantons},
  VOLUME = {173},
}

@ARTICLE{eguchi-gilkey-hanson80,
  AUTHOR = {Eguchi, T. and Gilkey, P. B. and Hanson, A. J.},
  DATE = {1980},
  DOI = {10.1016/0370-1573(80)90130-1},
  JOURNALTITLE = {Phys. Rep.},
  NUMBER = {6},
  PAGES = {213--393},
  TITLE = {Gravitation, gauge theories and differential geometry},
  VOLUME = {66},
}

@ARTICLE{elster84,
  AUTHOR = {Elster, T.},
  DATE = {1984},
  DOI = {10.1088/0264-9381/1/1/007},
  ISSUE = {1},
  JOURNALTITLE = {Class. Quantum Grav.},
  TITLE = {Quantum vacuum energy near a black hole: the Maxwell field},
  VOLUME = {1},
}

@BOOK{esposito,
  AUTHOR = {Esposito, G.},
  PUBLISHER = {Springer},
  DATE = {1992},
  DOI = {10.1007/978-3-662-14495-4},
  SERIES = {Lecture notes in Physics Monographs},
  TITLE = {Quantum gravity, quantum cosmology and Lorentzian geometries},
}

@ARTICLE{garnier_CQG23,
  AUTHOR = {Garnier, A.},
  PUBLISHER = {IOP Publishing},
  DATE = {2023},
  DOI = {10.1088/1361-6382/accbfe},
  JOURNALTITLE = {Class. Quantum Grav.},
  NUMBER = {13},
  PAGES = {135011},
  TITLE = {Motion equations in a {K}err--{N}ewman--de {S}itter spacetime: some methods of integration and application to black holes shadowing in {S}cilab},
  VOLUME = {40},
}

@ARTICLE{gergely-darazs,
  AUTHOR = {Gergely, L. A. and Darázs, B.},
  DATE = {2006},
  DOI = {10.48550/arXiv.astro-ph/0602427},
  JOURNALTITLE = {Publ. Astron. Dep. Eotvos Univ.},
  PAGES = {213--219},
  TITLE = {Weak gravitational lensing in brane-worlds},
  VOLUME = {17},
}

@INBOOK{gibbons_centenary,
  AUTHOR = {Gibbons, G. W.},
  EDITOR = {Hawking, S. W. and Israel, W.},
  PUBLISHER = {Cambridge University Press},
  BOOKTITLE = {General relativity - {A}n {E}instein centenary survey},
  DATE = {1979},
  TITLE = {Quantum field theory in curved spacetime},
}

@ARTICLE{gibbons-hawking77,
  AUTHOR = {Gibbons, G. W. and Hawking, S. W.},
  PUBLISHER = {American Physical Society},
  DATE = {1977},
  DOI = {10.1103/PhysRevD.15.2752},
  ISSUE = {10},
  JOURNALTITLE = {Phys. Rev. D},
  TITLE = {Action integrals and partition functions in quantum gravity},
  VOLUME = {15},
}

@ARTICLE{gibbons-hawking79,
  AUTHOR = {Gibbons, G. W. and Hawking, S. W.},
  DATE = {1979},
  DOI = {10.1007/BF01197189},
  JOURNALTITLE = {Commun. Math. Phys.},
  PAGES = {291--310},
  TITLE = {Classification of gravitational instanton symmetries},
  VOLUME = {66},
}

@ARTICLE{gibbons-vyska,
  AUTHOR = {Gibbons, G. W. and Vyska, M.},
  PUBLISHER = {{IOP} Publishing},
  DATE = {2012},
  DOI = {10.1088/0264-9381/29/6/065016},
  JOURNALTITLE = {Class. Quantum Grav.},
  NUMBER = {6},
  TITLE = {The application of {W}eierstrass elliptic functions to {S}chwarzschild null geodesics},
  VOLUME = {29},
}

@ARTICLE{hagihara,
  AUTHOR = {Hagihara, Y.},
  DATE = {1930},
  JOURNALTITLE = {Japanese Journal of Astronomy and Geophysics},
  PAGES = {67--176},
  TITLE = {{Theory of the relativistic trajectories in a gravitational field of {S}chwarzschild}},
  VOLUME = {8},
}

@ARTICLE{hartle-hawking,
  AUTHOR = {Hartle, J.  B. and Hawking, S. W.},
  PUBLISHER = {American Physical Society},
  DATE = {1976},
  DOI = {10.1103/PhysRevD.13.2188},
  ISSUE = {8},
  JOURNALTITLE = {Phys. Rev. D},
  PAGES = {2188--2203},
  TITLE = {Path-integral derivation of black-hole radiance},
  VOLUME = {13},
}

@INBOOK{hawking1979,
  AUTHOR = {Hawking, S. W.},
  EDITOR = {Lévy, M. and Deser, S.},
  PUBLISHER = {Springer},
  BOOKTITLE = {Recent Developments in Gravitation: {C}argèse 1978},
  DATE = {1979},
  DOI = {10.1007/978-1-4613-2955-8_4},
  PAGES = {145--173},
  TITLE = {Euclidean Quantum Gravity},
}

@ARTICLE{hawking77,
  AUTHOR = {Hawking, S. W.},
  DATE = {1977},
  DOI = {10.1016/0375-9601(77)90386-3},
  ISSUE = {2},
  JOURNALTITLE = {Phys. Lett. A},
  PAGES = {81--83},
  TITLE = {Gravitational instantons},
  VOLUME = {60},
}

@ARTICLE{he-lin16,
  AUTHOR = {He, G. and Lin, W.},
  DATE = {2016},
  DOI = {10.1088/0264-9381/33/9/095007},
  JOURNALTITLE = {Class. Quantum Grav.},
  NUMBER = {9},
  TITLE = {Gravitational deflection of light and massive particles by a moving Kerr--Newman black hole},
  VOLUME = {33},
}

@ARTICLE{he_et_al,
  AUTHOR = {He, G. and Zhou, X. and Feng, Z. and Mu, X. and Wang, H. and Li, W. and Pan, C. and Lin, W.},
  DATE = {2020},
  DOI = {10.1140/epjc/s10052-020-8382-z},
  JOURNALTITLE = {Eur. Phys. J. C},
  NUMBER = {835},
  TITLE = {Gravitational deflection of massive particles in {S}chwarzschild--de {S}itter spacetime},
  VOLUME = {80},
}

@ARTICLE{hoffmann32,
  AUTHOR = {Hoffmann, B.},
  DATE = {1932},
  DOI = {10.1093/qmath/os-3.1.226},
  JOURNALTITLE = {The Quarterly Journal of Mathematics},
  NUMBER = {1},
  PAGES = {226--237},
  TITLE = {On the spherically symmetric field in relativity},
  VOLUME = {os-3},
}

@THESIS{jante15,
  AUTHOR = {Jante, R.},
  INSTITUTION = {Heriot--Watt University},
  URL = {http://hdl.handle.net/10399/3083},
  DATE = {2015},
  TITLE = {On the spectrum of some gravitational instantons},
  TYPE = {phdthesis},
}

@ARTICLE{kunduri-lucietti22,
  AUTHOR = {Kunduri, H. K. and Lucietti, J.},
  DATE = {2021},
  DOI = {10.1007/s11005-021-01475-1},
  JOURNALTITLE = {Lett. Math. Phys.},
  NUMBER = {133},
  TITLE = {Existence and uniqueness of asymptotically flat toric gravitational instantons},
  VOLUME = {111},
}

@ARTICLE{perez-giz-levin,
  AUTHOR = {Levin, J. and Perez-Giz, G.},
  PUBLISHER = {American Physical Society},
  DATE = {2008},
  DOI = {10.1103/PhysRevD.77.103005},
  ISSUE = {10},
  JOURNALTITLE = {Physical Review D},
  TITLE = {A periodic table of black hole orbits},
  VOLUME = {77},
}

@ARTICLE{li-zhou-li-he,
  AUTHOR = {Li, Z. and Zhou, X. and Li, W. and He, G.},
  DATE = {2019},
  DOI = {10.1088/0253-6102/71/10/1219},
  JOURNALTITLE = {Commun. Theor. Phys.},
  NUMBER = {10},
  PAGES = {1219--1226},
  TITLE = {Gravitational deflection of massive particles by a {S}chwarzschild black hole in radiation gauge},
  VOLUME = {71},
}

@ARTICLE{lindberg-rayan18,
  AUTHOR = {Lindberg, A. and Rayan, S.},
  DATE = {2018},
  DOI = {10.1016/j.geomphys.2018.05.018},
  JOURNALTITLE = {J. Geom. Phys.},
  PAGES = {114--130},
  TITLE = {Geodesics on a Kerr–Newman–(anti-)de Sitter instanton},
  VOLUME = {132},
}

@ARTICLE{mars99,
  AUTHOR = {Mars, M. and Simon, W.},
  DATE = {1999},
  DOI = {10.1016/S0393-0440(99)00023-6},
  JOURNALTITLE = {J. Geom. Phys.},
  NUMBER = {2},
  PAGES = {211--226},
  TITLE = {A proof of uniqueness of the Taub--bolt instanton},
  VOLUME = {32},
}

@BOOK{mcmahon-snyder,
  AUTHOR = {Mc{M}ahon, J. and Snyder, V.},
  PUBLISHER = {American Book Company, New York},
  DATE = {1898},
  SERIES = {The {C}ornell {M}athematical {S}eries},
  TITLE = {Elements of the differential calculus},
}

@ARTICLE{mellor-moss89,
  AUTHOR = {Mellor, F. and Moss, I.},
  PUBLISHER = {IOP Publishing},
  DATE = {1989},
  DOI = {10.1088/0264-9381/6/10/008},
  ISSUE = {10},
  JOURNALTITLE = {Class. Quantum Grav.},
  PAGES = {1379--1385},
  TITLE = {Black holes and gravitational instantons},
  VOLUME = {6},
}

@ARTICLE{monteiro_santos,
  AUTHOR = {Monteiro, R. and Santos, J. E.},
  PUBLISHER = {American Physical Society},
  DATE = {2009},
  DOI = {10.1103/PhysRevD.79.064006},
  ISSUE = {6},
  JOURNALTITLE = {Phys. Rev. D},
  TITLE = {Negative modes and the thermodynamics of {R}eissner--{N}ordström black holes},
  VOLUME = {79},
}

@ARTICLE{mosna-tavares,
  AUTHOR = {Mosna, R. and Tavares, G.},
  PUBLISHER = {APS},
  DATE = {2009},
  DOI = {10.1103/PhysRevD.80.105006},
  ISSUE = {10},
  JOURNALTITLE = {Phys. Rev. D},
  PAGES = {105006},
  TITLE = {New self-dual solutions of $SU(2)$ Yang--Mills theory in Euclidean Schwarzschild space},
  VOLUME = {80},
}

@ARTICLE{elnaschie04,
  AUTHOR = {Naschie, M. S. El},
  DATE = {2004},
  DOI = {10.1016/j.chaos.2003.12.001},
  JOURNALTITLE = {Chaos, Solitons \& Fractals},
  NUMBER = {1},
  PAGES = {249--260},
  TITLE = {How gravitational instanton could solve the mass problem of the standard model of high energy particle physics},
  VOLUME = {21},
}

@ARTICLE{oh-park-yang11,
  AUTHOR = {Oh, J. J. and Park, C. and Yang, H. S.},
  DATE = {2011},
  DOI = {10.1007/JHEP04(2011)087},
  JOURNALTITLE = {J. High Energ. Phys.},
  NUMBER = {87},
  TITLE = {Yang--Mills instantons from gravitational instantons},
  VOLUME = {2011},
}

@ARTICLE{page78,
  AUTHOR = {Page, D.},
  DATE = {1978},
  DOI = {10.1016/0370-2693(78)90016-3},
  JOURNALTITLE = {Phys. Lett. B},
  NUMBER = {2},
  PAGES = {249--251},
  TITLE = {Taub--NUT instanton with an horizon},
  VOLUME = {78},
}

@ARTICLE{pang-jia,
  AUTHOR = {Pang, X. and Jia, J.},
  DATE = {2019},
  DOI = {10.1088/1361-6382/ab0512},
  JOURNALTITLE = {Class. Quantum Grav.},
  NUMBER = {6},
  TITLE = {Gravitational lensing of massive particles in Reissner--Nordström black hole spacetime},
  VOLUME = {36},
}

@ARTICLE{PYYY,
  AUTHOR = {Pu, H. and Yun, K. and Younsi, Z. and Yoon, S.},
  PUBLISHER = {The American Astronomical Society},
  DATE = {2016},
  DOI = {10.3847/0004-637X/820/2/105},
  JOURNALTITLE = {ApJ},
  NUMBER = {2},
  PAGES = {105--116},
  TITLE = {\texttt{Odyssey}: A public {G}{P}{U}-based code for general-relativistic radiative transfer in {K}err spacetime},
  VOLUME = {820},
}

@ARTICLE{robinson59,
  AUTHOR = {Robinson, I.},
  DATE = {1959},
  JOURNALTITLE = {Bull. Acad. Polon. Sci},
  PAGES = {351--352},
  TITLE = {A solution of the {M}axwell-{E}instein equations},
  VOLUME = {7},
}

@BOOK{straumann,
  AUTHOR = {Straumann, N.},
  PUBLISHER = {Springer},
  DATE = {2013},
  EDITION = {Second edition},
  SERIES = {Graduate texts in Physics},
  TITLE = {General Relativity},
}

@ARTICLE{tekin02,
  AUTHOR = {Tekin, B.},
  PUBLISHER = {APS},
  DATE = {2002},
  DOI = {10.1103/PhysRevD.65.084035},
  ISSUE = {8},
  JOURNALTITLE = {Phys. Rev. D},
  PAGES = {084035},
  TITLE = {Yang--Mills solutions on Euclidean Schwarzschild space},
  VOLUME = {65},
}

@ARTICLE{osiris,
  AUTHOR = {Velásquez-Cadavid, J. M. and Arrieta-Villamizar, J. A. and Lora-Clavijo, F. D. and Pimentel, O. M. and Osorio-Vargas, J. E.},
  DATE = {2022},
  DOI = {10.1140/epjc/s10052-022-10054-0},
  JOURNALTITLE = {Eur. Phys. J. C},
  NUMBER = {103},
  TITLE = {\texttt{OSIRIS}: a new code for ray tracing around compact objects},
  VOLUME = {82},
}

@ARTICLE{gyoto,
  AUTHOR = {Vincent, F. H. and Paumard, T. and Gourgoulhon, E. and Perrin, G.},
  PUBLISHER = {IOP Publishing},
  DATE = {2011},
  DOI = {10.1088/0264-9381/28/22/225011},
  JOURNALTITLE = {Class. Quantum Grav.},
  NUMBER = {22},
  TITLE = {\texttt{GYOTO}: a new general relativistic ray-tracing code},
  VOLUME = {28},
}

@ARTICLE{visinescu93,
  AUTHOR = {Visinescu, M.},
  DATE = {1993},
  DOI = {10.1007/BF01474631},
  JOURNALTITLE = {Z. Phys. C - Particles and Fields},
  PAGES = {337--341},
  TITLE = {The geodesic motion on generalized Taub--NUT gravitational instantons},
  VOLUME = {60},
}

@ARTICLE{viththani24,
  AUTHOR = {Viththani, D. and Joshi, A. and Bhanja, T. and Joshi, P.},
  DATE = {2024},
  EPRINT = {2402.02069},
  EPRINTCLASS = {gr-qc},
  EPRINTTYPE = {arXiv},
  TITLE = {Particle motion and tidal force in a non-vacuum-charged naked singularity},
}

@ARTICLE{yang-zhang23,
  AUTHOR = {Yang, Y. and Zhang, X.},
  DATE = {2023},
  DOI = {10.1140/epjc/s10052-023-11762-x},
  JOURNALTITLE = {Eur. Phys. J. C},
  NUMBER = {574},
  TITLE = {Geodesics on metrics of Eguchi--Hanson type},
  VOLUME = {83},
}

\end{document}